%% file: main.tex
\documentclass[11pt]{article}
\pdfoutput=1
\input{definitions}
\begin{document}

\title{\LARGE  Exact and practical pattern matching\\for quantum circuit optimization}

\author[1,2]{Raban Iten\thanks{These authors contributed equally to this work.}\thanks{Email: \href{mailto:itenr@itp.phys.ethz.ch}{itenr@itp.phys.ethz.ch}}}
\author[1]{Romain Moyard$^*$\thanks{Email: \href{mailto:rmoyard@ethz.ch}{rmoyard@ethz.ch}}}
    \author[1]{Tony Metger\thanks{Email: \href{mailto:tmetger@ethz.ch}{tmetger@ethz.ch}}}
\author[1,2]{David Sutter\thanks{Email: \href{mailto:dsu@zurich.ibm.com}{dsu@zurich.ibm.com}}}
\author[2]{Stefan Woerner\thanks{Email: \href{mailto:wor@zurich.ibm.com}{wor@zurich.ibm.com}}}
\affil[1]{Institute for Theoretical Physics, ETH Zurich, Switzerland}
\affil[2]{IBM Quantum, IBM Research -- Zurich, Switzerland}

\date{}

\maketitle
  
\begin{abstract}
Quantum computations are typically compiled into a circuit of basic quantum gates.
Just like for classical circuits, a quantum compiler should optimize the quantum circuit, e.g.~by minimizing the number of required gates.
Optimizing quantum circuits is not only relevant for improving the runtime of quantum algorithms in the long term, but is also particularly important for near-term quantum devices that can only implement a small number of quantum gates before noise renders the computation useless.
An important building block for many quantum circuit optimization techniques is pattern matching, where given a large and a small quantum circuit, we are interested in finding all maximal matches of the small circuit, called pattern, in the large circuit, considering pairwise commutation of quantum gates.

In this work, we present a classical algorithm for pattern matching that provably finds all maximal matches in time polynomial in the circuit size (for a fixed pattern size). Our algorithm works for both quantum and reversible classical circuits.
We demonstrate numerically that our algorithm, implemented in the open-source library Qiskit, scales considerably better than suggested by the theoretical worst-case complexity and is practical to use for circuit sizes typical for near-term quantum devices.
Using our pattern matching algorithm as the basis for known circuit optimization techniques such as template matching and peephole optimization, we demonstrate a significant ($\sim$30\%) reduction in gate count for random quantum circuits, and are able to further improve practically relevant quantum circuits that were already optimized with state-of-the-art techniques.

\end{abstract}

\clearpage

\section{Introduction}

\input{introduction}

%\appendices

\section{Preliminaries}
Before giving a formal description of our algorithm in \cref{sec:temp_match}, in this section we introduce some necessary notation and give details on the canonical form of quantum circuits (introduced in \cref{sec:algo_intro}). Additionally, we provide some intuition regarding possible difficulties one may encounter in pattern matching for quantum circuits in \cref{sec:possible_difficulties}.

\subsection{Notation} \label{sec:notation}
We write a circuit $C$ as a gate list  $C=(C_1,\dots,C_{|C|})$, where we assume that the gates are indexed from left to right (with the order of commuting gates chosen arbitrarily).
The unitary performed by the circuit is given by $U=U_{|C|}U_{|C|-1}\dots U_1$, where $U_i$ denotes the unitary corresponding to the gate $C_i$. 
A gate can be any description of a unitary operation together with an ordered list of qubits it acts on, e.g., the gate \cnot{}$(1,4)$ represents a \cnot{}{} gate controlled on the qubit with label 1 and acting on the qubit with label 4. 
If two gates $A$ and $B$ perform the same operation, albeit on qubits with different labels, we write $A \cong B$, e.g., \cnot{}$(1,4)\cong\, $\cnot{}$(2,1)$.
 If the unitaries that represent the circuits $C$ and $D$ are equal up to a global phase shift (which is a scalar multiplication of the matrix with $\mathrm{e}^{\mathrm{i} x}$ for $x\in \mathbb{R}$), we say that the two circuits are represented by the same operator, and we write $C \simeq D$.
The concatenation of two circuits $C=(C_1,\dots,C_{|C|})$ and $D=(D_1,\dots,D_{|D|})$ is denoted by $(C,D) = (C_1,\dots,C_{|C|},D_1,\dots,D_{|D|})$.\footnote{It should always be clear for the context if we mean a tuple of two circuits or the concatenation of them.}
We denote the commutator of the unitaries corresponding to two gates $A$ and $B$ by $\comu{A}{B}$. Moreover, we write $\comuCirc{i}{j}{C}=0$ if and only if $i=j$ or if we can pairwise commute gates in the circuit $C$ such that the order of the gates $C_i$ and $C_j$ is interchanged, i.e., if $i<j$ the gate $C_j$ can be moved before the gate $C_i$ and vice versa for the case $j<i$.
The set $S_n=\{(1,2,3,\dots,n),(2,1,3,\dots,n),\dots \}$ denotes the symmetric group of order $n$, i.e., the set of all possible permutations of $(1,2,\dots,n)$.

The canonical form of quantum circuits introduced in \cref{sec:algo_intro} represents circuits as directed acyclic graphs (DAGs). For a DAG, we denote the set of successors of a vertex $v_i$ in such a graph $G$ by \texttt{Succ}$(v_i,G)$, i.e., \texttt{Succ}$(v_i,G)$ contains all the vertices $v_j$ for which there is a (forward directed) path from vertex $v_i$ to vertex $v_j$. On the other hand, we denote the set of predecessors of a vertex $v_i$ in a graph $G$ by \texttt{Pred}$(v_i,G)$, i.e., the set  \texttt{Pred}$(v_i,G)$ contains all the vertices $v_j$ such that there is a path from vertex $v_j$ to vertex $v_i$. The direct successors and predecessors are the ones that are connected through only one edge to the considered vertex $v_i$, and we call these sets \texttt{DirectSucc}$(v_i,G)$ and \texttt{DirectPred}$(v_i,G)$, respectively.

%%%%%%%%%%%%%%%%%%%%%%%%%%%%%%%%%%
\subsection{Canonical form for quantum circuits}\label{sec:canonical_form}

The canonical form of a quantum circuit was introduced in \cref{sec:algo_intro}.
An algorithm that constructs the canonical representation of any quantum circuit $C$ with time complexity $\mathcal{O}\left(|C|^2\right)$ was given in~\cite{Rahman:2014:AQT:2711453.2629537} and we describe it in Algorithm~\ref{algo:CreateCanonicalForm} for completeness. 
Furthermore, we prove some basic properties of the canonical form in \cref{lem:Independence_CanonicalForm} and \cref{lem:Property_CanonicalForm}.

\begin{algorithm}[H] 
  \caption{ \texttt{CreateCanonicalForm}: Creates the canonical form of a quantum circuit}
  \label{algo:CreateCanonicalForm}
   \begin{algorithmic}[1]
   \State Input: Quantum circuit $C$  with $|C|$ gates   
   \State{Initialize an empty directed acyclic graph $G$}
   \For{$j \in \{1,2,\dots,|C|\}$}
    \State{Set the attribute \textit{isReachable} to true for all vertices in $G$}
   \State{Add a vertex with label $j$ to the graph $G$}
   \For{$i \in \{j-1,j-2,\dots,1\}$}
   \If{$G_i$.\textit{isReachable}  and $\comu{C_i}{C_j}\neq 0$} \label{algo:create_canonical_if}
   \State{Add a (directed) edge from vertex $G_i$ to vertex $G_j$ in $G$}
   \For{\textit{preDec}$\in$\texttt{Pred}($G_i,G)$}
   \State{\textit{preDec}.$\mathit{isReachable}\leftarrow \mathit{false}$}
    \EndFor
   \EndIf
   \EndFor
   \EndFor
\State Output: the canonical form $G$ of the circuit $C$ 
   \end{algorithmic}
\end{algorithm}

\begin{rmk} \label{rmk:constant_access_to_successors}
%The graph is stored in adjacency-list format, where for every vertex we store a list of all its predecessors and a list of all its successors.
%If we account worst-case time complexity $\mathcal{O}(|C|^3)$ to create the canonical form $G^C$ of a quantum circuit $C$, 
We can store the graph in an adjacency-list format, where for every vertex we store a list of all its predecessors and a list of all its successors. This is useful because we get constant time access to the list of all possible successors (or predecessors) of a vertex in $G^C$. Moreover, we can sort the vertices in these lists by their labels, as is done in Algorithm~\ref{algo:InitializeSuccessors}. To estimate the worst-case complexity of \cref{algo:InitializeSuccessors}, note that every vertex $v$ can have at most $C$ direct successor $v_i$, each of which can have at most $|C|$ successors. 
Merging $|C|$ ordered lists $(v_i).\textit{Successors}$, each of length at most $|C|$, in a pairwise manner has complexity $\mathcal{O}(|C|^2)$.
Repeating this for the $C$ vertices $v$ in $G^C$, we get an overall time complexity of $\mathcal{O}(|C|^3)$.
%To see this, it is enough to notice that the merging process of the ordered lists in Algorithm~\ref{algo:InitializeSuccessors} has worst-case time complexity $\mathcal{O}(|C|^2)$. Indeed, consider $k$ ordered lists, each of which containing at most $n$ numbers between 1 and $n$. Merging two of the ordered lists of length $n$ can be done in time $\mathcal{O}(n)$. Since we remove duplicates during the merging process, we end up again with a list of length at most $n$ (because there are only $n$ different numbers). Therefore, merging all pairs of the $k$ lists has time complexity $\mathcal{O}(kn)$, and we end up with $\lceil k/2 \rceil$ lists of length at most $n$. Going on recursively like this, gives a time complexity of $\sum_{i=0}^{\lceil \log_2(k)\rceil-1} \frac{1}{2^{i}} \mathcal{O}(kn)=  \mathcal{O}(kn)$.  In our case, we have at most $|G|$ direct successors $v_i$ for each vertex and each list $(v_i).\textit{Successors}$ can contain at most $|G|$ entries, i.e., $k=n=|G|=|C|$.
\end{rmk}

\begin{algorithm}[H] 
  \caption{ \texttt{InitializeSuccessors}: Sets an attribute \textit{Successors} of each vertex equal to the list of its successors (ordered according to their labels)}
  \label{algo:InitializeSuccessors}
   \begin{algorithmic}[1]
   \State Input: Canonical form $G$  of a quantum circuit
   \State{Initialize an attribute  $\mathit{Successors}\leftarrow\mathit{null}$  for all vertices in $G$}
   \For{$i \in \{|G|,|G|-1,\dots,1\}$}
   \State{ $(G_i).\mathit{Successors} \leftarrow$ (Sorted) merge of $\texttt{DirectSucc}(G_i,G) \cup \{ v.\mathit{Successors} \; | \; v \in \texttt{DirectSucc}(G_i,G) \}$}
   \EndFor
      \end{algorithmic}
\end{algorithm}

%\begin{algorithm}[H] 
%  \caption{ \texttt{FindSuccessorsWithRoot}: Sets an attribute \textit{Successors} for the input vertex and for all of its successors}
%  \label{algo:FindSuccessorsWithRoot}
%   \begin{algorithmic}[1]
%   \State Input: Canonical form $G$ of a quantum circuit   
%   \State{Initialize an attribute  \textit{Successors}$\leftarrow \mathit{null}$  for all vertices in $G$}
%   \State{Initialize a stack\footnotemark[6] for vertices \textit{VertexStack}}
%   \State{Initialize a list $\mathit{TerminalVertices}$ with all vertices of $G$ that do not have an outgoing edge}
%   \While{$\mathit{TerminalVertices}$ is not empty}
%   \State{$\mathit{VertexStack} \leftarrow (\mathit{TerminalVertices})_1$}
%   \State{Remove the first element of $\mathit{TerminalVertices}$}
%   \While{$\mathit{VertexStack}$ is not empty}
%   \State{$v \leftarrow \mathit{VertexStack}.\texttt{pop}$}
%   \For{$p \in \texttt{DirectPred}(v)$}
%   \EndFor
%   \EndWhile
%   \EndWhile
%   \State Output: the canonical form $G$ of the circuit $C$ 
%   \end{algorithmic}
%\end{algorithm}
%\footnotetext[6]{The method \textit{VertexStack}.\texttt{Pop}  gives back the lastly added element from \textit{VertexStack} and removes it from the stack. The method  \textit{VertexStack}.\texttt{Push}$(v)$ adds a vertex $v$ at the top of the stack.}
%\addtocounter{footnote}{+1} 

In the following, we will show some properties of the canonical form, which will turn out useful for proving the correctness of our matching algorithm. 

\begin{lemma}[Independence on pairwise gate commutations] \label{lem:Independence_CanonicalForm}
Consider two circuit $C \simeq C'$, where $C'$ is obtained from $C$ by swapping commuting gates, but keeping the gate labels from circuit $C$. Then, the canonical form of $C$ and $C'$ are given by the same directed acyclic graph.
\end{lemma}
\begin{proof}
In the construction in  Algorithm~\ref{algo:CreateCanonicalForm}, we only add edges for non-commuting gates.
The order of these non-commuting gates must be the same in $C$ and $C'$, so one can see that one ends up with the same canonical form (however, the vertices might have been added to the graph in different orders).
Let us prove this in more detail. Assume that $i$ is the smallest index such that $C'_i \neq C_i$. Then, there is a gate $C_j$ with $j>i$ such that $C'_i=C_j$ and $C_j$ can be commuted to the position $i$ in $C$ (without moving the gates $C_1,\dots,C_i$). In other words, we have $[C_j,C_k]=0$ for $k \in \{i,\dots,j-1\}$. We have to show that when adding the vertex $G^{C}_j$ to the canonical form in Algorithm~\ref{algo:CreateCanonicalForm}, we add the same edges from $G
^C_1,\dots,G^C_{i-1}$ to $G^C_j$ as the edges in the canonical form $G^{C'}$ from $G^{C'}_1,\dots,G^{C'}_{i-1}$ to $G^{C'}_i$ and that there are no edges added between $G
^C_i,\dots,G^C_{j-1}$ and  $G^{C}_j$. This is evident from line~\ref{algo:create_canonical_if} in the algorithm, since $[C_j,C_k]=0$ for $k \in \{i,\dots,j\}$. The proof can then be completed by induction.
\end{proof}

\begin{lemma} [Necessary and sufficient condition for interchanging gates]\label{lem:Property_CanonicalForm} 
Given a circuit $C$ with a canonical form $G$ and two indices $i<j$, then the following two statements are equivalent:
\begin{enumerate}
\item $\comuCirc{i}{j}{C}=0\, ,$
\item $G_j \notin \textnormal{\texttt{Succ}}(G_i,G)\, .$
\end{enumerate}
\end{lemma}
\begin{proof}
From the construction in Algorithm~\ref{algo:CreateCanonicalForm} it is clear that if $G_j \in $ \texttt{Succ}$(G_i,G)$, we have that the gate $C_j$ can not be moved before the gate $C_i$ and hence $\comuCirc{i}{j}{C}\neq 0$, because the edges indicate that gates do not commute. It thus remains to show that  $\comuCirc{i}{j}{C}\neq0$ implies $G_j \in \texttt{Succ}(G_i,G)$. We show the claim by induction over the number of vertices in the graph $G$, where we build up the graph by adding the vertices in increasing order according to their labels.
The claim is clear for a graph consisting of two vertices, since we add an edge between them in Algorithm~\ref{algo:CreateCanonicalForm} if and only if the two corresponding gates in the circuit $C$ do not commute. 
Let us now consider the graph consisting of vertices $G^C_1,\dots,G^C_n$ and assume that the claim is true for all possible pairs of indices $i<j$ with $i,j \in \{1,2,\dots,n\}$. We have to show the claim for the graph consisting of vertices $G^C_1,\dots,G^C_{n+1}$. Since the vertex $G^C_{n+1}$ has no effect for the statement for indices $i<j$ with $i,j \in \{1,2,\dots,n\}$, we have left to show the statement for indices $j=n+1$ and an arbitrary $i<n+1$.
Clearly,  $\comuCirc{i}{n+1}{C}\neq0$ implies that there exists an index $i \leq k \leq n$ with $\comu{C_{n+1}}{C_k}\neq 0$ and with $\comuCirc{i}{k}{C} \neq 0$ if $k \neq i$. Let us choose the largest such index $k$. 
Then, by the induction assumption,  $G_k \in \texttt{Succ}(G_i,G) \cup \{G_i \}$. By the construction in Algorithm~\ref{algo:CreateCanonicalForm}, the vertex $G_k$ is still accessible if it is visited in the inner loop with respect to $i=n+1$ in the outer loop, and hence, an edge from $G_k$ to $G_{n+1}$ is added and since   $G_k \in \texttt{Succ}(G_i,G) \cup \{G_i \}$, we find   $G_{n+1} \in \texttt{Succ}(G_i,G)$. Indeed, to see that the vertex $G_k$ is still accessible when it is visited in Algorithm~\ref{algo:CreateCanonicalForm}, assume by contradiction that it would not be accessible. Then, $G_k$ must be a predecessor of a vertex $G_{k'}$ with $i\leq k<k'<n+1$ such that $\comu{G_{k'}}{G_{n+1}}\neq 0$. Since  $G_k \in \texttt{Succ}(G_i,G) \cup \{G_i \}$, this would imply that $G_{k'} \in  \texttt{Succ}(G_i,G)$. By the induction assumption, this means $\comuCirc{i}{k'}{C} \neq 0$. However, we assumed that $k$ is the largest index with these properties, and hence this leads to $k=k'$, which contradicts $k<k'$.
\end{proof}

\subsection{Difficulties for matching quantum circuits} \label{sec:possible_difficulties}
% We start with a discussion about difficulties in constructing an efficient template matching algorithm. 
%They mainly arise  from the fact that some quantum gates commute with each other, where others do not. These difficulties include:
%\begin{itemize}
%\item The ambiguity of circuit representations due to commuting gates.
%\item The fact that not the whole template must be matched, and hence we have to find the maximal match, which might not start with the first gate in the template.
%\item There are $\frac{n_C!}{(n_C-n_T)!}$ many mappings from the qubit labels of the template $T$ to the qubit labels of the circuit $C$.
%\end{itemize}
%The major difficulty that makes the straightforward search for matches inefficient is the ambiguity of circuit representations due to commuting gates. 
In this section, we give some intuition regarding possible difficulties a pattern matching algorithm for quantum circuits must overcome.
These difficulties mainly arise since we want to find all maximal matches considering pairwise commutation of gates, and from the fact that only some quantum gates commute, while others do not.
The following list may help the reader understand the structure of the pattern matching algorithm given in Section~\ref{sec:temp_match}, which handles all of these problems efficiently.

\begin{enumerate}
\item \textbf{Maximal match requires reordering of gates.} The simplest problem that appears due to commuting gates is illustrated in Figure~\ref{fig_tempMatch}.
If we just start matching the first gate of the pattern with the first gate of the circuit, we need to assign the second qubit of the pattern with the third one of the circuit, and hence the third gate of the pattern will not match. However, clearly the two circuits could be fully matched by commuting the first two gates in the circuit.\footnote{We recall that the target as well as the control nodes of different \cnot{}{} gates commute. However a target node does not commute with a control node.} 
\begin{figure}[htb]
\centering
\begin{subfigure}[b]{0.35\textwidth}
%\subfloat[template]{
\[
\Qcircuit @C=0.4em @R=0.4em {
&&1&2&3\\\\
& \qw  & \ctrl{1} & \ctrl{2} & \targ & \qw  \\
& \qw  & \targ & \qw  & \qw  & \qw  \\
& \qw  & \qw  & \targ & \ctrl{-2} & \qw  
}
\]
\subcaption{Pattern}
\end{subfigure}
%\subfloat[circuit]{
\begin{subfigure}[b]{0.35\textwidth}
\[\Qcircuit @C=0.4em @R=0.4em {
&&2&1&3\\ \\
& \qw  & \ctrl{2} & \ctrl{1} & \targ & \qw  \\
& \qw  & \qw  & \targ & \qw  & \qw  \\
& \qw  & \targ & \qw  & \ctrl{-2} & \qw  
}
\]
\subcaption{Circuit}
\end{subfigure}
\caption{Ordering problem when matching a pattern with a circuit. The numbers above the pattern label its gates. The numbers above the circuit which gate from the pattern should be matched with which gate from the circuit to achieve a maximal match.}
\label{fig_tempMatch}
\end{figure}

\item \textbf{Greedy strategy is not optimal.} \label{problems:additional_gates} It may happen that matching some gates early in the matching process will block other gates later in the matching process. Hence, a straightforward greedy approach is not necessarily optimal. Let us consider the following pattern and circuit depicted in Figure~\ref{fig_tempMatch2}.
If we match the first two gates, the third gate will not match. Furthermore, it is not possible to commute the matched gates next to the last gate in the circuit (which matches the third gate of the pattern) or vice versa. However, there does exist a full match, shown in Figure~\ref{fig_tempMatch2}, that does not match the second gate in the second (even though it could) in order to be able to match more gates later on.
\begin{figure}[htb]
\centering
\begin{subfigure}[b]{0.35\textwidth}
\[
\Qcircuit @C=0.4em @R=0.4em {
&&1&2&3\\ \\
& \qw  & \ctrl{1} & \ctrl{2} & \qw  & \qw  \\
& \qw  & \targ & \qw  & \targ & \qw  \\
& \qw  & \qw  & \targ & \ctrl{-1} & \qw  
}
\]
\subcaption{Pattern}
\end{subfigure}
\begin{subfigure}[b]{0.35\textwidth}
\[
\Qcircuit @C=0.4em @R=0.4em {
&&1&&&2&3\\ \\
& \qw  & \ctrl{1} & \ctrl{2} &\ctrl{1} & \ctrl{2} & \qw  & \qw  \\
& \qw  & \targ & \qw  & \targ & \qw  & \targ & \qw  \\
& \qw  & \qw  & \targ & \ctrl{-1} & \targ & \ctrl{-1} & \qw  
 \gategroup{3}{4}{5}{4}{.5em}{-}
}
\]
\subcaption{Circuit}
\end{subfigure}
\caption{Greedy matching does not always lead to a maximal match. The additional gate that would be matched in a greedy matching process starting at the first gate is marked with a solid box. The numbers above the pattern and circuit show the maximal match that can be found with a non-greedy strategy.}
\label{fig_tempMatch2}
\end{figure}

\item \textbf{Disturbing gates.} 
We consider a pattern and a circuit as given in Figure~\ref{fig_tempMatch3}. The second gate in the circuit ``disturbs'' the match. The maximal match is found by commuting it as far as possible to the left or the right. In the considered case, we can match three gates (instead of two) by commuting the second gate as far as possible to the left, i.e., leaving it where it is, instead of commuting it to the right.
%will lead to a complexity term in the template matching algorithm that growth exponential in the template size in the worst-case, since one has to consider all the different possibilities of commuting them to the left or the right. However, we expect the algorithm to handle these cases quite efficiently in practice. 
Disturbing gates are difficult to handle in general, since it is a priori unclear whether one should try to move them to the right or to the left in the circuit. If one always considers both options, the time complexity of such an algorithm would be exponential in the number of disturbing gates. 
\begin{figure}[htb]
\centering
\begin{subfigure}[b]{0.35\textwidth}
\[
\Qcircuit @C=0.4em @R=0.4em {
&&1&2&3&4\\\\
& \qw  & \qw  & \ctrl{2} & \qw  & \targ & \qw  \\
& \qw  & \ctrl{2} & \qw  & \ctrl{1} & \qw  & \qw  \\
& \qw  & \qw  & \targ & \targ & \ctrl{-2} & \qw  \\
& \qw  & \targ & \qw  & \qw  & \qw  & \qw  
}
\]
\subcaption{Pattern}
\end{subfigure}
\begin{subfigure}[b]{0.35\textwidth}
\[
\Qcircuit @C=0.4em @R=0.4em {
&&&&2&3&4\\\\
& \qw  & \qw  & \qw  & \ctrl{2} & \qw  & \targ & \qw  \\
& \qw  & \ctrl{2} & \targ & \qw  & \ctrl{1} & \qw  & \qw  \\
& \qw  & \qw  & \qw  & \targ & \targ & \ctrl{-2} & \qw  \\
& \qw  & \targ & \qw  & \qw  & \qw  & \qw  & \qw  \\
& \qw  & \qw  & \ctrl{-3} & \qw  & \qw  & \qw  & \qw  
 \gategroup{4}{4}{7}{4}{.5em}{-}
}
\]
\subcaption{Circuit}
\end{subfigure}
\caption{It might be unclear to which position we should move a disturbing gate (marked by a solid box in the circuit) to find a maximal match. The numbers above the pattern label its gates. The number above the circuit refer to the labels of the gates in the pattern that can be matched with the corresponding gate in the circuit.}
\label{fig_tempMatch3}
\end{figure}
\end{enumerate}

%%%%%%%%%%%%%%%%%%%%%%%%%%%%%%%%%%
\section{Pattern matching algorithm}\label{sec:temp_match}

We now describe our main contribution, a pattern matching algorithm for quantum circuits. An example of this algorithm was already given in \cref{sec:algo_intro} and we recommend thoroughly understanding the example before proceeding with the pseudocode and correctness proof in this section.

\subsection{Pseudocode for the matching algorithm} \label{sec:pseudo_code}
In this section, we describe the pseudocode of the pattern matching algorithm. We stress that the focus of the code is readability and we do not optimize the constants of the runtime.
%The template matching algorithm uses many different subroutines. The most important ones are given here and the remaining ones are explained in Appendix~\ref{app_pseudoCode}. 
We usually think of working with pointers to circuit or graph objects. Hence, an object might be modified by a method call, however it is not given back as an output. As a result, we may sometimes have to copy an object $o$, by calling $o.$\texttt{copy}. 
%We will use the notation $(o,p).$\texttt{copy} to copy an object $p$ together with a pointer $p$ that points on a sub-object of $o$. Hence, the copied pointer will show to the corresponding sub-object in the copy. 

The algorithm primarily works with the canonical form of quantum circuits (see \cref{sec:canonical_form}). In the pseudocode, we will allow to add different attributes to vertices in the graph representing the circuit, which we access by ``vertex.attribute''. In particular, we always add an attribute ``label'' referring to the gate index corresponding to the vertex. Further, we use $G_i$ to access the vertex with label $i$ in the graph $G$ and we denote the number of vertices in $G$ by $|G|$. Note that we can store the vertices of the graph in an array at positions according to their index to have constant time access to any vertex with known label. Storing incoming and outgoing edges together with each vertex (as pointers to the direct successors and predecessors), we also have constant time access to all direct successors and predecessors of any given vertex. 

We first give the pseudocode for the main algorithm \texttt{PatternMatch}, followed by the two main subroutines \texttt{ForwardMatch} and \texttt{BackwardMatch}.
On the relation between the pseudocode and the example in \cref{sec:algo_intro} we note that the property $\textit{isBlocked}$ in the pseudocode indicates that a vertex has been either right-blocked or left-blocked.
The distinction between right-blocking and left-blocking from the example is useful for intuition, but not necessary for the algorithm.

The algorithm \texttt{PatternMatch} is roughly structured as follows: we loop over all possible starting matches of a gate in the circuit with a gate in the pattern, and over the possible qubit assignments. Then, the algorithms \texttt{ForwardMatch} and \texttt{BackwardMatch} are used to maximally expand the match, under the condition that the starting match must be preserved. Importantly, it is not necessary to always consider the full pattern for matching: if the $i$-th gate of the pattern is matched with a gate in the circuit at the start, it suffices to maximally expand the match for the sub-pattern $(T_i,\dots,T_{|T|})$. The reason for this is that any match that includes a gate $T_j$ from $T_1,\dots,T_{i-1}$ would already have been found in a previous round of the loop where $T_j$ was chosen as the starting gate. This is proven in more detail in Lemma~\ref{thm_correctness}.

%\RI{ToDo: Avoid finding partial matches of already found longer matches. Should be easy. I will add this soon.}
%\tm{this was commented out, is it still a todo?} \RI{@Tony: in principle... but we did never do it.. and it would complicate the code further. So, let's ignor it :)}

\begin{algorithm}[H] 
  \caption{\texttt{PatternMatch}: Pattern matching algorithm }
  \label{algo:TempMatch}
   \begin{algorithmic}[1]
   \State Input: $(C,T)$ \begin{itemize}
   \item Circuit $C$ with $n_C$ qubits and $|C|$ gates
   \item Pattern $T$ with $n_T$ qubits and $|T|$ gates
   \end{itemize}
%   \State  Create the $n_T!$ templates that one gets by reordering the qubits of $T$ %, bring them into the canonical form
%    and store them into a list $L_T$. (This has only to be performed once when a new template is added.)
  %\State Bring the circuit $C$ into the canonical form (we call the new circuit again $C$ for simplicity).
  \State Initialize a list $W$ to store matches
  \State $G^C \leftarrow \texttt{CreateCanonicalForm}(C)$ \Comment{stored in the adjacency-list format}  \label{TempMatch_canonical_form_C}
  \State $G^T \leftarrow \texttt{CreateCanonicalForm}(T)$ \label{TempMatch_canonical_form_T}
  \State $L_{\textnormal{q}} \leftarrow \{1,2,\dots,n_{C}\}$
%   \For{$k \in \{1,2,\dots,|L_T|\}$}  \Comment{{\color{gray}loop through all possible qubit orderings}}
%     \State $T \leftarrow  L_T[k]$
  \For{$i \in \{1,2,\dots,|T|\}$}  \Comment{loop through all gate indices of $T$ for starting a match at $T_i$}
   \For{$r \in \{k \in \{1,\dots,|C| \} : C_k \cong T_i \}$ } \Comment{loop through the indices of gates in $C$ with $C_r\cong T_i$}
   \State $L_{\textnormal{q}}^{\textnormal{action}}\leftarrow  \{s \in L_{\textnormal{q}}: \textnormal{the gate $C_r$ is acting non trivially on the qubit with label $s$} \}$
  \For{$L_{\textnormal{q}}^{\textnormal{sel}} \in \{L  \in {L_{\textnormal{q}} \choose n_T}:  L_{\textnormal{q}}^{\textnormal{action}}\subset L \}$}  \Comment{loop through all possible choices of qubits} \label{TempMatch_loop_qubit_choice}
    \State{$L_{\textnormal{q}}^{\textnormal{sel}} \leftarrow$ Sort the qubits in $L_{\textnormal{q}}^{\textnormal{sel}}$ in increasing order according to their labels (store as list)}
  \For{$p \in S_{n_T}$}  \Comment{{\color{gray}loop through all possible qubit orderings}}
  \State{$\tilde{T} \leftarrow$ Label the qubits in $T$ with the labels in $L_{\textnormal{q}}^{\textnormal{sel}}$ using the mapping $t \mapsto  (L_{\textnormal{q}}^{\textnormal{sel}})_{p(t)}$} \label{TempMach_relabeling}
%   \State{  {\color{gray}$\#$ Note that the canonical form of $T$ and $\tilde{T}$ is the same}}
     % \State $\tilde{C}_{\textnormal{DAG}}= \textit{CommuteLeft}(C_{\textnormal{DAG}},r,L_{\textnormal{q}}^{\textnormal{sel}})$
      %Move all gate in $C$ with index $k>r$ that do act on the qubits in $L_{\textnormal{q}}$ and on at least one qubit that is not in $L_{\textnormal{q}}$ before the root gate with index $r$ if it commutes with all the gate in-between.
     %  \State $\hat{C}_{\textnormal{DAG}}= \textit{CommuteRight}(\tilde{C}_{\textnormal{DAG}},r,L_{\textnormal{q}})$
      %\State $C_{\textnormal{sub}}=\textnormal{Extract}[\hat{C}_{\textnormal{DAG}},L_{\textnormal{q}},r]$
       %Move all gate in $C$ with index $k>r$  that do act on the qubits in $L_{\textnormal{q}}$ and on at least one qubit that is not in $L_{\textnormal{q}}$ as far to the ``right'' in the circuit as possible
   \If{$C_r=\tilde{T}_i$} \Comment{Check if the qubit permutation is such that $C_r$ matches $\tilde{T}_i$}
    \State {{\color{gray}$\# $ Find the maximal matches of the partial pattern $(\tilde{T}_i,\dots, \tilde{T}_{\tilde{|T|}})$ in the circuit $C$}}
     \State {{\color{gray}$\# $ under the restriction that $\tilde{T}_i$ is matched with $C_r$}}
    \State {{\color{gray}$\# $ We match the maximal part in forward direction of $\tilde{T}_i$:}}
    \State {$\MRight \leftarrow$ \texttt{ForwardMatch}$(C,G^C,\tilde{T}, G^T, L_{\textnormal{q}}^{\textnormal{sel}},r,i)$}  \Comment{Attributs are added to $G^C$}
     \State{  {\color{gray}$\#$  Expand forward match to maximal ones with partial pattern $(\tilde{T}_i,\dots,\tilde{T}_{|\tilde{T}|})$ :}} 
      \State $L' \leftarrow$ \texttt{BackwardMatch}$(C,G^C,\tilde{T},G^T,L_{\textnormal{q}}^{\textnormal{sel}},r,i,\MRight)$
   %\RI{Maybe we can remove $L_{\textnormal{q}}^{\textnormal{sel}}$ from the input, since the information is part of the qubit labeling in $\tilde{T}$}
   \State{Add the matches in $L'$ to the list of matches $W$} 
   \EndIf
     %\State{Determine maximal match starting with gate $r$ in $C$ and starting position $i$ in the template.}
           \EndFor
       \EndFor
        \EndFor
          \EndFor
   \State Output: the list $W$ of matches 
   \end{algorithmic}
\end{algorithm}

Note that the elements $M$ in the output $W$ in Algorithm~\ref{algo:TempMatch} are sets containing index pairs $(i,j)$ of matched gates, i.e., if the match $M$ contains $(i,j)$, the gate $T_i$ from the pattern was matched with the gate $C_j$ from the circuit. The qubit mapping can then be recovered from the matched gates.

\begin{rmk}
The output of \texttt{PatternMatch} indicates how to match the gates of the pattern with the gates in the circuit. However, the information of how to commute the gates in the circuit to move the matched gates next to each other is not contained for simplicity. This information can be restored efficiently by commuting the gates in between the matched gates to the left or to the right of the match.
\end{rmk}

\begin{algorithm}[H] 
  \caption{\texttt{ForwardMatch}: Find maximal match in forward direction}
  \label{algo:ForwardMatch}
   \begin{algorithmic}[1]
   \State Input: $(C,G^C,T, G^T,L_{\textnormal{q}}^{\textnormal{sel}},r,i)$ 
   \begin{itemize}
%     \item Circuit $C$ with $n_C$ qubits and $|C|$ gates
%    \item Canonical form $G^C$ of $C$  
%      \item Template $T$ with $n_T$ qubits and $|T|$ gates
%      \item Canonical form $G^T$ of $T$ 
%      \item List of qubit labels $L_{\textnormal{q}}^{\textnormal{sel}}$ in $C$ (on which we try to match the template on)
%    \item Gate index $r$ in $C$ (where we start matching)
%     \item Gate index $i$ in $T$ (where we start matching)   
   \item Circuit $C$ with canonical form $G^C$
   \item Pattern $T$ with canonical form $G^T$
   \item Ordered list $L_{\textnormal{q}}^{\textnormal{sel}}$ of qubit labels in $C$ (with the first qubit in $L_{\textnormal{q}}^{\textnormal{sel}}$ matched with the first pattern qubit, etc.)
    \item Gate indices $r$ in $C$ and $i$ in $T$ (where we start matching, i.e., the first match is $C_r=T_i$)
   \end{itemize}
     \State{{\color{gray}$\#$ \textbf{Initialization:}}}
     \State{Initialize a set $M\leftarrow \{(i,r)\}$ to store matched gate indices}
     \State{ For all vertices in $G^C$, initialize attributes 
     \begin{itemize}
     \item  $\textit{SuccessorsToVisit} \leftarrow ()$, except set $G^C_{r}.\mathit{SuccessorsToVisit}\leftarrow$ $\big($list containing \texttt{DirectSucc}$(G^C_r,G^C)$, in ascending order of label$\big).$\footnotemark[9]
     \item \textit{matchedWith} $\leftarrow$ \textit{null}, except set $G^C_{r}.\mathit{matchedWith}\leftarrow G^T_i$. 
     \item $\mathit{isBlocked}\leftarrow \mathit{false}$.
     \end{itemize}}
   \State{Initialize a list \textit{MatchedVertexList} and add $G^C_r$ as a first element. The list is ordered in ascending order according to the label of the first element of $\textit{SuccessorsToVisit}$ of each vertex, i.e., $G^C_i$ precedes $G^C_j$ if the label of the first element in $G^C_i.\textit{SuccessorsToVisit}$ is smaller than the label of the first element in $G^C_j.\textit{SuccessorsToVisit}$.\footnotemark[10] }
   \State{{\color{gray}$\#$ \textbf{Forward matching proccess:}}}
   \While{\textit{MatchedVertexList} is not empty}
     \State {$v_0\leftarrow \textit{MatchedVertexList}.\texttt{get}(1)$} \Comment{matched vertex as a root for further matching}
   \If{ $v_0.\mathit{SuccessorsToVisit}$ is empty}
   \State{GoTo ``EndOfWhileLoop''}
\EndIf
%   \If{$v_0.\mathit{SuccessorsToVisit}=\emptyset$}
%   \State{GoTo ``EndOfWhileLoop''}
%   \Else
%   \State{$v$ $\leftarrow$ vertex in the set $v_0$.\textit{SuccessorsToVisit} with minimal label} 
%      \State{Remove the vertex $v$ from the set $v_0$.\textit{SuccessorsToVisit}}
%    \State{\textit{MatchedVertexList}.\texttt{Enqueue}($v_0$) } \Comment{Enqueue vertex with modified attribute to the queue}
%   \EndIf
  \State{$v \leftarrow v_0.\mathit{SuccessorsToVisit}.\texttt{get}(1)$; $s \leftarrow v.\mathit{label}$}  \Comment{vertex to consider for matching}
%  \If{$d_{G^C}(G^C_r,v)>d_{G^C}(G^C_r,v_0)+1$}
% \State{GoTo ``EndOfForLoop''} \Comment{The vertex $v$ is two layers away and is considered later on}
%  \EndIf
\State{\textit{MatchedVertexList}.\texttt{Insert}($v_0$) } \Comment{put  vertex back with modified $\mathit{SuccessorsToVisit}$ } \label{ForwardMatch_insert_v0}
\If{$v.\mathit{isBlocked}$ or $v.\mathit{matchedWith} \neq \mathit{null}$} \Comment{vertex already blocked or matched}
   \State{GoTo ``EndOfWhileLoop''}
\EndIf
   \State{{\color{gray}$\#$  We now try to add the vertex $v$ to the match $M$.}}
       \State{$\mathit{CandidateIndices}\leftarrow \mathit{FindForwardCandidates}(G^T,v_0.\mathit{matchedWith}, M)$} \label{ForwardMatch_find_candidates}
       \If{There exist a $j \in \mathit{CandidateIndices}$ with $C_s=T_j$ } \Comment{\textbf{we found a match with $v$}} \label{ForwardMatch_if} \label{ForwardMatch_matching_case}
       \State{$j \leftarrow$ Choose the  minimal $j \in \mathit{CandidateIndices}$ with $C_s=T_j$}
   \State{$v.\mathit{matchedWith}\leftarrow G^T_j$; Add $(j,s)$ to $M$}
    \State \parbox[t]{\dimexpr\linewidth-\algorithmicindent}{$v.\mathit{SuccessorsToVisit} \leftarrow \{w \in \texttt{DirectSucc}(v,G^C): w.\mathit{isBlocked}=\mathit{false} ;w.\mathit{matchedWith}=\mathit{null} \}$ sorted in ascending order of label } \label{forward_match_set_to_order}
%    \State{Remove the elements $v \in S$ with $d_{G^C}(G^C_r,v)>d_{G^C}(G^C_r,v_0)+1$}
    \State{\textit{MatchedVertexList}.\texttt{Insert}($v$)}
\Else \Comment{\textbf{no match with $v$ was found}}  \label{ForwardMatch_no_matching_case}
  \State{Set the attribute \textit{isBlocked} equal to \textit{true} for the vertex $v$ and all of its successors } \label{algo:ForwardMatch_blocking_successors}
  \EndIf
    \State{Label ``EndOfWhileLoop''}
   \EndWhile
    \State Output: $M$
   \end{algorithmic}
\end{algorithm}
\footnotetext[9]{The method \textit{SuccessorsToVisit}.\texttt{Get}$(i)$   returns the $i$th vertex from \textit{SuccessorsToVisit} and removes it from the list.}

\footnotetext[10]{The method \textit{MatchedVertexList}.\texttt{Get}$(i)$   returns the $i$th vertex from  \textit{MatchedVertexList} and removes it from the list. The method \textit{MatchedVertexList}.\texttt{Insert}$(v)$  adds a vertex $v$ at the correct position according to the ordering of \textit{MatchedVertexList} described in the algorithm.}
\addtocounter{footnote}{+2} 

\begin{algorithm}[H] 
  \caption{\texttt{BackwardMatch}: Find maximal expansions of the forward match in backward direction}
  \label{algo:BackwardMatch}
   \begin{algorithmic}[1]
    \State Input: $(C,G^C,T,G^T,L_{\textnormal{q}}^{\textnormal{sel}},r,i,\MRight)$
     \begin{itemize}
   \item Circuit $C$ with canonical form $G^C$ with attributes ``\textit{matchedWith}'' and ``\textit{isBlocked}'' (which, in the context of the algorithm \texttt{PatternMatch}, have been assigned by \texttt{ForwardMatch})
    \item Pattern $T$ with canonical form $G^T$
   \item List $L_{\textnormal{q}}^{\textnormal{sel}}$ of $n_T$ qubit labels we are matching on
  \item Gate indices $r$ in $C$ and $i$ in $T$ (where we start matching)
   \item Set of matched index pairs  $\MRight$
   \end{itemize}
     \State{{\color{gray}$\#$ \textbf{Initialization:}}}
   %\State{Initialize an attribute \textit{MatchedSuccessors} for all vertices in $G^C$}
   %\State{\textit{MatchC}$=\{G_i:i \in \{1,\dots,|G|\} \textnormal{ with $G_i$.\textit{matchedWith}}\neq$\textit{null})}
   %\State{\texttt{UpdateAttributeMatchedSuccessors}$(G^C, \mathit{MatchC}\}$}
   \State{$W \leftarrow ()$ }  \Comment{List to store matchings}
   \State{Initialize an attribute ``\textit{isBlocked}=\textit{false}'' and ``\textit{matchedWith}=\textit{null}''  for all vertices in $G^T$}
   \State{Update the attribute \textit{matchedWith} of $G^T$ according to the matched index pairs listed in $\MRight$}
   \State{Set ``\textit{isBlocked}=\textit{true}'' for all successors of $G_i^T$ }
   %\State{$\mathit{numberOfForwardMatches} \leftarrow |\MRight|$}
   \State{$\mathit{GateIndices} \leftarrow \{l \in \{1,2,\dots |G^C| \}: G^C_l.\mathit{isBlocked}=\mathit{false} \textnormal{ and } G^C_l.\mathit{matchedWith}=\mathit{null} \}$ } \label{BackWardMatch_GateIndices1}
   \State{$\mathit{GateIndices} \leftarrow$ Order $\mathit{GateIndices}$ in decreasing order (and store as list)}  \label{BackWardMatch_GateIndices2}
   \State{$\mathit{counter} \leftarrow 1 $ \Comment{used to loop through \textit{GateIndices}}}
   \State{$\mathit{numberOfGatesLeftToMatch} \leftarrow |T|-(i-1)-|\MRight|$ \Comment{number of remaining unmatched gates in the sub-pattern $(T_i, \dots, T_{|T|})$}}
 % \State{ \textit{Candidates}$\leftarrow$  \texttt{FindBackwardCandidates}$(G^C,\MRight)$}
     \State{{\color{gray}$\#$ Initialize a stack to save all matching scenarios that should be considered for expansion:}}
      \State{Initialize a stack\footnotemark[11]  \textit{MatchingScenarios} and call \textit{MatchingScenarios}.\texttt{Push}$(G^C,G^T,\MRight,\mathit{counter} )$} 
        \State{{\color{gray}$\#$ \textbf{Start the matching process:}}}
   \While{\textit{MatchingScenarios} is not empty} \label{algo_while_loop}
   \State{($G^C,G^T,M,counter)\leftarrow \mathit{MatchingScenarios}.\texttt{Pop}$ \Comment{consider top matching scenario in stack}}
   \State{$s\leftarrow (\mathit{GateIndices})_\mathit{counter}$} \Comment{consider gate $C_s$ for matching}
   \State{$v \leftarrow G^C_s$}
     \State{{\color{gray}$\#$ \textbf{Trivial cases}}}
     \State{$\MLeft \leftarrow M \setminus \MRight$ \Comment{matches added during backward match so far}}
     \If{$\mathit{counter} = |\mathit{GateIndices}|$ or $|\MLeft|=\mathit{numberOfGatesLeftToMatch}$} \label{algo_counter_check}
     \State{Add $M$ to $W$}
       \State{GoTo ``EndOfWhileLoop''}
     \EndIf
   \If{$v.\mathit{isBlocked}$ }
   \State{{\color{gray}$\#$Blocked gates are skipped by simply incrementing \textit{counter}}}
   \State{\textit{MatchingScenarios}.\texttt{Push}$(G^C,G^T,M,\mathit{counter}+1 )$}
       \State{GoTo ``EndOfWhileLoop''}
   \EndIf
    \algstore{myalg}
  \end{algorithmic}
\end{algorithm}
\footnotetext[11]{The method \textit{MatchingScenarios}.\texttt{Pop} returns the top element of the stack \textit{MatchingScenarios} and removes it from the stack. The method \textit{MatchingScenarios}.\texttt{Push}$(v)$ adds a vertex $v$ at the top of the stack.}
   \footnotetext[12]{\label{footmark:additional_check}Fixed matches are the starting match and matches added in previous rounds of the \texttt{BackwardMatch} algorithm. We do not need to consider scenarios that ``unmatch'' gates matched during previous rounds of \texttt{BackwardMatch} because these are already considered as separate scenarios on the stack (see the proof of \cref{lem_backwardsMatch}). To improve the runtime, we could in addition check if the length of the match $\tilde{M}$ plus the number of gates in the pattern that could possibly be matched in the further backwards matching process, is smaller than the length of the initial forward match $\MRight$. If this is the case, we could ignore this matching scenario because it can not lead to a match that is at least as long as the forward match, and hence not to a maximal match. }
\addtocounter{footnote}{+2}   % we need to compensate for two footnote numbers that we added manually.
   
 \begin{algorithm}[H] 
%  \label{algo:CompleteLeftCommute}
   \begin{algorithmic}[1]
\algrestore{myalg}
  \State{{\color{gray}$\#$ \textbf{Try to match the gate $C_s$ corresponding to vertex $v$:}}}
    \State{\textit{CandidateIndices} $\leftarrow$  \texttt{FindBackwardCandidates}$(G^T,i)$ }
   \If{$C_{s} \in \{T_{j}: j \in \mathit{CandidateIndices} \}$ } \label{BackwardMatch:if_condition}
   \State{{\color{gray}$\#$ There exists at least one match with the gate $C_s$}} \label{BackwardMatch:case1}
   \State{\textit{I} $\leftarrow \{k \in$ \textit{CandidateIndices} : $T_k=C_s\}$}
    \State{{\color{gray}$\#$ Remove candidates leading to equivalent matches}}
   \State{Successively remove index $i$ from \textit{I} for which there exists $j \in$\label{algo_index_pruning} \textit{I} with $[i,j]_T=0$} 
   \State {{\color{gray}\bf$\#$ Option 1.1: we match gate $C_s$ and add the result to the stack \textit{MatchingScenarios}} }
 \For{$j \in \textit{I}$ \Comment{loop over all possible inequivalent matches for $C_s$}}  \label{algo_backwards_for_loop}
   \State{$\tilde{G}^C \leftarrow G^C.\texttt{copy}$, $\tilde{G}^T \leftarrow G^T.\texttt{copy}$ and  $\tilde{M}\leftarrow M.\texttt{copy}$ } 
     \State{Block all successors of $\tilde{G}^T_j$ in $\tilde{G}^T$ that were not already matched}
     \State{Block all successors of all the blocked vertices in $\tilde{G}^T$}
     \State{Set \textit{matchedWith}=\textit{null} for the blocked vertices in $\tilde{G}^T$ and update $\tilde{M}$ accordingly} \label{algo:blocking}
     \If{$(i,r) \in \tilde{M}$ and $\MLeft \subset \tilde{M}$ }   \Comment{Check if we block ``fixed'' matches$^{\textnormal{\ref{footmark:additional_check}}}$}
     \State{$\tilde{M} \leftarrow \tilde{M} \cup  \{(j,s)\}$}
        \State{$\tilde{v} \leftarrow \tilde{G}^C_s$}
     \State{$\tilde{v}.\mathit{matchedWith} \leftarrow \tilde{G}^T_j$} \Comment{Updates attribute of vertex $\tilde{v}$ in graph $\tilde{G}^C$}    
     \State{\textit{MatchingScenarios}.\texttt{Push}$(\tilde{G}^C,\tilde{G}^T,\tilde{M},\mathit{counter}+1 )$}
     \EndIf
             \EndFor
 \State {{\color{gray}\bf$\#$ Option 1.2: we block the vertex $v$ corresponding to gate $C_s$} }
   \State{$v.\mathit{isBlocked} \leftarrow true$}
 \State{$\mathit{followingMatches} \leftarrow \{w \in \texttt{Succ}(v,G^C):w.\mathit{matchedWith}\neq \mathit{null} \}$ }
  \State {{\color{gray}$\#$ Option 1.2a: right-block or left-block without interfering with previously matched gates.}} 
 \If{$\texttt{Pred}(v,G^C)=\emptyset$ or $\mathit{followingMatches}=\emptyset$}  \label{BackwardMatch_move_to_start_end}
  \State{\textit{MatchingScenarios}.\texttt{Push}$(G^C,G^T,M,\mathit{counter} +1)$}
\Else \Comment{blocking requires unmatching some of the previously matched gates}
     \State {{\color{gray}$\#$ Option 1.2b: we right-block the vertex $v$ corresponding to gate $C_s$ }}  \label{algo:BackwardMatch_option2} 
        \State{$\hat{G}^C \leftarrow G^C.\texttt{copy}$, $\hat{G}^T \leftarrow G^T.\texttt{copy}$ and  $\hat{M}\leftarrow M.\texttt{copy}$ } 
         \State{$\hat{v} \leftarrow \hat{G}^C_s$}
      \State{$\hat{v}.\mathit{isBlocked}\leftarrow \mathit{true}$}
     \State{Block all successors of $\hat{v}$ and set their attribute \textit{matchedWith}=\textit{null} } \label{algo:line_blocking}
     \State{Update $\hat{M}$ if matched gates were blocked in line~\ref{algo:line_blocking}}
     \If{$(i,r) \in \hat{M}$ and $\MLeft \subset \hat{M}$ }\Comment{Check if we blocked ``fixed'' matches.\footnotemark[12]} \label{BackwardMatch:adding_to_stack_condition_1}
       \State{\textit{MatchingScenarios}.\texttt{Push}$(\hat{G}^C,\hat{G}^T,\hat{M},\mathit{counter} +1)$}
        \EndIf
    \algstore{myalg}
     \end{algorithmic}
\end{algorithm}

  \begin{algorithm}[H] 
%  \label{algo:CompleteLeftCommute}
   \begin{algorithmic}[1]
\algrestore{myalg}
\State {{\color{gray}$\#$ Option 1.2c: we left-block the vertex $v$ corresponding to gate $C_s$. This option is only}}
\State {{\color{gray}$\#$  viable if matching $C_s$ or blocking it without blocking previously matched gates (line~\ref{BackwardMatch_move_to_start_end})}}
\State {{\color{gray}$\#$ is impossible.}}
\label{algo:BackwardMatch_option3}
     \If{every option in $I$ considered in line~\ref{algo_backwards_for_loop} required blocking of gates in line~\ref{algo:blocking}} \label{backwardsMatch_check_if_matches_destroyed}
         \State{Block all predecessors of the vertex $v$}
      \State{\textit{MatchingScenarios}.\texttt{Push}$(G^C, G^T,M,\mathit{counter} +1)$}
    \EndIf
    \EndIf
 \Else
   \State{{\color{gray} $\#$ Option 2.1: We only block (and cannot match)}}
  \State{{\color{gray}$\#$ There is no gate in the pattern that matches the gate $C_s$. Hence, we block $v$.}}  \label{BackwardMatch:case2}
  \State{{\color{gray}$\#$ The options we consider are similar to the ones above where we also blocked $v$.}}
   \State{$v.\mathit{isBlocked} \leftarrow true$} \label{algo_backwardMatch_block}
 \State{$\mathit{followingMatches} \leftarrow \{w \in \texttt{Succ}(v,G^C):w.\mathit{matchedWith}\neq \mathit{null} \}$ }
 \State {{\color{gray}$\#$Option 2.1a}}
 \If{$\texttt{Pred}(v,G^C)=\emptyset$ or $\mathit{followingMatches}=\emptyset$}  \label{BackwardMatch_move_to_start_end2}
\State {{\color{gray}$\#$ We can right-block or left-block without interfering with previously matched gates.}}
  \State{\textit{MatchingScenarios}.\texttt{Push}$(G^C,G^T,M,\mathit{counter} +1)$}
    \Else
     \State {{\color{gray}$\#$ The gate $C_s$ corresponding to $v$ might disturb the expansion of the match. }} 
   \State {{\color{gray}$\#$ We can either right-block or left-block $v$. }} 
           \State {{\color{gray}$\#$ Option 2.1b: right-block $v$}} 
   \State{$\tilde{G}^C \leftarrow G^C.\texttt{copy} $, $\tilde{G}^T \leftarrow G^T.\texttt{copy} $ and  $\tilde{M}\leftarrow M.\texttt{copy}$ }
        \State{$\tilde{v} \leftarrow \tilde{G}^C_s$}
         \State{Block all successors of the vertex $\tilde{v} $ and and set their attribute \textit{matchedWith}=\textit{null} } \label{algo:line_blocking_2}
          \State{Update $\tilde{M}$ if matched gates were blocked in line~\ref{algo:line_blocking_2}}
 \If{$(i,r) \in \tilde{M}$ and $\MLeft \subset \tilde{M}$ }  \label{BackwardMatch:adding_to_stack_condition_2} \Comment{Check if we block ``fixed'' matches$^{\textnormal{\ref{footmark:additional_check}}}$}
       \State{\textit{MatchingScenarios}.\texttt{Push}$(\tilde{G}^C,\tilde{G}^T,\tilde{M},\mathit{counter} +1)$}
        \EndIf
        \State {{\color{gray}$\#$ Option 2.1c: left-block $v$}} 
        \State{Block all predecessors of the vertex $v$}
          \State{\textit{MatchingScenarios}.\texttt{Push}$(G^C,G^T,M,\mathit{counter} +1)$}
          
  \EndIf
   \EndIf
     \State{Label ``EndOfWhileLoop''}
   \EndWhile
   \State{$\mathit{maxLength} \leftarrow \textnormal{max}\{|M|:M \in W \}$}
   \State{Remove all the matches in $W$ that have smaller length than $\mathit{maxLength}$}
      \State Output: $W$
 \end{algorithmic}
\end{algorithm}

\begin{algorithm}[H] 
  \caption{ \texttt{FindForwardCandidates}: Finds the indices of the gates that might match next for forward matching}
  \label{algo:FindForwardCandidates}
   \begin{algorithmic}[1]
   \State Input: ($G$, $v$, $M$)
   \begin{itemize}
       \item Canonical form  $G$
       \item Vertex $v$ in $G$
       \item Set of matched gate indices $M$, where the first index corresponds to the labels in $G$
   \end{itemize}
   \State $Match \leftarrow \{i:i\in \{1, . . . ,|G|\}$ such that there exists a $j$ with $(i, j)\in M\}$
   \State $Block \leftarrow \{\}$
   \For{$l \in Match~\setminus~v.label$}
   \For{$v'\in \texttt{DirectSucc}(l,G) \setminus Match$}
   \State Add labels of $\texttt{Succ}(v',G)$ to $\mathit{Block}$     \Comment{Exclude candidates leading to unconnected matches}
   \EndFor
   \EndFor
   \State $\mathit{CandidatesIndices} \leftarrow \left(\texttt{DirectSucc}(v,G) ~\setminus~\mathit{Match} \right)~\setminus~\mathit{Block}$
   \State Output: $\mathit{CandidatesIndices}$
   \end{algorithmic}
\end{algorithm}

\begin{algorithm}[H] 
  \caption{\texttt{FindBackwardCandidates}: Finds the indices of the gates that might match next for backward matching}
  \label{algo:FindBackwardCandidates}
   \begin{algorithmic}[1]
   \State Input: $(G,i)$ \begin{itemize}
   \item Canonical form $G$
    \item Start index $i$
   \end{itemize}
   \State{$S \leftarrow \{l \in \{i+1,\dots,|T|\}: G_l$ is not a successor of $G_i$ and $G_l.\textit{matchedWith}\neq null  \}$}
   \State{ \textit{CandidateIndices} $ \leftarrow \{l \in S: G_l.\textit{isBlocked}=false \}$}
   \State Output: \textit{CandidateIndices}
   \end{algorithmic}
\end{algorithm}

%   \begin{algorithm}[H] 
%  \caption{\texttt{AddMatchesToList}: Add a match to a list of matches and clear up the list}
%  \label{algo:AddMatchesToList}
%   \begin{algorithmic}[1]
%   \State Input: $(L_{\textnormal{toAdd}},W)$ \begin{itemize}
%   \item List of matches $L_{\textnormal{toAdd}}$ that should be added to the list $W$
%   \item List of matches $W$
%   \end{itemize}
%   \For{$ M \in L_{\textnormal{toAdd}}$ }
%   \For{$\tilde{M} \in W$}
%    \If{$\tilde{M}\subset M$}
%     \State{Remove $\tilde{M}$ in $W$}
%      \ElsIf{$M \subset \tilde{M}$}
%      \State{Goto ``End''}
%     \EndIf
%     \EndFor
%     \State{Add $M$ to $W$}
%     \State{Label ``End''}
%      \EndFor
%   \State Output: $W$
%   \end{algorithmic}
%\end{algorithm}

%%%%%%%
%%%%%%%%%%%%%%%%

%%%%%%%%%%%%%%%%%%%%%%%%%%%%%%%%%%%%%%%%%%%%%%%%%%%%%%%%%%%%%%%%%%%%%

%%%%%%%
\subsection{Correctness of the algorithm} \label{sec_correctness}
In this section we formally prove the correctness of~\texttt{PatternMatch} (Algorithm~\ref{algo:TempMatch}), i.e., that for any circuit $C$ and any pattern $T$ the algorithm finds all maximal matches.
Let us first formally define what we consider to be a ``pattern match'' and when it is called maximal.

\begin{definition}[Connected part of a circuit]
We say that  $E=(E_{1},\dots,E_{|E|})$ is a connected part of  a circuit $C=(C_1,\dots,C_{|C|})$ if, by commuting gates pairwise with each other, one can bring the circuit to the form $C \simeq (D, E,  F)$, where the circuits $D$ and $F$ consist of all the gates of the circuit $C$ except the ones listed in $E$.
\end{definition}

In terms of the canonical form, a part $E$ of a circuit $C$ is connected if, and only if, for all the vertices corresponding to the gates in $E$ in the canonical form of $C$ we have the following property: if two vertices are connected by a path, then all the vertices that lie on the path have also to correspond to gates in $E$.
Whether or not a  circuit is connected is not related to whether its canonical graph is connected (in the usual sense of connectedness for graphs).

\begin{definition}[Equivalence of circuits up to qubit relabeling]
A circuit $C$ is equivalent to a circuit $E$ up to qubit relabeling if, and only if, there exists a bijective mapping from the qubit labels in circuit $C$ to the ones in circuit $E$, such that for the resulting circuit $C'$ (that one gets by relabeling the qubits in circuit $C$) we have that $C'=E$ (i.e.,  $C'_i=E_i$ for all $i \in \{1,2,\dots,|C|=|E|\}$).
\end{definition}

\begin{definition} [Pattern match] \label{def:template_match}
We say that a pattern $T$ has a match of length $m$ in a circuit $C$ if there exists a connected part $E^T$ of $T$ of length $|E^T|=m$ that is equal up to qubit relabeling to a connected part $E^C$ of $C$. We refer to such a match $M$ by a set of tuples  of gate indices, where a tuple $(i,j)$ means that we matched the gate $T_i$ with the gate $C_j$. 
\end{definition}

\begin{definition} [Maximal pattern match] \label{def:maximal_template_match}
We say that a match $M$ is maximal if there are no matches $\tilde{M}$ in the circuit such that
\begin{itemize}
\item $|\tilde{M}| >|M|$, and
\item $M \cap \tilde{M} \neq \emptyset$, i.e., $M$ and $\tilde{M}$ have at least one element of matched gate indices $(i,j)$ in common. 
\end{itemize}
Intuitively, a maximal match is one that cannot be extended further (but another larger disjoint match might exist nonetheless).
\end{definition}

\begin{definition}[Equivalence of sub-circuits]
Let us consider a circuit $C$ and two subset of gate indices $A,B\subset \{1,\dots,|C|\}$ with $|A|=|B|$.  We say that the subsets $A$ and $B$ describe equivalent sub-circuits of $C$ if, and only if, there exists a bijective mapping $f:A \mapsto B$, such that for all $i \in A$ we have 
\begin{itemize}
\item $C_{i}=C_{f(i)}$ and 
\item the gates in the circuit $C$ can be commuted, such that in the resulting circuit  we have the gate with index $f(i)$ at the positions $i$.
\end{itemize}
\end{definition}

In other words, two sub-circuits are equivalent if one can swap commuting gates to replace the gates from one sub-circuit with the same gates from the other sub-circuit.

\begin{definition}[Equivalence of pattern matches]
For a match $Q$, let us denote the set of matched indices in the pattern by $Q^T:=\{i:(i,j) \in Q $ for some $j \in \{1,\dots,|Q|\}\}$ and the set of matched indices in the circuit by $Q^C:=\{ j:(i,j) \in Q $ for some $i \in \{1,\dots,|Q|\}\}$.
 We say that two matches $M$ and $\tilde{M}$ are equivalent if, and only if, 
\begin{itemize}
\item $M^T$ and $\tilde{M}^T$ describe equivalent subcircuits of the pattern $T$ and
\item  $M^C$ and $\tilde{M}^C$ describe equivalent subcircuits of the circuit $C$.
\end{itemize}
\end{definition}
%\RI{@David: Check carefully if well defined. I think that it is ok, since moving other gates is never necessary. Indeed, we are in the setting of a sub-circuit $(a,b,c,d...,a)$, hence the first gate $a$ can be commuted to the other gate $a$ if and only if it commutes with all the gates in between.}
We are now ready to state and prove the formal statement ensuring the correctness of \texttt{PatternMatch}. This ensures that Algorithm~\ref{algo:TempMatch} always succeeds, i.e., there are no situations where the algorithm does not deliver the desired output.
\begin{theorem}[Correctness of \texttt{PatternMatch}] \label{thm_correctness}
Given a circuit $C$ and a pattern $T$. Then Algorithm~\ref{algo:TempMatch} finds all maximal pattern matches (up to equivalent matches) of $T$ in $C$.
\end{theorem}
We note that not all the matches given as an output of Algorithm~\ref{algo:TempMatch} might be maximal and there may be equivalent matches in the output. We ignored this for simplicity of the pseudocode and since for certain applications, it might be more efficient to work with this output instead of removing the non-maximal and equivalent matches from it.
The proof of Theorem~\ref{thm_correctness} is given in Section~\ref{app_correctness}.

%%%%%%%%%%%%%%%%%%%%%%%%%%%%%%%%%%%%%%%%%%%%%%%%%%%%%%%%

\subsection{Complexity of the algorithm} \label{sec_complexity}
In the previous section we have shown that Algorithm~\ref{algo:TempMatch} is correct in the sense that it always finds all the maximal matches for any given pattern and any given circuit. Thus what remains to be understood is how efficiently Algorithm~\ref{algo:TempMatch} finds these matches. This is settled by the following theorem.

\begin{theorem}[Complexity of \texttt{PatternMatch}] \label{thm_Complexity}
The worst-case time complexity of Algorithm~\ref{algo:TempMatch} for a circuit $C$ and a pattern $T$ is

$\mathcal{O}\left( |C|^{|T|+3} |T|^{|T| + 4} \frac{n_C!}{(n_C - (n_T-1))!} \right)$. 
\end{theorem}
We note that the running time stated in Theorem~\ref{thm_Complexity} may be simplified as
\begin{align}
\mathcal{O}\left(  |C|^{|T|+3} |T|^{|T| + 4} \frac{n_C!}{(n_C - (n_T-1))!}  \right)
\leq \mathcal{O}\left(  |C|^{|T|+3} |T|^{|T| + 4}\, n_C^{n_T-1} \right)  \, .
\end{align}
From this we see immediately that Algorithm~\ref{algo:TempMatch} is efficient (i.e., polynomial) in $|C|$ and $|n_C|$ and inefficient (i.e., exponential) in $|T|$ and $|n_T|$. 

To prove the assertion of Theorem~\ref{thm_Complexity} we first need to understand the running time of the two subroutines \texttt{ForwardMatch} and \texttt{BackwardMatch}. This is done in the following lemmas, where we assume that $|T| \leq |C|$. Furthermore, we assume that we can check if two gates commute in constant time. For a fixed (finite) gate set, one possibility to achieve this is by storing the commutation relations between all the gates in a table. Since we will account time complexity $\mathcal{O}(|D|^3)$ to create the canonical form of a circuit $D$ in line~\ref{TempMatch_canonical_form_C} and line~\ref{TempMatch_canonical_form_T} in Algorithm~\ref{algo:TempMatch}, we can assume constant time access to an ordered list of successors or predecessors of all vertices in the canonical form (see Remark~\ref{rmk:constant_access_to_successors}).

\begin{lemma}[Complexity of \texttt{ForwardMatch}] \label{lem:ComplexityForwardMatch}
The worst-case time complexity of Algorithm~\ref{algo:ForwardMatch} for a circuit $C$ and a pattern $T$ is $\mathcal{O}(|T|^2|C|^2)$, under the assumption that we have constant time access to an ordered list of all successors for any vertex in the canonical forms $G^C$ and $G^T$ of the circuit $C$ and the pattern $T$, respectively. 
\end{lemma}
\begin{proof}
The while-loop runs at most $|T||C|$ times, since there are at most $|T|$ vertices that can be matched and added to \textit{MatchedVertexList} and for each vertex, there are at most $|C|$ successors to visit. 
The computationally most expensive parts of the algorithm are (i) the insertion of the vertex $v_0$ into the list \textit{MatchedVertexList} in line~\ref{ForwardMatch_insert_v0}; 
(ii)  finding the candidates for further matches by running   \textit{FindForwardCandidates} in line~\ref{ForwardMatch_find_candidates}; 
and (iii) the code in the if-condition starting in line~\ref{ForwardMatch_if}.

The insertion in (i) has worst-case time complexity $\mathcal{O}(\log |T|)$ (since the attribute \textit{SuccessorsToVisit} is an ordered list) and can occur at most $|T||C|$ times in the while-loop, hence it adds a term  $\mathcal{O}(|T| |C| \log |T|)$ to the complexity of the complete algorithm.

Finding the candidates in (ii) has worst-case time complexity $\mathcal{O}(|T|
^2)$, since we have two loops that each run at most $|T|$ times and we have at most $|T|$ successors to consider (and we have constant time access to a list that contains all the successors). Since \textit{FindForwardCandidates} is called at most once per run of the while-loop, it adds a complexity of $\mathcal{O}(|T|^3|C|)$ to the complete algorithm.

Let us now analyze the complexity of (iii) by considering the two cases of the if-condition in line~\ref{ForwardMatch_if} in Algorithm~\ref{algo:ForwardMatch} separately. 
In the case where we found a match (see line~\ref{ForwardMatch_matching_case}), the most expensive part is to create the list appearing in line~\ref{forward_match_set_to_order} that contains at most $|C|$ direct successors that we have left to visit from the point of view of the vertex $v$. Since we have constant time access to an ordered list of all successors, we also have constant time access to an ordered list of  all direct successors (since they must appear in the first part of the ordered list of all successors). Hence, the complexity of creating the list is $\mathcal{O}(|C|)$. Since the case where we found a match can occur at most $|T|$ times, the complexity added by this case  to the complete algorithm is $\mathcal{O}(|T||C|)$. 
In the case where we cannot match (see line~\ref{ForwardMatch_no_matching_case}), we have to block the vertex $v$ and all of its successors. Since we assume constant time access to a list containing all successors of $v$, this can be done in time $\mathcal{O}(|C|)$. The case where we cannot match can occur at most as many times as we have to run the while-loop. Hence, the complexity added by this case  to the complete algorithm is  $\mathcal{O}(|T||C|^2)$.

%The method  \textit{MatchedVertexList}.\texttt{insert} has worst-case time complexity $\mathcal{O}(|T|)$ (since the attribute \textit{SuccessorsToVisit} is an ordered list). Finding the candidates for further matches by running   \textit{FindForwardCandidates} has worst time complexity $\mathcal{O}(|T|
%^2)$, since we have two loops that each run at most $|T|$ times and we have at most $|T|$ successors to consider (and we have constant time access to a list that contains all the successors). 

We conclude that the worst-case complexity of \texttt{ForwardMatch} is given by $\mathcal{O}(|T||C| \log |T| + |T|^3|C| +|T||C| +|T||C|^2) \leq \mathcal{O}(|T|^2|C|^2) $, where we assumed $|T| \leq |C|$.
\end{proof}

\begin{lemma}[Complexity of \texttt{BackwardMatch}] \label{lem:ComplexityBackwardMatch}
The worst-case time complexity of Algorithm~\ref{algo:BackwardMatch} for a circuit $C$ and a pattern $T$ is $\mathcal{O}\left( |T|^{|T| + 3}  |C|^{|T|+2} \right)$, under the assumption that we have constant time access to an ordered list of all successors and predecessors for any vertex in the canonical forms $G^C$ and $G^T$ of the circuit $C$ and pattern, respectively. 
\end{lemma}
\begin{proof}
To estimate the complexity of \cref{algo:BackwardMatch}, we first need to upper-bound the number of times the while-loop in line~\ref{algo_while_loop} is executed, i.e., how many different matching scenarios are added to the stack \textit{MatchingScenarios}.
It is easiest to visualize the matching scenarios as a tree: when a particular matching scenario, i.e., a vertex $v$ in this tree, is considered in line~\ref{algo_while_loop}, then a number of options (e.g. matching with different gates, left-blocking, right-blocking) are considered and potentially added to the stack \textit{MatchingScenarios}.
These form the children of the vertex $v$ in the tree of options.
It is easy to see that every vertex in this tree is only ever considered once in the while-loop, so it suffices to upper-bound the number of vertices in this tree. 

Each vertex in the tree can be specified by the tuple $(s_1, \dots, s_n)$ of options that were taken to arrive at this vertex.
Hence, to count the number of vertices, we can count the number of such tuples.
The maximum length of such a tuple is $|C|$, since for each vertex in the option tree, its children had the variable $\textit{counter}$ incremented by 1, and $\textit{counter}$ cannot exceed $|C|$ by line~\ref{algo_counter_check} in \cref{algo:BackwardMatch}.
For each $s_i$, we distinguish two types: the \emph{trivial type}, where we block the gate under consideration without blocking any gates matched during \texttt{ForwardMatch} (options 1.2a, 2.1a or options 1.2c, 2.1c in \cref{algo:BackwardMatch});
and the \emph{non-trivial type}, where we either match the gate, or block it and also need to block successors that were matched during \texttt{ForwardMatch} (options 1.1, 1.2b or 2.1b in \cref{algo:BackwardMatch}).\footnote{Note that line~\ref{BackwardMatch:adding_to_stack_condition_1} ensures that we do not consider options where we would block gates that were previously matched during \texttt{BackwardMatch}.}

We now argue that any tuple $(s_1, \dots, s_n)$ can contain at most $|T|$ non-trivial $s_i$. 
For this, recall that if we match a gate during \texttt{BackwardMatch}, then this match is fixed and cannot be blocked subsequently. On the other hand, if we block a vertex (and its successors) that was matched during \texttt{ForwardMatch}, then this vertex will never be matched again because \texttt{BackwardMatch} does not consider gates in the forward direction for matching.
Since there are $|T|$ possible gates to match, we can match a gate or permanently block a previously matched gate at most $|T|$ times.
Hence, there can be at most $|T|$ non-trivial $s_i$.

With this, we can count the number of possible tuples $(s_1, \dots, s_n)$. 
Consider a fixed $n$ and a fixed number $k$ of non-trivial $s_i$.
We have $n \choose k$ ways of placing the non-trivial $s_i$ in the tuple $(s_1, \dots, s_n)$.
For each of the $k$ non-trivial $s_i$, there are $|T|+1$ possible scenarios that can be added to the stack $\textit{MatchingScenarios}$: 
we could match with at most $|T|$ different gates (option 1.1), or we could right-block (one of the mutually exclusive options 1.2b and 2.1b).
Having fixed choices for the non-trivial $s_i$, there is no further freedom in choosing the trivial $s_i$. This is because for trivial $s_i$, in the case where something can be matched (option 1), \cref{algo:BackwardMatch} always chooses option 1.2a (in \cref{algo:BackwardMatch}) or, if 1.2a cannot be chosen, option 1.2c. 
Analogously, if nothing can be matched (option 2), \cref{algo:BackwardMatch} always chooses option 2.1a (in \cref{algo:BackwardMatch}) or, if 2.1a cannot be chosen, option 2.1c.
Hence, for a fixed choice of circuit, pattern, $n$, and $k$, the non-trivial $s_i$ uniquely determine the trivial ones.

The above reasoning means that for fixed $n$, $k$, there are 
\begin{equation}
{n \choose k} (|T|+1)^{k}
\end{equation}
possible tuples $(s_1, \dots, s_n)$. Summing over the possible values of $n$ and $k$, we obtain the total number of vertices in the option tree (with ${n \choose k} = 0$ for $k > n$):
\begin{align*}
\sum_{n = 1}^{|C|} \sum_{k = 1}^{|T|} {n \choose k} (|T| + 1)^{k} 
&\leq (|T| + 1)^{|T|} |C| \sum_{k = 1}^{|T|} {|C| \choose k} \\
&\leq (|T| + 1)^{|T|} |C| |T| |C|^{|T|} = \mathcal{O}\left( (|T| |C| )^{|T| + 1} \right) \,.
\end{align*}

The complexity of a single round of the while-loop is at most $\mathcal{O} \left( |T|^2 + |C| \right)$, where the leading contributions come from the removal of equivalent matches in line~\ref{algo_index_pruning}\footnote{To see why this requires time $|T|^2$, note that we can first sort the list $I$. Then to check $[i, j]_T = 0$ for fixed $i$ and all $j \in I$, we only need to compare the sorted list of successors of the vertex with label $i$ with the sorted list $I$, which requires time linear in the length of the lists (i.e., at most $\mathcal{O}(|T|)$). Doing this for all $i$, we get a total complexity of $\mathcal{O}(|T|^2)$.} and the blocking of successors in the graph $G^C$ in line~\ref{algo:line_blocking}.
Therefore, the worst-case complexity of \texttt{BackwardMatch} is (somewhat loosely) upper-bounded by
\begin{equation}
\mathcal{O} \left( |T|^{|T| + 3}  |C|^{|T|+2} \right) \,.
\end{equation}

\end{proof}

\begin{proof}[Proof of  Theorem~\ref{thm_Complexity}]
The assertion of Theorem~\ref{thm_Complexity} now follows from Lemma~\ref{lem:ComplexityForwardMatch}  and Lemma~\ref{lem:ComplexityBackwardMatch} and inspection of the loop structure in \texttt{PatternMatch}.   
\end{proof}

%The overall complexity would be at most $\mathcal O(d_C^2  \log(d_C)  n_C^{n_T}  d_T^3 \,(n_T!))$ for matching a template $T$ in a circuit $C$. For a fixed size of the templates, we would hence get $\mathcal O(d_C^2  \log(d_C)  n_C^{n_T})$. However, I would expect this to reduce to a much lower complexity under some reasonable assumptions on the circuit structure and after taking care about some optimizations. In practice, one might be able to get the scaling $\mathcal O(d_C n_C \, \log(d_T)  d_T^4 n_T (n_T!))$. \RI{This is very rough at the moment... and I have to check all of these things carefully!}
%%The factor $n_T !$ could be avoided in practice by using the \textit{longest sub-string algorithm} to find good candidates for a qubit labeling with a long match. 
%For a fixed template size, one would then get a complexity  $\mathcal O(d_C n_C )$.

\section{Heuristics} \label{sec:heuristics_full}
In this section, we provide a more detailed description of the heuristics mentioned in \cref{sec:heuristics_intro}.
These heuristics allow us to further speed up our algorithm, but using them might lead the algorithm to miss some maximal matches.
Depending on the size of the quantum circuit and the amount of classical computation time available to optimize it, one may choose to either use the exact algorithm, or to use the heuristics by setting ``quality parameters'' (controlling the tradeoff between faster runtimes and less missed matches) to the desired values.

% Since the worst-case complexity of the algorithm can be a high degree polynomial in the circuit size and in the number of qubits in the circuit for large templates (see Theorem~\ref{thm_Complexity}), we introduced some heuristics in Section~\ref{sec:heuristics_intro} to lower the runtime in trade-off with finding all the matches. 

\paragraph{Heuristics for the choice of the qubits.} 
In the exact algorithm \texttt{PatternMatch}, we loop over all possible assignments of pattern qubits to circuit qubits.
If we were able to choose, for a given starting gate, the correct qubits $L_{\textnormal{q}}^{\textnormal{sel}}$ (out of the $n_C$ circuit qubits) that lead to a maximal match, we would save a lot of runtime. 
A simple heuristic is to consider $F$ additional gates around the starting gate. If the number of successors (in the canonical graph) of the starting gates is more than half of the pattern size, then we consider the gates corresponding to the first $F$ successors and add the qubits these gates act on to $L_{\textnormal{q}}^{\textnormal{sel}}$. 
Otherwise, we pick the $F$ gates with the largest label that are not successors of the starting gate and use these to add qubits to $L_{\textnormal{q}}^{\textnormal{sel}}$ accordingly.
If the maximal match indeed contains all the qubits the $F$ explored gates act on, 
%$F$ gates following the starting gate, then 
we will still find the maximal match; if the maximal match leaves some of these qubits out, we will miss it.
For large enough $F$, such heuristics essentially reduce the term $n_C^{n_T-1}$ in the worst case complexity given in \cref{eq_runntime} to a constant.
The pseudocode for this heuristic is given in Algorithm~\ref{algo:HeuristicsQubit} and based on the two subroutines Algorithm~\ref{algo:ExploreCircuitForward} and Algorithm~\ref{algo:ExploreCircuitBackward}. We define the method \texttt{Qubits}$(v)$, that takes a vertex $v$ of a canonical form as input and returns the qubit labels that the gate corresponding to $v$ acts on. The output of the Algorithm~\ref{algo:HeuristicsQubit} may fix less qubits than the number of qubits $n_T$ contained in the pattern. In this case, we still loop over all possible choices for the remaining qubits in \texttt{PatternMatch}. Finally, we loop trough all possible permutations of the qubits. 

\paragraph{Heuristics for \texttt{BackwardMatch}.}
Another major contribution to the runtime of our algorithm comes from having to consider all possibilities in the tree of matching scenarios. 
We can use the following heuristic to reduce the number of branches that we consider: we  evolve all matching scenarios for $L$ steps, and then prune the tree by choosing only the $S$ branches that have matched the most gates so far for further consideration. This process is repeated until the algorithm terminates. Such a heuristic reduces the worst-case complexity of  \texttt{BackwardMatch} from $\mathcal{O}\left( |T|^{|T| + 3}  |C|^{|T|+2} \right)$ to $\mathcal{O}\left( |C|^{2}  \right)$ for constant $L$ and $S$. 
% \tm{L and S don't appear in this? for const L, S, presumably? do we know the scaling in L, S} \RI{Indeed, I added a sentence. Further, I have changed $\mathcal{O}\left( |C|^{2} n_C
% ^{n_T-1} \right)$ to $\mathcal{O}\left( |C|^{2} \right)$. @Tony: Have you added the term $n_C
% ^{n_T-1}$? }
This procedure can be easily integrated into \cref{algo:BackwardMatch}, and we omit the (rather lengthy) pseudocode for readability. 

The quality parameters controlling the tradeoff between runtime and output quality are $F$ for the qubit assignment heuristic, and $L, S$ for the heuristic to speed up \texttt{BackwardMatch}.

\begin{algorithm}[H] 
\caption{\texttt{HeuristicsQubits}: Constrain the qubits configurations to explore}
  \label{algo:HeuristicsQubit}
   \begin{algorithmic}[1]
   \State Input: $(G^C, G^T, n_T, r,i, L)$
   \begin{itemize}
       \item Canonical form $G^C$ of the circuit 
       \item Canonical form $G^T$ of the pattern
       \item Total number of qubits $n_T$ in the pattern
 \item Gate indices $r$ in $C$ and $i$ in $T$ (where we start matching, i.e., the first match is $G^C_r=G^T_i$)
       \item Length $L$ to explore 
   \end{itemize}
   \If{$|\texttt{Succ}(G^T_i,G^T)| \geq \frac{1}{2}\left( |T|-i+1 \right)$}
   \State $L_{\textnormal{q}}^{\textnormal{heur}} \leftarrow \texttt{ExploreCircuitForward}(G^C, G^T,n_T,r, L,j)$
   \Else
   \State $L_{\textnormal{q}}^{\textnormal{heur}}  \leftarrow \texttt{ExploreCircuitBackward}(G^C, G^T, n_T, r, L,j)$
   \EndIf
   \State Ouput: $L_{\textnormal{q}}^{\textnormal{heur}} $
 \end{algorithmic}
\end{algorithm}

\begin{algorithm}[H] 
\caption{\texttt{ExploreCircuitForward}: Constrain the qubits configurations by exploring in forward direction}
  \label{algo:ExploreCircuitForward}
   \begin{algorithmic}[1]
   \State Input: $(G^C, n_T, r, L)$
   \begin{itemize}
       \item Canonical form $G^C$ for the circuit 
       \item Total number of qubits $n_T$ in the pattern
       \item Start label $r$ in the circuit
       \item Length $L$ to explore 
   \end{itemize}
   \State j=1
    \State $L_{\textnormal{q}}^{\textnormal{forw}}=\texttt{Qubits}(G^C_r)$
    \State $I =$ List of labels of vertices in $\texttt{Succ}(G^C_r,G^C)$ ordered in increasing order
    \While{$|L_{\textnormal{q}}^{\textnormal{forw}} \cup  \texttt{Qubits}(G^C_{I_j})|  \leq n_T$ and $j \leq L$ and $j \leq |I|$}
    %\For{$v \in \texttt{Succ}(G^C_r,G^C)$} \Comment{We loop in the order according to the labels of the vertices.}
    %\If{$|L_{\textnormal{q}}^{\textnormal{forw}} \cup  \texttt{Qubits}(v)|  \leq n_T$ and $j \leq L$}
    \State{$L_{\textnormal{q}}^{\textnormal{forw}} \leftarrow L_{\textnormal{q}}^{\textnormal{forw}} \cup \texttt{Qubits}(G^C_{I_j})$}
    \State $j=j+1$
    %\Else
    %\State $break$
    %\EndIf
    \EndWhile
    \State Output: $L_{\textnormal{q}}^{\textnormal{forw}}$
 \end{algorithmic}
\end{algorithm}

\begin{algorithm}[H] 
\caption{\texttt{ExploreCircuitBackward}: Constrain the qubits configurations by exploring in backwards direction}
  \label{algo:ExploreCircuitBackward}
  \begin{algorithmic}[1]
   \State Input: $(G^C, n_T, r, L)$
   \begin{itemize}
       \item Canonical form $G^C$ for the circuit 
       \item Total number of qubits $n_T$ in the pattern
       \item Start label $r$ in the circuit
       \item Length $L$ to explore 
   \end{itemize}
   \State j=1
   \State $L_{\textnormal{q}}^{\textnormal{back}}=\texttt{Qubits}(G^C_r)$
    \State $I =$ List of labels of vertices in $\{G^C_{|C|},\dots,G^C_1\} \setminus \texttt{Succ}(G^C_r,G^C)$ ordered in decreasing order
    \While{$|L_{\textnormal{q}}^{\textnormal{back}} \cup \texttt{Qubits}(G^C_{I_j})|  \leq n_T$ and $j \leq L$ and $j \leq |I|$}
    %\For{$v \in \texttt{Succ}(G^C_r,G^C)$} \Comment{We loop in the order according to the labels of the vertices.}
    %\If{$|L_{\textnormal{q}}^{\textnormal{forw}} \cup  \texttt{Qubits}(v)|  \leq n_T$ and $j \leq L$}
    \State{$L_{\textnormal{q}}^{\textnormal{back}} \leftarrow L_{\textnormal{q}}^{\textnormal{back}} \cup \texttt{Qubits}(G^C_{I_j})$}
    \State $j=j+1$
    %\Else
    %\State $break$
    %\EndIf
    \EndWhile
    \State Output: $L_{\textnormal{q}}^{\textnormal{back}}$
 \end{algorithmic}
\end{algorithm}

\section{Numerical analysis} \label{sec:numerics_full}
In this section, we investigate the numerical scaling of our algorithm in more detail. All numerical experiments were implemented in Python and run on an Intel Core i7-9700K (3.60 GHz) processor with 2x16GB DDR4 RAM (on Ubuntu). The numerical performance could be improved by implementing the algorithm in a different program language such as C++, and by parallelizing it, which is straightforward since many of the loops in our algorithm can be performed in parallel. For example, one could easily try all qubit assignments in parallel, rather than in sequence as in our implementation. We leave a high-performance implementation for practical applications as future work.

\subsection{Scaling for creating the canonical form}
\label{sec:numerics_canonical_form}

Our work makes extensive use of the canonical form of quantum circuits, explained in Section~\ref{sec:algo_intro} and introduced in~\cite{Rahman:2014:AQT:2711453.2629537}.
The creation of the canonical form can be separated from our main algorithm in practice (because e.g. the canonical form could be created once, stored, and used many times). 
Therefore, the numerical scaling results in the following sections assume that the canonical form is given as an input.
In this section, we numerically investigate how much time it takes to actually create the canonical form itself.
While the worst-case time complexity to create the canonical form of a circuit $C$ using Algorithm~\ref{algo:CreateCanonicalForm} is $\mathcal{O}(|C|^3)$ (if we also store an ordered list containing all successors and predecessors for each vertex), our numerics suggests that the practical scaling for random circuits is close to linear (see Figure~\ref{fig:runtime_canonical_form}) and does not make a dominant contribution to the total runtime of our algorithm.

\begin{figure}[t]
\centering
\includegraphics[width=0.55\textwidth]{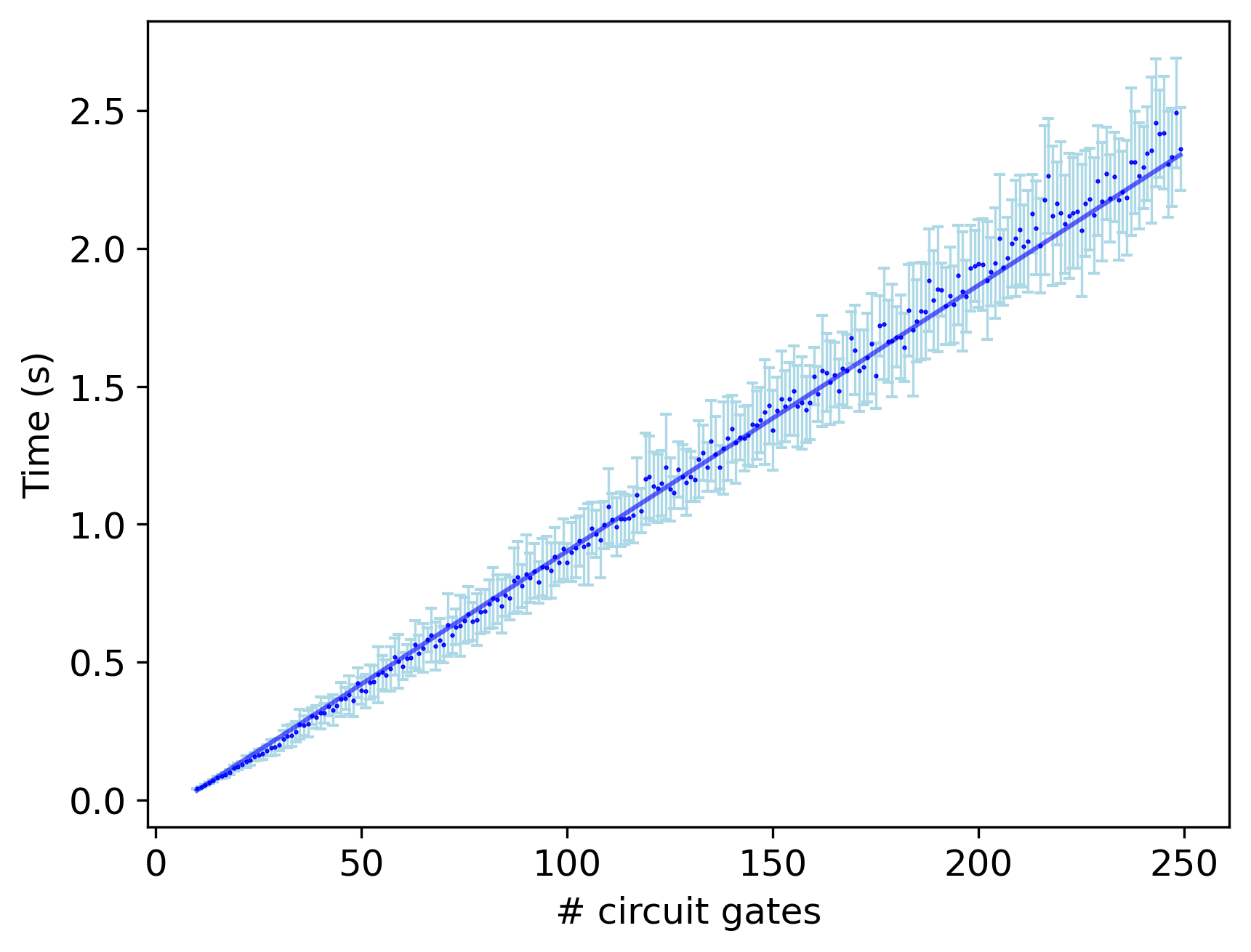}
\caption{We analyse the runtime of Algorithm~\ref{algo:CreateCanonicalForm} that creates the canonical form of a circuit together with an ordered list of all successors and predecessors for each vertex (see Remark~\ref{rmk:constant_access_to_successors}). For each gate count $|C|$, we run Algorithm~\ref{algo:CreateCanonicalForm} on 15 random circuits and calculate the mean and standard deviation of the runtimes, plotted in the figure. We use the same random circuits as for Figure~\ref{fig:comparison_2_mean} consisting of X, \cnot{} and Toffoli gates chosen uniformly at random.}
\label{fig:runtime_canonical_form}
\end{figure}

\subsection{Scaling with the number of gates in the circuit}
The scaling of our algorithm \texttt{PatternMatch} in the number of gates $|C|$ in the circuit was already discussed in Section~\ref{sec:numerics_intro} for the pattern given in Figure~\ref{fig:pattern}. Here, we give a second example, with a different pattern shown in Figure~\ref{fig:pattern_2}. Since only the last \cnot{} gate commutes with the other gates in this pattern, we expect that the subroutine \texttt{ForwardMatch} will be able to match most gates and leave little work to do for the less efficient subroutine \texttt{BackwardMatch}. 
Indeed, we find that using heuristics for the \texttt{BackwardMatch} (Section
~\ref{sec:heuristics_full}) has essentially no effect on the runtime. 
The runtimes for running \texttt{PatternMatch} without heuristics and with heuristics for the qubit choice with $F=1$ are shown in Figure~\ref{fig:numerics_scaling_C_mean_2}. 
The runtimes are lower than the ones shown in Figure~\ref{fig:comparison_2_mean}, mainly since \texttt{ForwardMatch} is more efficient than \texttt{BackwardMatch}. Further, the loss of matches due to the heuristic to choose the qubits is lower than for the pattern given in Figure~\ref{fig:pattern}, since for the given pattern, guessing which qubits to choose is easier due to the fact that nearly no gates commute.

\begin{figure}[H]
\begin{subfigure}[b]{0.3\textwidth}
\centering
\scalebox{0.85}{
$$
    \Qcircuit @C=1.0em @R=0.2em @!R {
	 	 & \qw & \ctrl{1} & \qw & \ctrl{1} & \ctrl{2} & \qw & \qw\\
	  & \ctrl{1} & \targ & \ctrl{1} & \targ & \qw & \qw & \qw\\
	 	 & \targ & \qw & \targ & \qw & \targ & \qw & \qw\\
	 }
$$
}
\subcaption{Pattern}
\label{fig:pattern_2}
\end{subfigure}
\hfill
\begin{subfigure}[b]{0.65\textwidth}
\centering
\includegraphics[width=1\textwidth]{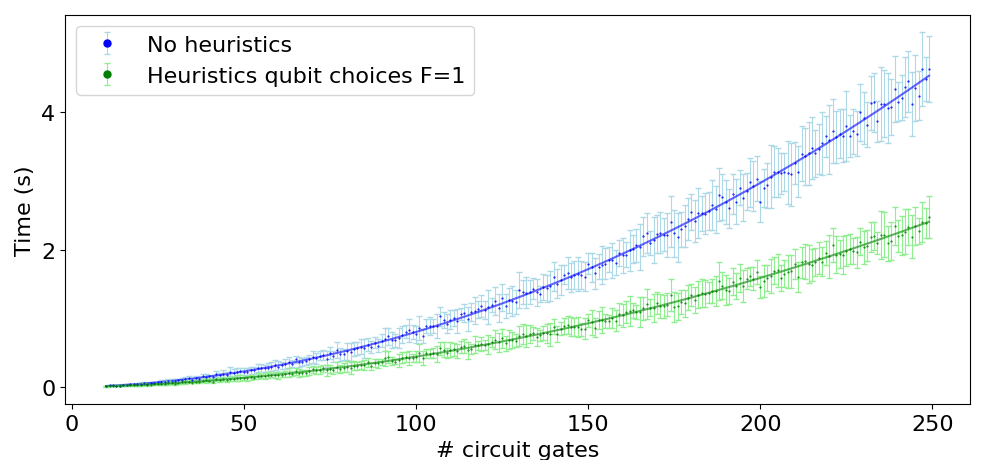}
\subcaption{Runtimes dependent on circuit size $|C|$}
\label{fig:numerics_scaling_C_mean_2}
\end{subfigure}
\caption{ The figure shows the runtimes of \texttt{PatternMatch} for finding maximal matches of the pattern~\ref{fig:pattern_2} ($n_T=3$, $|T|=5$) in randomly generated circuits with $|C| \in [10,250]$ gates on $n_C=6$ qubits consisting of X, \cnot{} and Toffoli gates. For each gate count, 15 random quantum circuits are generated, and the mean and the standard deviation are plotted. A polynomial fit of the means yields $6.8 \cdot 10
^{-5}(|C|^2 + 18.38|C| - 130)$. If we include a cubic term in the fit, the ratio of the cubic to the quadratic coefficient is 1/4300, so the cubic term can be safely dropped. Further, we run our algorithm with heuristics for the qubit choice, which leads to lower runtimes (but similar scaling) and causes the algorithm to miss $\sim$26\% of matches longer than half the pattern.}
\label{fig:numerics_scaling_C_2}
\end{figure}

\subsection{Scaling with the number of qubits in the circuit} \label{sec:scaling_qubit_number}

The worst-case time complexity given in \cref{eq_runntime} scales with the number of qubits as $n_C^{n_T-1}$, where $n_C$ and $n_T$ are the number of qubits in the circuit and the pattern, respectively. This term arises due to the fact that the algorithm \texttt{PatternMatch} loops trough all possible assignments of qubits in the pattern to qubits in the circuit. The $``-1''$  arises due to the fact that at least one qubit is fixed by matching the starting gate. More generally, we have a worst-case time complexity of $n_C^{n_T-q}$, where $q$ denotes the minimal number of qubits that are fixed by the starting match. Since the algorithm \texttt{PatternMatch} always loops trough all these assignments, we expect to see the worst-case complexity also in the practical implementation of the algorithm. 
For example, for the pattern given in Figure~\ref{fig:template_scaling_qubit}\,, we have $n_T=4$ and since there are starting matches with single-qubit gates, we have $q=1$. Hence we expect cubic scaling in the number of qubits, which is indeed the case in Figure~\ref{fig:comparison_2_mean_qubits}.

We also investigate the effect of the heuristics for choosing the qubit assignment with different exploration lengths $F=1,2,3$ (Section~\ref{sec:heuristics_full}). As shown in Figure~\ref{fig:comparison_2_mean_qubits}, the runtime improves for $F = 1$ and $F = 2$, but $F=3$ brings hardly any further speedup.
This is due to the structure of the pattern: for each starting match, there are at most two gates to explore in the forward as well as in the backward direction. Since matches are particularly useful if more than half of the pattern is matched (since such matches can be used to reduce the gate count with pattern matching if the pattern implements the identity), we consider only the number of missed matches that would have been  at least $|T|/2$ long. For the pattern given in Figure~\ref{fig:template_scaling_qubit}, we find that for $F=1$, the mean proportion of missed matches is around $45\%$, and for $F>1$ around $72\%$. 
For example, for $F=1$ and $n_C=25$, there is a $90\%$ decrease for the runtime, and an average of $45\%$ of matches are missed. This means that we still find $55\%$ of the matches (longer than half the size of the pattern) within 10\% of the runtime, which might be a good tradeoff for larger-scale applications.

\begin{figure}[H]
\centering
\begin{subfigure}[b]{0.2\textwidth}
\centering
\begin{equation*}
    \Qcircuit @C=0.5em @R=0.0em @!R {
	  & \targ & \qw & \qw & \qw \\
	 	 & \ctrl{-1} & \gate{X} & \ctrl{1} & \qw \\
	  & \gate{X} & \ctrl{1} & \targ & \qw \\
	 	 & \qw & \targ & \qw & \qw \\
	 }
\end{equation*}
\subcaption{Pattern}
\label{fig:template_scaling_qubit}
\end{subfigure}
\centering
\begin{subfigure}[b]{0.75\textwidth}
\centering
\includegraphics[width=1\textwidth]{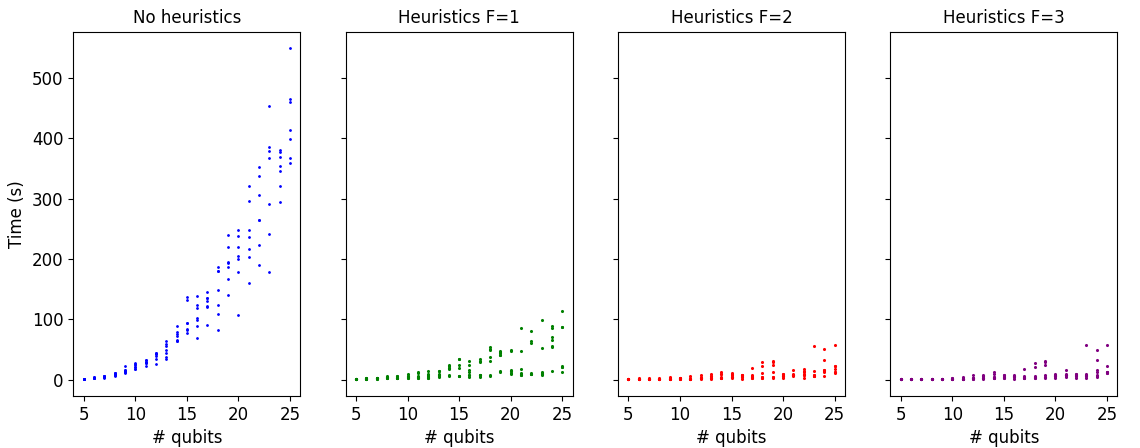}
\subcaption{Scaling with the number of qubits $n_C$}
\label{fig:comparison_2_mean_qubits}
\end{subfigure}
\caption{The figure shows the time scatter plots for running the pattern matching algorithm with the pattern shown in (a) on random circuits with $|C|=50$ gates (X, \cnot{} and Toffoli gate chosen uniformly at random) with increasing number of qubits $n_C$ in the circuit (for a fixed number of gates, i.e., the circuit becomes wider and shallower). We run the algorithm for 7 random circuits for each choice of $n_C$. Further, the results for different choices of the parameter $F=1,2,3$ for the heuristics to choose the qubits (see Section~\ref{sec:heuristics_full}) are shown. 
For $F=1$ the mean proportion of missed matches larger than half the pattern size is around $45\%$ and for $F>1$ roughly constant around $72\%$.
Fitting a cubic polynomial to  the mean values of the runtimes without heuristics yields $2.56\cdot 10^{-2}\left(n_C^3+2.8 n_C^2-59 n_C+146\right)$. Including a degree-4 term, we find that the factor between the quartic and the cubic coefficient is around 166, suggesting that the scaling in the number of qubits is approximately cubic. T
his is expected from the worst-case time complexity $n_c ^{n_T-1}$, since we have $n_T=4$ here.}
%\rmo{mean loss L=1 45\%, L=2 71 \%, L=3 72\%}
\label{fig:qubit_number_time_scatter}
\end{figure}

\begin{comment}
\begin{figure}[H]
\centering
\begin{equation*}
    \Qcircuit @C=1.0em @R=0.0em @!R {
	 	\lstick{ \textnormal{qubit 1} :  } & \targ & \qw & \qw & \qw \\
	 	\lstick{ \textnormal{qubit 2} :  } & \ctrl{-1} & \gate{X} & \ctrl{1} & \qw \\
	 	\lstick{ \textnormal{qubit 3} :  } & \gate{X} & \ctrl{1} & \targ & \qw \\
	 	\lstick{ \textnormal{qubit 4} :  } & \qw & \targ & \qw & \qw \\
	 }
\end{equation*}
\caption{This Figure shows the fixed pattern that is used for numerically analyze the scaling of the algorithm with the number of qubits in the circuit.}
\label{fig:template_scaling_qubit}
\end{figure}

\begin{figure}[H]
\centering
\includegraphics[width=0.75\textwidth]{figures/high_q_time.png}
\caption{This figure shows the time scatter plots for running the pattern matching algorithm with increasing number of qubits in the circuit. The pattern is from Figure~\ref{fig:template_scaling_qubit}. It shows results for different qubit choices heuristics $L=1,~L=2$ and $L=3$.}
\label{fig:qubit_number_time_scatter}
\end{figure}
\end{comment}

\subsection{Circuit optimization} \label{sec:circuit_opt}
Pattern matching is a useful tool for different optimization techniques for quantum circuits. In this section, we demonstrate two such techniques: template matching and peephole optimization. There are other useful applications for pattern matching: for example, we could search for a pattern that appears often in a larger circuit, optimize the pattern for a specific hardware architecture with brute-force methods, and replace the found instances of the pattern with the optimized version. 

\subsubsection{Template matching}
Template matching is explained in Section~\ref{sec:prev_work}. The results of applying it for the optimization of \emph{random} circuits are shown in Table~\ref{tab:optimization_random}. Here we apply template matching to benchmark circuits that represent relevant computations for quantum algorithms and that have already been heavily optimized using state-of-the-art techniques~\cite{Nam_2018}. The results are summarized in Table~\ref{tab:optimization_bench} and show that we can further optimize some of these circuits with the application of a single simple template shown in Figure
~\ref{fig:template_first_type}.
This shows the potential that using an \emph{exact} pattern matching algorithm that finds all matches of a pattern has for optimizing quantum circuits beyond previous state-of-the-art methods.

\begin{figure}[H]
\centering
\begin{equation*}
    \Qcircuit @C=1.0em @R=0.2em @!R {
	 	\lstick{ \textnormal{qubit 1} :  } & \qw & \ctrl{1} & \qw & \ctrl{1} & \ctrl{2} & \qw & \qw\\
	 	\lstick{ \textnormal{qubit 2} :  } & \ctrl{1} & \targ & \ctrl{1} & \targ & \qw & \qw & \qw\\
	 	\lstick{ \textnormal{qubit 3} :  } & \targ & \qw & \targ & \qw & \targ & \qw & \qw\\
	 }
\end{equation*}
\caption{Template $T$ that is used to optimize the benchmark circuits optimized in~\cite{Nam_2018}. We use this template because it is especially suited to reduce the \cnot{} count in a circuit, which is important because the \cnot{} gate is usually more challenging to implement experimentally than single-qubit gates. Further, the benchmark circuits were not already optimized with this template. }
\label{fig:template_first_type}
\end{figure}

\begin{table}[H]
\centering
\scalebox{0.89}{
\begin{tabular}{c|ccc|ccc}
\toprule
{} &\multicolumn{3}{c|}{Before template matching \cite{Nam_2018}} & \multicolumn{3}{c}{After template matching}\\ \midrule
Circuit name & \# gates in circuit $|C|$ & \# \cnot{} & \# qubits $n_C$ & \# gates in circuit $|C|$ & \# \cnot{} & time (s)\\ \toprule
$\text{Mod5}_4$ & 51 & 28 & 5 & 49 & 26 & 0.41  \\
$\text{VBE-Adder}_3$ & 89 & 50 & 10 & 84 & 45 & 3.26 \\
$\text{CSLA-MUX}_3$ & 155 & 70 & 15 & 153 & 68 & 17.45  \\
$\text{QCLA-Com}_7$ & 284 & 132 & 24 & 280 & 128 & 71.30 \\
$\text{QCLA-Mod}_7$ & 624 & 292 & 26 & 614 & 282 & 341.43 \\
$\text{QCLA-Adder}_{10}$ & 399 & 183 & 36 & 393 & 177 & 269.30 \\
$\text{Adder}_8$ & 606 & 291 & 24 & 598 & 285 & 384.22\\
$\text{GF}(2^{4})\text{-Mult}$ & 187 & 99 & 12 & 183 & 95 & 18.88  \\
$\text{GF}(2^5)\text{-Mult}$ & 296 & 154 & 15 & 289 & 147 & 55.83 \\
$\text{GF}(2^6)\text{-Mult}$ & 403 & 221 & 18 &  392 & 210 & 124.52 \\
$\text{GF}(2^7)\text{-Mult}$ & 555 & 300 & 21 & 539 & 284 & 282.95 \\
$\text{GF}(2^8)\text{-Mult}$ & 712 & 405 & 24 & 687 & 380 & 508.21 \\
$\text{GF}(2^9)\text{-Mult}$ & 891 & 494 & 27 & 862 & 465 & 930.93\\
$\text{GF}(2^{10})\text{-Mult}$ & 1070 & 609 & 29 & 1033 & 572 & 1487.34 \\
$\text{GF}(2^{16})\text{-Mult}$ & 2707 & 1581 & 48 & 2588 & 1462 & 3147\\ \midrule
$\text{CSUM-MUX}_9$ & 266 & 140 & 30 & 266 & 140 & 85.22\\
$\text{RC-Adder}_6$ & 140 & 71 & 14 & 140 & 71 & 7.65 \\
$\text{Mod-Red}_{21}$ & 180 & 77 & 11 & 180 & 77 & 7.11\\
$\text{Mod-Mult}_{55}$ & 91 & 40 & 9 & 91 & 40 & 2.41 \\
$\text{Toff-Barenco}_{3}$ & 40 & 18 & 5 & 40 & 18 & 0.19 \\
$\text{Toff-NC}_{3}$ & 35 & 14 & 5 &  35 & 14 & 0.21 \\
$\text{Toff-Barenco}_{4}$ & 72 & 34 & 7 & 72 & 34 & 0.67\\
$\text{Toff-NC}_{4}$ & 55 & 22 & 7 & 55 & 22 & 0.53\\
$\text{Toff-Barenco}_{5}$ & 104 & 50 & 9 & 104 & 50 & 1.76\\
$\text{Toff-NC}_{5}$ & 75 & 30 & 9 & 75 & 30 & 1.11\\
$\text{Toff-Barenco}_{10}$ & 264 & 130 & 19 & 264 & 30 & 26.72  \\
$\text{Toff-NC}_{10}$ & 175 & 70 & 19 & 175 & 70 & 10.86  \\
\bottomrule
\end{tabular}
}
\caption{The table shows the results of applying template matching based on $\texttt{PatternMatch}$ (without any heuristics) on benchmark circuits using only one template given in  Figure~\ref{fig:template_first_type}. The top half of the table shows circuits where applying template matching reduced the \cnot{} count. The second part of the table shows circuits where no reduction in \cnot{} count could be achieved.}
\label{tab:optimization_bench}
\end{table}
\newpage
\subsubsection{Peephole optimization} \label{sec:peephole}

\begin{figure}[t]
\captionsetup{font=footnotesize,labelfont=footnotesize, width=\textwidth}
   \begin{subfigure}[b]{0.45\textwidth}
    \centering
   \begin{equation*}
   \scalebox{0.75}{
    \Qcircuit @C=0.5em @R=1.2em @!R {
         &1&2&3&4&5&6&7&8&9&10&11&12\\   
	 	\lstick{ \textnormal{qubit 1} :  } & \qw & \qw & \qw & \qw & \qw & \ctrl{1} & \ctrl{2} & \qw & \qw & \qw & \qw & \qw & \qw & \qw\\
	 	\lstick{ \textnormal{qubit 2}  :  } & \qw & \ctrl{1} & \gate{H} & \ctrl{1} & \gate{H} & \targ & \qw & \qw & \targ & \qw & \targ & \gate{X} & \qw & \qw\\
	 	\lstick{ \textnormal{qubit 3}  :  } & \gate{X} & \targ & \qw & \targ & \qw & \qw & \targ & \gate{H} & \ctrl{-1} & \gate{H} & \ctrl{-1} & \qw & \qw & \qw
		\gategroup{3}{2}{4}{6}{0.6em}{--}
		\gategroup{3}{9}{4}{13}{0.6em}{--}
    	 }
    	 }
    \end{equation*}
	 \subcaption{Longest gate sequence found by adapted \texttt{PatternMatch}}
	 \label{fig:circuit_chain}
    \end{subfigure}
\begin{subfigure}[b]{0.45\textwidth}
    \centering
   \begin{equation*}
   \scalebox{0.7}{
    \Qcircuit @C=0.5em @R=1em @!R {
     &7&&&&&&&6\\   
	 	 & \ctrl{2} & \qw & \qw & \qw & \qw & \qw & \ctrl{1} & \qw & \qw\\
	 	 & \qw & \gate{U_3(\frac{\pi}{2},\frac{\pi}{2},\frac{3\pi}{2})} & \ctrl{1} & \gate{U_3(\frac{\pi}{2},\frac{\pi}{2},\pi)} & \ctrl{1} & \gate{U_3(\frac{\pi}{2},0,\frac{3\pi}{2})} & \targ & \qw & \qw\\
	 	 & \targ & \gate{U_3(\pi,0,\pi)} & \targ & \gate{U_3(\frac{\pi}{2},0,\frac{-3\pi}{2})} & \targ & \gate{U_3(0,2.38,2.33)} & \qw & \qw & \qw\\
	 }
	 }
	     \end{equation*}
	 \subcaption{Circuit (a) after applying peephole optimization}
	 \label{fig:chain_opt}
    \end{subfigure}
    \caption{\textbf{Peephole optimization.} 
    We use the adapted \texttt{PatternMatch} algorithm to find the longest (connected) sequences of gates on all subsets of two qubits in the circuit given in (a). The longest sequence found on qubits 2 and 3 is marked. Note that the marked gates can be commuted next to each other by commuting the gate $C_6$ to the end of the circuit and the gate $C_7$ to the start. Then, the sequence consisting of the marked gates is represented as a two-qubit unitary and optimally decomposed into a sequence of single-qubit rotations and \cnot{} gates using the synthesizing method from~\cite{shende_recognizing_2004}. The resulting circuit is shown in (b) and saves two \cnot{} gates compared to the original circuit provided in (a).}
\end{figure}

The goal of peephole optimization is to find longest sequences of gates acting on a small subset of qubits, optimize them, and then replace the sequences with their optimized versions. To optimize the gate sequences on a small number of qubits, one can for example use re-synthesizing methods: we multiply all the unitaries corresponding to the gates in the sequence and then use general methods to decompose the unitary again into a sequence of gates, trying to minimize the gate count used for the decomposition. The best known synthesizing methods for arbitrary isometries are provided as a software package that was introduced in~\cite{iten_introduction_2019} and are based on
~\cite{iten_quantum_2016} and references therein. For the special case of two-qubit unitaries, optimal decomposition methods are known~\cite{shende_recognizing_2004}. In the special case of optimizing Clifford circuits, re-synthesizing methods were proposed in~\cite{kliuchnikov_optimization_2013} and successfully applied to reduce the gate count by  about 50\% with peephole optimization, using a heuristic algorithm introduced in~\cite{prasad_data_2006} to find the longest gate sequences.

Finding longest gate sequences can be considered as a special case of our algorithm \texttt{PatternMatch}, where we do not fix a specific pattern, but instead consider every gate in the circuit to be a match as long as it only acts on the selected qubits. Since the correctness proof of our algorithm works for any pattern, we are still guaranteed to find all longest sequences in a circuit. 
To show how finding longest sequences acting on a few qubits can help with circuit optimization, an example for a longest sequence on two qubits found by \texttt{PatternMatch} is given in Figure~\ref{fig:circuit_chain}, and the optimized circuit using peephole optimization is shown in Figure~\ref{fig:chain_opt}. In the example, one has to cleverly commute gates in the middle of the circuit, which are not part of the longest sequence of gates that is found. The algorithm proposed in~\cite{prasad_data_2006} would not consider this option and hence would not find the longest sequence for the given example. Having found a longest sequence using \texttt{PatternMatch}, one can then use optimal synthesizing methods for two qubit unitaries~\cite{shende_recognizing_2004} to reduce the number of \cnot{} gates from $4$ to $2$. Apart from \cnot{} gates, the optimized circuit contains arbitrary single-qubit gates $U_3$ parameterized by the Euler angles $\theta$, $\lambda$ and $\phi$ and defined by
\begin{equation*}
U_3\left(\theta,\phi,\lambda\right) = 
\begin{pmatrix}
\cos{\frac{\theta}{2}} & -e^{i\lambda} \sin{\frac{\theta}{2}}  \\
e^{i\phi}\sin{\frac{\theta}{2}} & e^{i(\phi+\lambda)}\cos{\frac{\theta}{2}} \\
\end{pmatrix}.
\end{equation*}

\begin{figure}[t]
    \centering
    \includegraphics[width=0.7\textwidth]{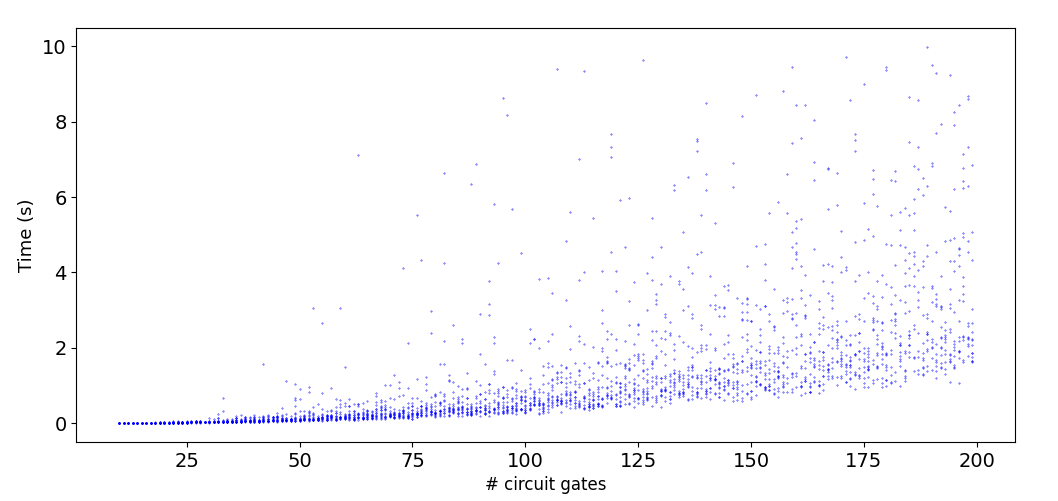}
    \caption{ The figure shows the runtime of \texttt{PatternMatch} adapted for searching the longest (connected) gate sequences on all subsets of two qubits in randomly generated circuits with $|C| \in [10,200]$ gates on $n_C=6$ qubits consisting of X, \cnot{} and Toffoli gates. We create 15 random circuits for each choice of $|C|$. The plot shows that the runtimes generally scale amenably with the circuit size, but have a strong dependence on the structure of the (randomly sampled) circuits, since runtimes vary a lot even for fixed circuit size.
    }
    \label{fig:scaling_chain}
\end{figure}

%\footnote{The matrix representation of the gate $U_3$ is given by }
%gates to perform single qubit rotation. 
%The gate $U_3$ represents an arbitary single-qubit gate where the parameters $\theta$, $\lambda$ and $\phi$ are the three Euler angles. Note that on current experimental architectures, single-qubit have a much lower implementation cost than two qubit gates, such as \cnots{}.

We note that while using our algorithm for peephole optimization works well in practice, we loose our bound on the worst-case time complexity if we search for longest sequences instead of a specific pattern, since the size of the ``pattern'' is now unbounded. However, we provide numerical results for the runtime of our algorithm for finding maximal sequences on two qubits in Figure~\ref{fig:scaling_chain}.

%It is possible to adapt the algorithms \texttt{ForwardMatch}, \texttt{BackwardMatch} and \texttt{TemplateMatch} to search for every two qubits chains in a quantum circuit.  For every gate in the circuit we first try to expand the chain in the forward direction (directly analog to \texttt{ForwardMatch}), the criteria that replaces the matching is the condition to act on the same two qubits or on a subset of these. Then we try to expand the chain in the backward direction the same way as \texttt{BackwardMatch}. This adapted algorithm yields two qubits chains in the circuit. Afterwards, we can optimize the the count of two qubits gates with known approximation methods\cite{Cross_2019}.

%The previous method is implemented in Qiskit \cite{Qiskit} as \texttt{TwoQubitBasisDecomposer}, it  synthesizes $4\times4$ unitaries into a circuit with minimal number of the basis two qubits gates, in our case the basis is \cnot{}. 

%%%%%%%%%%%%%%%%%%%%%%%%%%%%%%%%%%%%%%%%%%%%%%%%%%%%%%%%%%%%%%%%%%%%%%%%%%%
%%%%%%%%%%%%%%%%%%%%%%%%%%%%%%%%%%%%%%%%%%%%%%%%%%%%%%%%%%%%%%%%%%%%%%%%%%%
%%%%%%%%%%%%%%%%%%%%%%%%%%%%%%%%%%%%%%%%%%%%%%%%%%%%%%%%%%%%%%%%%%%%%%%%%%%

%%%%%%%%%%%%%%%%%%%%%%%%%%%%%%%%%%%%%%%%%%%%%%%%%%%%%%%%%%%%%%%%%%%%%%%%%%%
%%%%%%%%%%%%%%%%%%%%%%%%%%%%%%%%%%%%%%%%%%%%%%%%%%%%%%%%%%%%%%%%%%%%%%%%%%%
%%%%%%%%%%%%%%%%%%%%%%%%%%%%%%%%%%%%%%%%%%%%%%%%%%%%%%%%%%%%%%%%%%%%%%%%%%%
\appendix

\section{Correctness (proof of Theorem~\ref{thm_correctness})} \label{app_correctness}
To prove the assertion of Theorem~\ref{thm_correctness} we need a couple of lemmas that ensure the correctness of the individual subroutines used by the \texttt{PatternMatch} algorithm.

%%%%%%%%%%%%

\begin{lemma}[Correctness of \texttt{ForwardMatch}
(Algorithm~\ref{algo:ForwardMatch})] \label{lem_forwardMatch}
Consider a circuit $C$ and a pattern $T$ with a fixed assignment of pattern qubits to circuit qubits (given by $L_{\textnormal{q}}^{\textnormal{sel}}$ in \cref{algo:ForwardMatch}) and fixed starting gates $C_r = T_i$.
Split the pattern as $T = (\TLeft, \TRight)$, where $\TRightI{1} = T_i$ is the starting gate and the forward part $\TRight$ is defined by 
\begin{align} \label{eq:ForwardMatch_template_property}
\comuCirc{1}{j}{\TRight}\neq 0 \quad \textnormal{ for all $j\in \{2,3,\dots,|\TRight| \}$} \, ,
\end{align}
and $\TLeft$ contains the remaining gates.
Then, \cref{algo:ForwardMatch} finds a maximal match (up to equivalence) of the partial pattern $\TRight$ in $C$.
Furthermore, \texttt{ForwardMatch} only blocks vertices in the canonical form of the circuit $C$ that  correspond to gates  that could never be matched with gates in the full pattern $T$ (or which would lead to equivalent matches).
% Given a circuit $C$, a start index $r$, a set of qubits $L_{\textnormal{q}}$, and a template $T$, such that  $(T_i,T_{i+1},\dots,T_{|T|}) \simeq (\TLeft,\TRight)$ for some $i \in \{1,\dots,|T|\}$ and with $\TRightI{1}=T_i$. Let us further assume that the circuit $\TRight$ satisfies
% \begin{align} \label{eq:ForwardMatch_template_property}
% \comuCirc{1}{j}{\TRight}\neq 0 \quad \textnormal{ for all $j\in \{2,3,\dots,|\TRight| \}$} \, .
% \end{align}
% Then Algorithm~\ref{algo:ForwardMatch} finds one of the equivalent maximal matches of the template $\TRight$ on the qubits with labels in $L_{\textnormal{q}}$ in $C$.
% Furthermore, \texttt{ForwardMatch} only blocks vertices in the canonical form of the circuit $C$ that  correspond to gates  that could never be matched with gates in the full template $T$ (or which would lead to equivalent matches).
\end{lemma}
\begin{proof}
Let us denote the canonical form of the circuit $C$  and of the pattern $T$ by $G^C$ and $G^T$, respectively. We have only to consider gates $C_s$ with $G^C_s \in  \texttt{Succ}(G^C_r,G^C)$ for matching with $\TRight$. Indeed, let us show this rigorously.

\textbf{Matching gates must correspond to vertices in $G^C$ that are successors of $G^C_r$}: Assume by contradiction that there exists a pattern match (see Definition~\ref{def:template_match}) where a gate $T_j \in \TRight$  is matched with a gate $C_s$ (with $s \neq r)$, such that $G^C_s \notin  \texttt{Succ}(G^C_r,G^C)$. 
%Hence, by Lemma~\ref{lem:Property_CanonicalForm}, we have that $s<r$ or $\comuCirc{r}{s}{C}=0$. 
By~\eqref{eq:ForwardMatch_template_property} and Lemma~\ref{lem:Property_CanonicalForm}, we have $G^T_j \in  \texttt{Succ}(G^T_i,G^T)$. Choose a path $p$ from the vertex $G^T_i$ to $G^T_j$. Since a pattern match is connected, we find that all the gates with labels corresponding to the vertices on the path $p$ must also be matched with gates in the circuit. Therefore, we find a corresponding path from  $G^C_r$ to $G^C_s$ in the circuit, and hence $G^C_s\in  \texttt{Succ}(G^C_r,G^C)$, which contradicts the initial assumption.

Let us now go on with the proof of the correctness of \texttt{ForwardMatch}. By construction, the algorithm goes through all the successors of $G^C_r$ in the while-loop, and hence it visits all  gates that possibly could be matched with gates in $\TRight$. Indeed, the successors of $G^C_r$ are either considered as a direct successor of a matched vertex that was inserted into the list  \textit{MatchedVertexList}, or as a successor of a vertex that is blocked in a cycle of the while-loop (see line~\ref{algo:ForwardMatch_blocking_successors} in Algorithm~\ref{algo:ForwardMatch}). Moreover, the order of considering the successor vertices of $G^C_r$ that might be matched corresponds to the order of the corresponding gates in the circuit (from left to right). This is ensured by always considering the successors with minimal label that we have to visit first in the while-loop.

It remains to be shown that \texttt{ForwardMatch} handles each successor of $G^C_r$ in a way that leads to one of the (equivalent) longest matches. %Note that we put the circuit at the right of the position $I(C_r,C')$ into the canonical form at the start  and that all the steps in the while-loop keep the circuit and the template in the canonical form. 
Since a pattern match must be connected and we start matching $C_r$ with $\TRightI{1}$ from the left to the right in the circuit, the indices \textit{CandidateIndices} of the gates in the pattern that could be matched  next must be direct successors of the already matched gates.  Further, direct successors are only valid candidates if matching them does not make the match in the pattern unconnected. Specifically, we have to ensure the following: a direct successor $v$ of an already matched gate is a candidate if and only if there is no unmatched vertex $v'$ that is a predecessor of $v$ and a successor of an already matched gate. Indeed, also note that in such a case, the unmatched vertex $v'$ could not be matched later on in the forward matching proccess, since we would never consider predecessors of the vertex in the circuit that is matched with $v$ later on in the matching process (because predecessors would have lower labels and we are looping trough the circuit from left to the right with increasing labels).
Therefore, all valid candidates are found in Algorithm~\ref{algo:FindForwardCandidates}. 
We have the following two cases in the while-loop for a matched root vertex $v_0$ and a direct successor $v$ of $v_0$ with label $s$ that we consider for matching in the canonical form $G^C$ of the circuit $C$.
\begin{enumerate}
\item \label{proof_forward:case1} \textbf{Gate $C_s$ matches with a gate $T_j$ with $j\in \mathit{CandidateIndices}$.} In the following we show that the optimal strategy in this case is indeed to greedily match the two gates. There are two cases: 
\begin{enumerate}
\item If the vertex $v=G^C_s$ is the only successor of $G^C_r$ with label bigger or equal $s$, such that $C_{s}=T_j$, we should match the two gates, since matching them will not disturb any possible further matches (but can only lead to longer matches). Indeed, not matching would reduce the possible candidates in the pattern for further matches and might block vertices in $G^C$ that are successors of the unmatched gate. Note that the order in which the gates are visited in the while-loop ensures that all predecessors of $v$ that are also successors of the starting vertex $G^C_r$ are already matched if we end up in line~\ref{ForwardMatch_matching_case} in Algorithm~\ref{algo:ForwardMatch}, and hence, matching leads to a connected match.

\item However if we have a further vertex $G^C_t \in \texttt{Succ}(G^C_r,G^C)$ with $t>s$ and with $C_{t}=T_j$, it could a priori happen that we should not match, because matching $G^C_t$ instead of $v$ could lead to a longer match  (see point~\ref{problems:additional_gates} in Section~\ref{sec:possible_difficulties}). In the following we show that such a vertex $G^C_t$ would have to be a direct successor of the vertex $v_0$ to be properly matched  (i.e., to lead to a connected match including the starting gate $G^C_r$), and hence, matching $G^C_t$ instead of $v$ would lead to equivalent matches. Indeed, if $G^C_t \notin \texttt{DirectSucc}(v_0,G^C)$, it has to be a successor of $v_0$ with at least one vertex in between (because $\comu{C_{v_0.\mathit{label}}}{C_{t}}\neq 0$, since $C_t=T_j=C_s$, $G^C_s\in  \texttt{DirectSucc}(v_0,G^C)$ and $s<t$). In this case, the vertex in between cannot be matched, and hence this scenario would not lead to a connected match (as long as we do not block the matched predecessors of $G^C_t$, which is not allowed since it would also block the starting vertex $G^C_r$).
\end{enumerate}

\item \textbf{Gate $C_s$  does not match with a gate $T_j$ with $j\in \mathit{CandidateIndices}$.} 
Similarly as above, one can see that matching the gate $C_s$ with any gate in $\TRight$ or in $\TLeft$ cannot lead to a connected match. Since  $G^C_s$ is a successor of the starting vertex  $G^C_r$, all the successors of $G^C_s$ can also not be matched, because matching them would not lead to a connected match. Hence, we only block vertices that could never be matched with any gate in $\TRight$ or in $\TLeft$, which shows the last statement of the lemma.
 \end{enumerate}

This finishes the proof.
\end{proof}

For the proof of the correctness of  \texttt{BackwardMatch}, we will further use the following fact related to the construction of the algorithm \texttt{ForwardMatch}.

\begin{definition}
A match $S$ (i.e., a set $S$ of matched gate indices) is called sub-match of a match $M$, if there exists a match $\tilde{S}$ that is equivalent to the match $S$ and such that $\tilde{S} \subset M$.
\end{definition}

\begin{lemma}[Characterization of non-maximal forward matches] \label{lem_partialForwardMatches}
Consider a circuit $C$, a pattern $\TRight$ that satisfies the property~\eqref{eq:ForwardMatch_template_property}, a qubit assignment specified by $L_{\textnormal{q}}^{\textnormal{sel}}$, and a start index $r$ in $C$.
Then, all possible (not necessarily maximal) matches $M$ of the pattern $\TRight$ in $C$ (with the specified qubit assignment and the starting gate $C_r$ matched with $\TRightI{1}$) are sub-matches of the maximal match $M^{\textnormal{max}}$ found by \texttt{ForwardMatch}.
\end{lemma}
\begin{proof}
Assume by contradiction that there is a match $M$ of $\TRight$ in the circuit $C$ (matching $\TRightI{1}$ with $C_r$), such that there is no sub-match in  $M^{\textnormal{max}}$ that is equivalent to $M$. 
%Then, there  is a match $m=(\TRight_i,C_j) \in M$ with $m \notin M^{\textnormal{max}}$ and that is not equivalent to a match in $ M^{\textnormal{max}}$. 
Since the algorithm \texttt{ForwardMatch} goes through all gates that could possibly be matched with $\TRight$ (see the proof of Lemma~\ref{lem_forwardMatch}), there is a first gate $C_s$ in $M$ considered in this process that  is not matched  in  the maximal match  $M^{\textnormal{max}}$ found by \texttt{ForwardMatch} (and there is also no equivalent match). 
Since matching a gate in this process can never disturb future matches (apart from equivalent ones) as shown in case~\ref{proof_forward:case1} in the proof of Lemma~\ref{lem_forwardMatch}, we conclude that the gate $C_s$ could also be matched in the forward matching process. However, according to the algorithm \texttt{ForwardMatch}, the gate would then indeed be matched, leading to a contradiction with the assumption that $C_s$  is not matched  in  the maximal match  $M^{\textnormal{max}}$ found by \texttt{ForwardMatch}.
\end{proof}

\begin{lemma}[Correctness of \texttt{BackwardMatch} (Algorithm~\ref{algo:BackwardMatch})] \label{lem_backwardsMatch}
Consider a pattern $T$ and a starting index $i \in \{1,\dots,|T|\}$, and split the sub-pattern $(T_i,T_{i+1},\dots,T_{|T|})$ as $(T_i,T_{i+1},\dots,T_{|T|}) \simeq (\TLeft,\TRight)$, with $\TRight$ defined by Property~\eqref{eq:ForwardMatch_template_property}.
Suppose we are further given a circuit $C$ and a maximal match $\MRight$ of $\TRight$ in $C$, with a qubit assignment specified by $L_{\textnormal{q}}^{\textnormal{sel}}$ and a starting gate $C_{r}$ matched with $\TRightI{1}$. 
Then, Algorithm~\ref{algo:BackwardMatch}  finds all maximal matches of $T$ with qubit assignment $L_{\textnormal{q}}^{\textnormal{sel}}$ and starting gate $C_r = \TRightI{1}$.
\end{lemma}
%\ds{$C$ is not defined in the lemma above.}
Let us first give an overview of the proof idea. Since we have already given a maximal match of $\TRight$ in $C$, it remains to verify how many gates of $\TLeft$ can be matched. In general, it can happen that a gate that disturbs the match of $\TLeft$ can be moved to the right in $C$. This may allow us to match further gates in $\TLeft$, but may also disturb the maximal match of $\TRight$. To handle this tradeoff, we have to consider both possibilities: (i) moving the disturbing gate as far as possible to the right, and (ii) moving it as far as possible to the left. These options correspond to blocking the successors or the predecessors of the vertex corresponding to the disturbing gate in the canonical form of the circuit $C$. We then have to keep track of both possibilities and go on with matching in both cases building up a tree of possible matching scenarios that could lead to a maximal match. Let us now give a detailed proof of the correctness of the algorithm.

\begin{proof}[Proof of Lemma~\ref{lem_backwardsMatch}]
%By the property of the circuit $\tilde{C}$, there is no gate on the right of $\tilde{C}_{\tilde{r}}$ that could match a gate of $\TLeft$. First, we move the disturbing gates as far as possible to the left by using \textit{RightCommute} to commute all the non disturbing gates as far as possible to the right, but not further than at the position of $\tilde{C}_{\tilde{r}}$. By Lemma~\ref{lem:RightCommute} we get a circuit (that we again call $\tilde{C}$), such that for any index $l \leq I(C_r,\tilde{C})$ and $j = I(l,\tilde{C})+1$, we have 
%\begin{align}
%\comu{\tilde{C}_{I(l,\tilde{C})}}{\tilde{C}_{j}}=0 \implies f(\tilde{C}_{j})=1 \, ,
%\end{align}
%where  $f(\tilde{C}_i) =1$ if and only  if $\tilde{C}_i$ acts only on qubits with labels in  $L_{\textnormal{q}}$. This property is conserved if we put the circuit on the qubits with labels in $L_{\textnormal{q}}$ into the canonical form.
 The matching process works with the canonical representations $G^C$ of the circuit $C$, where the matched gates listed in $\MRight$ are marked as matched in $G^C$, and the remaining gates in the forward part are blocked (corresponding to the state after running \texttt{ForwardMatch}). 
Let us first show that the strategy of first finding the maximal forward match, and then maximally extending it in the backward direction, indeed leads to all possible maximal matches of the full pattern. 
%In other words, we would like to show that starting the algorithm \texttt{BackwardMatch} with a maximal forward match $\MRight$ as an input leads to the maximal match of the full template. 
A priori, it could happen that a non maximal forward match could lead to a maximal match of the full pattern. However, by Lemma~\ref{lem_partialForwardMatches}, different forward matches could only be sub-matches of the maximal match. Since in the process of backwards matching (as analyzed in detail below) further matched gates on the right of the starting gate $C_r$ with gates in $\TRight$ can only lead to longer matches, we conclude that the strategy of first finding the maximal forward match is indeed optimal.\\

Let us know consider the backward matching process. The while-loop of the matching process in Algorithm~\ref{algo:BackwardMatch} goes through all the vertices in the canonical form $G^C$ of the circuit that are not blocked or already matched (as long as there are gates left to consider in $\TLeft$ that could possibly be matched). Since the blocked gates will never match (see Lemma~\ref{lem_forwardMatch}), we loop through all vertices corresponding to gates that could possibly be matched with $\TLeft$. The indices of these vertices are stored in a list $\mathit{GateIndices}$ in decreasing order (see line~\ref{BackWardMatch_GateIndices1} and line~\ref{BackWardMatch_GateIndices2} in Algorithm~\ref{algo:BackwardMatch}). The variable $\mathit{counter}$ keeps track of the number of gates with indices listed in $\mathit{GateIndices}$ that we have already considered in the backwards matching process. During the matching, we create a stack of possible matching scenarios that may lead to a maximal match. All of these scenarios are then considered for further matching in the next steps of the while-loop. Therefore, it remains to show that each step in the while-loop with parameters $(G^C,M,\mathit{counter})$ puts all the matching possibilities for the gate with index  $s=(\mathit{GateIndices})_\mathit{counter}$ that might lead to a maximal match on the stack $\mathit{MatchingScenarios}$ for further investigation. The indices  \textit{CandidateIndices} of the gates in the pattern that could possibly match are found by Algorithm~\ref{algo:FindBackwardCandidates}. To show the correctness of one cycle of the while-loop, let us consider two cases separately. 
%For each case, we can ignore matching possibilities where more than $|\TLeft|$ matched gates in $\MRight$ are disturbed, since in this case, choosing the original right match $\MRight$ matches more gates and hence, such possibilities cannot lead to maximal matches.
\begin{enumerate}
\item \textbf{Gate $C_s$ matches with at least one gate in the pattern.} \label{item:BackwardMatch_case_1}
We have to consider all possible cases that could lead to a maximal match:
\begin{enumerate}
    \item Match the gate $C_s$ with any candidate  $T_j=C_s$ with $j \in \mathit{CandidateIndices}$,\label{case_1_backward_match_proof}
    \item right-block the vertex $C_s$ (i.e., block the vertex and move it as far as possible to the right in the circuit $C$),\label{case_2_backward_match_proof}
    \item left-block the vertex $C_s$ (i.e., block the vertex and move it as far as possible to the left in the circuit $C$).\label{case_3_backward_match_proof}
\end{enumerate}
Instead of always adding all of these options to the stack $\mathit{MatchingScenarios}$, we can save runtime by ignoring options that can not lead to maximal matches:

In case~\ref{case_1_backward_match_proof}, there might be successors of the matched gate in the pattern that were not matched so far. In this case, we block these vertices and all of its successors in the pattern, since they could never be matched later on (since we traverse trough the gates from right to left in the circuit). If we block the initial match of $C_r$ with $T_i$, we do not have to consider this matching scenario, since we consider the initial match to be fixed. Further, if we block a gate $T_l$ that was matched during \texttt{BackwardMatch}, we do also not have to consider this matching scenario, since it is put on the stack $\mathit{MatchingScenarios}$ during another interation of the while-loop where $T_l$ was chosen to be not matched. 

The reasoning for case~\ref{case_2_backward_match_proof}, is analogous to case~\ref{case_1_backward_match_proof}, but considering blocking in the circuit instead of in the pattern.

The case~\ref{case_3_backward_match_proof} only has to be considered if there is no candidate in case~\ref{case_1_backward_match_proof} whose matching causes no blocking of already matched gates. Indeed, let us show this in detail and assume that there is a candidate $T_j$ that matches $C_s$ without disturbing any already matched gates. We need to show that blocking the gate $C_s$ and moving it as far as possible to the left cannot lead to a maximal match (that is not equivalent to another one that is found by the algorithm).
Indeed, moving the gate to the left (and ignoring it for the further matching process) could lead to one of the following two cases:
\begin{enumerate}
\item we match the gate  $T_j$ later on in the matching process with a different gate $C_{t}$ with $t<s$ in the circuit, \label{it_label_a}
\item we never match the gate $T_j$.
\end{enumerate}
It can be verified that case~(a) leads to matches that are equivalent to the ones where we match $C_s$ with $T_{j}$. Indeed, since we block all the predecessors of the vertex $G^C_s$, the vertex $G^C_{t}$ cannot be a predecessor of $G^C_s$. Hence, by Lemma~\ref{lem:Property_CanonicalForm}, the gate $C_t$ can be moved next to $C_s$, and we end up with equivalent matches.

In case (b), not matching $T_j$ with $C_s$ can only disturb the backwards matching process later on, since in this case gates that do not commute to the right of $C_s$ are blocked and the additional matching candidate $T_j$ will never be matched by assumption.
Hence, we can ignore the case where we move $C_s$ as far as possible to the left.

\item \textbf{Gate $C_s$ does not match with any gate in the pattern with index $s \in \mathit{CandidateIndices}$.} One can see that the gate $C_s$ can never match and may disturb the matching. 
%@Raban: Add further details here.
We consider both options: moving it as far as possible to the left and to the right in the circuit, and blocking the predecessors or successors, respectively. If we can move the gate $C_s$ to the right of all the matched gates in $\MRight$ or to the start of the circuit $C$ (see the condition in line~\ref{BackwardMatch_move_to_start_end} in Algorithm~\ref{algo:BackwardMatch}) without blocking any previously matched gates, doing so will lead to a maximal match, since the gate $C_s$ does not disturb the following matching process or the current match. Hence, in this case we only add this possibility to the stack $\mathit{MatchingScenarios}$. 
Otherwise, we do not know whether left-blocking or right-blocking will turn out to be the better choice overall, so we have to add both options.
As in case~\ref{item:BackwardMatch_case_1}, we do not have to add matching conditions to the stack if we disturb the initial match or any match with gates in $\TLeft$.
\end{enumerate}
This finishes the proof.
\end{proof}
%%%%%%

We are finally ready to prove Theorem~\ref{thm_correctness}.
\begin{proof}[Proof of Theorem~\ref{thm_correctness}]
Consider a maximal match $M$ of a pattern $T$ in the circuit $C$. We need to show that this match is found by Algorithm~\ref{algo:TempMatch}. The algorithm loops over all possible assignments of pattern qubits to circuit qubits, formalised by considering patterns $\tilde{T}$ obtained from the original pattern $T$ by qubit permutation.
Looping over all possible assignments ensures that there is a run of the loop with a pattern $\tilde{T}$, such that for all the  index pairs $(j,s) \in M$, we have $\tilde{T}_j=C_s$ (which in particular means that the two gates are acting on qubits with the same labels).
Assume that the indices of the gates $\tilde{T}_1,\tilde{T}_2,\dots,\tilde{T}_k$ are not listed (as a first entry of an element) in $M$, but $\tilde{T}_{k+1}$ is the first matched gate in the pattern, i.e., there is a tuple $(k+1,r) \in M$ for some $r \in \{1,2,\dots,|C| \}$.
Then, once in the loop of the algorithm \texttt{PatternMatch}, we set the start index of the pattern $i:=k+1$ and the start index in the circuit to $r$. By the correctness of \texttt{ForwardMatch} and \texttt{BackwardMatch} (see Lemma~\ref{lem_forwardMatch} and Lemma~\ref{lem_backwardsMatch}), all the maximal matches (up to equivalent ones) of $(\tilde{T}_{i},\dots,\tilde{T}_{|\tilde{T}|})$ in the circuit $C$ are found starting with matching the gate $T_i$ with $C_r$. Since the gates  $\tilde{T}_1,\tilde{T}_2,\dots,\tilde{T}_{i-1}$ are not contained in the match we are searching for, \texttt{PatternMatch} hence finds the maximal match $M$ (or an equivalent one).
\end{proof}

\section{Combinatorial algorithm} \label{app:combinatorial_algo}

Robin Kothari~\cite{kothari_private} pointed out a simple combinatorial algorithm for pattern matching, which also has polynomial scaling in the circuit size for a fixed pattern size. 
For the sake of completeness, we give a short overview of the main idea here and estimate the worst-case complexity.
As shown in \cref{fig:numerics_scaling_C}, while the worst-scaling is comparable to that of our algorithm, the practical performance is not.

The idea of the algorithm is to search for \emph{complete} matches, i.e., matches where every gate in the pattern is matched, by simply looping through all possible assignments of gates in the pattern to gates in the circuit, and checking whether any such assignment constitutes a complete match. Given such an algorithm for complete matches, one can then search for all maximal (possibly incomplete) matches of a pattern in a circuit by running the algorithm for all possible sub-patterns, i.e., any connected sub-sequence of gates in the pattern.

Clearly, such an algorithm is computationally quite expensive. 
However, the worst-case time complexity is still polynomial in the circuit size (for a fixed pattern size). 
To see this, first we note that there are $2^{|T|}$ different subsets of gates that can be chosen from the $|T|$ gates in the pattern. 
If we create the canonical for of the whole pattern once at the start (in time $\mathcal{O}(|T|^3)$), then checking if a subset of the pattern is connected can be done in time $\mathcal{O}(|T|^2)$  (by checking for each vertex in $G^T$ not contained in the chosen subset whether it has both successors and predecessors  that are contained in the subset; if this is the case, the subset is not connected; otherwise, it is a valid sub-pattern). 
Hence, all the sub-patterns can be found in time $\mathcal{O}(2^{|T|}|T|^2)$. 

For the remainder of the algorithm we need the canonical forms of all sub-patterns. 
Creating them requires $\mathcal{O}(|T|^3)$ time per sub-pattern\footnote{Instead of creating the canonical form for each sub-pattern from scratch, we could also cut out the canonical graphs for the sub-patterns from $G^T$, and update the lists of successors and predecessors in time $\mathcal{O}(|T|^2 \log{|T|})$, using binary search to remove the labels of the vertices that do not appear in the sub-pattern from the ordered lists. For simplicity, here we just assume that we create the canonical from scratch.}. This takes time $\mathcal{O}(2^{|T|}|T|^3)$ in total for all sub-patterns. 

For each sub-pattern, we consider all assignments of gates in the sub-pattern to gates in the circuit, of which there are
$$\mathcal{O}\left({|C| \choose |T|} |T|! \right) \, .$$

Each gate assignment was done without considering restrictions imposed by the structure of the circuit, i.e., disregarding the order of the gates, and which qubits they act on. Now, for each assignment, we have to check whether it describes a valid match. More concretely, we have to check whether
\begin{enumerate}
    \item  the gates in the pattern can be reordered to match the order of the gates they have been assigned to in the circuit, \label{item:check_order}
    \item the gates in the circuit that have had a pattern gate assigned to them form a connected sub-circuit, and \label{item:check_con}
    \item  there is a consistent global assignment of pattern qubits to circuit qubits such that all gates act on the requisite qubits. \label{item:check_qubit_ass}
\end{enumerate}

Let us consider a canonical form $\tilde{G}^T$ of a sub-pattern together with an assignment to gates in the circuit $C$. We denote the vertex in $G^C$ that is assigned to a vertex $v \in \tilde{G}^T$ by $v^{C}$. To check condition~\ref{item:check_order}, for each vertex $v \in \tilde{G}^T$  we check whether its successors are also assigned with successors of $v^{C}$ in $ G^C$. For each vertex in $\tilde{G}^T$, this can be done in time $\mathcal{O}(|T|\log(|C|))$, by looping trough the list of successors of $v$ (which is stored as an ordered list as part of $\tilde{G}^T$), and searching for the gates assigned to them using binary search in the ordered list of successors of $v^C$. Repeating this for all $|T|$ vertices, checking~\ref{item:check_order} requires time  $\mathcal{O}(|T|^2\log(|C|))$ in total. 

Checking condition~\ref{item:check_con}  can be done analogously to checking whether a subset of gates of the pattern is connected (which we do as part of checking condition~\ref{item:check_order}), and hence requires time $\mathcal{O}(|C|^2)$.

 Finally, checking~\ref{item:check_qubit_ass} can be done by collecting the qubit labels of all the qubits for which there is a  matched gate in the circuit (i.e., one that is assigned to one in the sub-pattern) that acts non trivially on it. We directly stop this proccess if more than $n_T$ qubits are collected, since in this case, there cannot be a valid qubit assignment. Hence, this can be done in time $\mathcal{O}(n_T |T|)$, where $n_T$ is the number of qubits in the pattern. If the number of collected qubits corresponds to the number of qubits in the sub-pattern, we loop trough all permutations of the qubit assignments, of which there are at most $n_T!$, and check in time $\mathcal{O}(n_T |T|)$ whether the qubit assignment is valid. Hence, checking~\ref{item:check_qubit_ass}, takes time $\mathcal{O}(n_T |T| + n_T! n_T |T|) = \mathcal{O}(n_T! n_T |T|)$.\footnote{In our numerical experiments, we use a somewhat different implementation that is more efficient in practice. For simplicity of describing our algorithm, let us assume that there are no gates in the pattern that have the same (non-trivial) action on different qubits (so we assume that e.g. the circuit contains no Toffoli gates, since the Toffoli gate acts in the same non-trivial way on its two control qubits). In this case, we can just loop trough the gates in the pattern and assign the corresponding pattern qubits to the qubits in the circuit. If this does not lead to contradicting assignments, we have found a valid qubit assignment. If gates like the Toffoli gate are present, one needs to check the different possible assignments separately, which is only efficient if the pattern contains a small number of such gates.}

Therefore, checking all three conditions requires time $\mathcal{O}(|T|^2\log(|C|) + |C|^2 + n_T! n_T |T|)$

We conclude that the worst case complexity of the combinatorial algorithm is given by
\begin{align}
&\mathcal{O}\left(2^{|T|}|T|^3{|C| \choose |T|} |T|! \left \{|T|^2\log(|C|) + |C|^2 + n_T! n_T |T|\right\} \right ) \\
\leq & \mathcal{O}\left(2^{|T|}|T|^3 |C|^{|T|} \left \{|C|^2\log(|C|) n_T! n_T |T|\right\} \right )
\\
\leq & \mathcal{O}\left(n_T!n_T2^{|T|}|T|^4 |C|^{|T|+3}  \right ) \, ,
\end{align}
where we chose a somewhat loose upper bounded to simplify the expression.
The worst-case complexity is polynomial in the circuit size for a fixed pattern size and has a similar scaling in $|C|$ as our algorithm (see \cref{eq_runntime}). However, crucially, the worst-case complexity of the combinatorial algorithm is close to the average-case complexity, since it really has to loop trough all possible assignments, whereas our algorithm performs much better in practice as demonstrated in Figure~\ref{fig:comparison_2_mean}. 
We also note that the scaling of the combinatorial algorithm is independent of $n_C$ (whereas our algorithm scales with $n_C^{n_T - 1}$); however, in practice we typically have that $|C| \gg n_C$, so this is of little use for most applications.

\newpage

%The complexity of the combinatorial algorithm in the number of qubits $n_C$ of the circuit seems to scale better than in our algorithm at first glance, indeed, the complexity even seems to be independent on $n_C$. However, one should note that working with gates that act on at most $k$ qubits, we have $n_C \leq k|C|$, since otherwise, there would be qubits on which no gate is acting. Hence, the actual scaling in the number of qubits is given by $\mathcal{O}\left(n_C^{|T|+3}\right)$. \tm{not sure what point we're making with this last remark. for fixed circuit size, our algo scales badly with $n_C$, and this doesn't seem to?}

%%%%%%%%%%%%%%%%%%%%%%%%%%%%%%%%%%%%%%%%%%%%%%%%%%%%%%%%%%%%%%%%%%%%%%%%%%%
%%%%%%%%%%%%%%%%%%%%%%%%%%%%%%%%%%%%%%%%%%%%%%%%%%%%%%%%%%%%%%%%%%%%%%%%%%%
\bibliographystyle{alpha}
\bibliography{template_bib}

\end{document}

%% file: definitions.tex
\usepackage{fullpage}
\usepackage{authblk}

\usepackage[usenames,dvipsnames]{xcolor}
\usepackage{amsfonts}
\usepackage{amsmath,amsthm,amssymb,dsfont}
\usepackage{enumerate}
\usepackage[english]{babel}
\usepackage{graphicx}
\usepackage{url}
\usepackage{todonotes}
\usepackage{bbm}
\usepackage{verbatim}
\usepackage{algorithmicx}
\usepackage{color}
\usepackage{pgf}
\usepackage{cases}
\usepackage{algorithm}
\usepackage{algpseudocode}
\usepackage[titletoc,title]{appendix}
\usepackage{times}
\usepackage{booktabs}
\usepackage{tikz}
\usepackage[most]{tcolorbox}

\usepackage[hyperfootnotes=false,bookmarks,colorlinks,breaklinks]{hyperref}  % PDF hyperlinks, with coloured links
\definecolor{dullmagenta}{rgb}{0.4,0,0.4}   % #660066
\definecolor{darkblue}{rgb}{0,0,0.4}
\hypersetup{linkcolor=red,citecolor=blue,filecolor=dullmagenta,urlcolor=darkblue} % coloured links
\usepackage[capitalize]{cleveref}
\usepackage[compact]{titlesec}
\titlespacing{\section}{0pt}{20pt}{8pt}
\titlespacing{\subsection}{0pt}{7pt}{4pt}

\usepackage{xy}
\xyoption{matrix}
\xyoption{frame}
\xyoption{arrow}
\xyoption{arc}

\usepackage{ifpdf}
\ifpdf
\else
\PackageWarningNoLine{Qcircuit}{Qcircuit is loading in Postscript mode.  The Xy-pic options ps and dvips will be loaded.  If you wish to use other Postscript drivers for Xy-pic, you must modify the code in Qcircuit.tex}
\xyoption{ps}
\xyoption{dvips}
\fi

\entrymodifiers={!C\entrybox}

\newcommand{\cnot}{{\sc C-not}}

\newcommand{\qw}[1][-1]{\ar @{-} [0,#1]}
\newcommand{\qwx}[1][-1]{\ar @{-} [#1,0]}
\newcommand{\gate}[1]{*+<.6em>{#1} \POS ="i","i"+UR;"i"+UL **\dir{-};"i"+DL **\dir{-};"i"+DR **\dir{-};"i"+UR **\dir{-},"i" \qw}   
\newcommand{\control}{*!<0em,.025em>-=-<.2em>{\bullet}}
\newcommand{\ctrl}[1]{\control \qwx[#1] \qw}
\newcommand{\targ}{*+<.02em,.02em>{\xy ="i","i"-<.39em,0em>;"i"+<.39em,0em> **\dir{-}, "i"-<0em,.39em>;"i"+<0em,.39em> **\dir{-},"i"*\xycircle<.4em>{} \endxy} \qw}
\newcommand{\push}[1]{*{#1}}
\newcommand{\gategroup}[6]{\POS"#1,#2"."#3,#2"."#1,#4"."#3,#4"!C*+<#5>\frm{#6}}
\newcommand{\lstick}[1]{*!R!<.5em,0em>=<0em>{#1}}
\newcommand{\Qcircuit}{\xymatrix @*=<0em>}
\usepackage[font=small]{caption}
\usepackage{subcaption}

\usepackage{amsbsy,verbatim,graphicx}

\usepackage{braket}

\newcommand{\comu}[2]{[#1, #2]}

\newcommand{\comuCirc}[3]{[#1, #2]_{#3}}

\renewcommand{\rho}{\varrho}
\renewcommand{\phi}{\varphi}
\def\id{\mathbbm 1}

\def\x{1.5} %distance x
\def\y{0.9} %distance y

\newcommand{\TRightI}[1]{{T^{\textnormal{forward}}_{#1}}}

\newcommand{\MLeft}{{M^{\textnormal{backward}}}}
\newcommand{\TRight}{{T^{\textnormal{forward}}}}
\newcommand{\TLeft}{{T^{\textnormal{backward}}}}
\newcommand{\CRight}{{C^{\textnormal{forward}}}}
\newcommand{\CLeft}{{C^{\textnormal{backward}}}}
\newcommand{\MRight}{{M^{\textnormal{forward}}}}

\newtheorem{theorem}{Theorem}[section]

\newtheorem{lemma}[theorem]{Lemma}

\theoremstyle{definition}
\newtheorem{definition}[theorem]{Definition}
\newtheorem{rmk}[theorem]{Remark}

\usepackage{xfrac}
\usepackage{setspace}

%% file: introduction.tex
Quantum computers promise computational advantages over their classical counterparts for various problems such as finding prime factors~\cite{shorAlgo}, solving linear equations~\cite{HHL09}, or finding elements in a database~\cite{groverAlgo}.
However, the experimental implementation of a quantum computer is a challenging task and in the near future, we can only expect to have access to noisy intermediate scale quantum (NISQ) devices~\cite{preskill2018quantum}.
These devices suffer from two main restrictions: firstly, the number of qubits is limited to $\sim 100$, and secondly these qubits are noisy, meaning that the quantum information stored in them degrades over time --- an effect called decoherence \cite{nielsenChuang_book}.

Decoherence limits the number of quantum gates that can be applied in a quantum circuit, and hence limits the complexity of computations that can be performed on near-term devices. 
Hence, it is crucial to optimize the quantum circuit executed on the device.
One approach is to try to reduce the number of quantum gates, i.e., to find an equivalent way of decomposing a quantum computation into quantum gates (from a fixed universal gate set) that requires less gates than the original circuit.
While finding the optimal circuit is a \textsf{QMA}-hard problem\footnote{\textsf{QMA} is the quantum analog of the classical complexity class \textsf{MA}, which in turn is the probabilistic analog of \textsf{NP}.}~\cite{janzing2003identity}, various practical approaches have been suggested for simplifying quantum circuits that might not lead to optimal circuits, but nonetheless can significantly lower the implementation cost (see for example~\cite{maslov_quantum_2006,Prasad:2006:DSA:1216396.1216399,mathias,Rahman:2014:AQT:2711453.2629537,Duncan2019GraphtheoreticSO}).
Beyond reducing the gate count, a number of more device-specific approaches have been investigated. 
For example, it can be beneficial to replace common sequences of gates by equivalent (and possibly longer) sequences that mitigates errors, e.g. if on a particular experimental device, implementing one gate introduces much more noise than implementing another.
Additionally, if the device consists of multiple groups of qubits on different chips, implementing operations involving qubits on different chips might be costly, so the quantum circuit should be optimized to minimize the number of such inter-chip gates.

In this work, we address an essential building block underlying many of these methods, namely the ability to find patters (or sub-circuits) in a large quantum circuit.
This enables a classical pre-processor to find patterns in a given quantum circuit and replace them with optimized variants that minimize gate counts or mitigate errors.
More precisely, we are interested in finding all matches of a pattern under pairwise commutation of gates, i.e., considering all possible orderings of quantum gates that arise by repeatedly swapping commuting gates in the circuit. The fact that different quantum gates may or may not commute is what makes this task challenging.
Suppose we are given a (potentially large) quantum circuit $C$ and a \emph{pattern}, i.e., another (small) quantum circuit $T$, expressed in terms of an arbitrary fixed set of gates.
Our goal is to find all maximal matches of the pattern $T$ in the circuit $C$, i.e., all instances of a maximally-sized sub-circuit of $T$ being equal to a sub-circuit in $C$, up to pairwise commutation of gates.
In particular, we do not require that the whole pattern $T$ can be matched to a sub-circuit in $C$, which is called a \emph{complete match}. 
If no complete match exists, we want to find the largest possible sub-circuit of $T$ that does have a match in the circuit $C$.

We give the first practical algorithm for this task that provably finds all matches of a pattern and whose worst-case complexity scales polynomially as a function of the circuit size (for a constant-sized pattern). 
Specifically, the time complexity of our algorithm as a function of the number of gates $|T|$ and qubits $n_T$ in the pattern and gates $|C|$ and qubits $n_C$ in the circuit scales as 
\begin{align} \label{eq_runntime}
 \mathcal{O}\left(  |C|^{|T|+3} |T|^{|T| + 4}\, n_C^{n_T-1} \right) \, .
\end{align}

A typical application is to find all matches of a small (constant-sized) pattern in a very large circuit $|C|$.  While in that case, our algorithm scales polynomially in principle, the degree $|T| + 3$ of the polynomial can be quite large.
In numerical experiments on random quantum circuits, we find that our algorithm scales considerably better for a fixed pattern size, achieving a roughly quadratic scaling in the number of gates in the circuit $C$ for a fixed number of qubits $n_C$ (\cref{sec:numerics_intro}).
In addition, we present heuristics that can significantly improve the runtime of the algorithm, at the cost of no longer being guaranteed to find all maximal matches (\cref{sec:heuristics_intro}). In particular, these heuristics address the potentially costly scaling $n_C^{n_T - 1}$ in the number of qubits $n_C$ and $n_T$ in the circuit and pattern, respectively.
The trade-off between finding all matches and lower runtime can be handled by the choice of parameters that regulate the heuristics. We numerically analyze how these heuristics impact the runtime and number of found matches. 

% \RI{Not sure if we should sell this differently...} \tm{put a shorter alternative above}
% On the other hand, the scaling in the number of qubits in the circuit $n_C^{n_T-1}$ in~\eqref{eq_runntime} might lead to long runtimes if the pattern also consists of many qubits (however, in praxis, we usually have $n_T=2,3,4$, and the algorithm still determines in a reasonable time.). For such cases, we present heuristics that can significantly improve the running time of the algorithm, at the cost of no longer being guaranteed to find all maximal matches. The trade-off between finding all matches and lower runtime can be handled by the choice of parameters that regulate the heuristics. We numerically analyze how these heuristics impact the running time and number of matches. 

% Crucially, the term $n_C^{n_T-1}$ does only appear in the worst-case complexity in the case where we have to consider all assignments of the qubits in the template with the qubits in the circuit. If we, however, are interested in problems where this does not matter (such as for peephole optimization described below), the worst-case complexity is reduced to $\mathcal{O}\left( (2|T|)^{|T|+1}|C|^{3+|T|} \, \right)$.

\subsection{Previous work} \label{sec:prev_work}

In this section, we discuss previous work on pattern matching in circuits and explain the additional difficulties one encounters in the quantum setting compared to the classical one (see also Section~\ref{sec:possible_difficulties} for a detailed description of arising difficulties).
Pattern matching for classical circuits is well studied (see for example~\cite{watanabe_new_1983,luellau_technology_1984,boehner_logex-automatic_1988,ohlrich_subgemini:_1993}) and has found many applications in the context of computer-aided design (see~\cite{ohlrich_subgemini:_1993} and references therein). 
If all gates in a circuit commute, pattern matching is straightforward: we can simply check whether all gates in the pattern can be found in the circuit.
The other extreme case is to assume that none of the gates in a circuit commute (apart from the trivial commutations of gates acting on different wires in the circuit).
For (non-reversible) classical circuits, this is a realistic assumption.
This case was reduced to the  \emph{subgraph isomorphism problem} in~\cite{ohlrich_subgemini:_1993}, which is an \textsf{NP}-complete problem (for a variable size of the subgraph)~\cite{Cook:1971:CTP:800157.805047}, so we cannot expect to find a general polynomial-time algorithm for this problem.
However, for a fixed size of the pattern a polynomial-time algorithm for pattern matching is possible since for a fixed size of the subgraph, the subgraph isomorphism problem can be solved efficiently~\cite{ullmann76}.

\paragraph{Pattern matching for reversible and quantum circuits.}
The quantum case (and the case of reversible circuits, to which our results also apply) considered here lies in between the two extremes of fully commuting and fully non-commuting gates: some of the quantum gates commute, while others do not.
This introduces an additional difficulty compared to the fully non-commuting case since in comparing the pattern and the circuit, we must also consider all possible re-orderings of gates in both the pattern and circuit and see whether one of these re-orderings yields a match.
Translated to the picture of subgraph isomorphism, where the circuit is represented as a directed acyclic graph and one tries to find the subgraph representing the pattern, this corresponds to solving the subgraph isomorphism problem with the additional difficulty that some of the vertices are allowed to be interchanged.
The rules for interchanging vertices are derived from the commutation of the quantum gates correspoding to these vertices.
As far as we know, this more complex task has not yet been studied in the general context of graph matching algorithms.

In~\cite{Maslov_reversible_logic, Rahman:2014:AQT:2711453.2629537}, heuristic pattern matching algorithms were introduced. The algorithm presented in~\cite{Maslov_reversible_logic} was then applied in~\cite{Maslov_Toffoli,maslov_quantum_2006} for reversible logic synthesis and achieves low runtimes. In~\cite{mathias}, a pattern matching algorithm is presented that provably finds all matches. It is based on mapping the circuit to a satisfiability modulo theory problem and applying a specific solver to it. Moreover, it is shown that finding all the matches indeed helps to significantly reduce the gate counts further compared to heuristic approaches. Unfortunately, this improvement comes in trade-off with the runtime of the algorithm. The algorithm is not efficient, i.e., its worst-case time complexity grows exponentially in $|C|$, and its practical runtimes are significantly higher than for the heuristic approaches (see~\cite{mathias} for a comparison with~\cite{Maslov_reversible_logic}).\footnote{Unfortunately, we could not get access to the code for the algorithm in~\cite{mathias} and hence were unable to perform a detailed numerical comparison. However, a rough comparison suggests that our algorithm scales significantly better than the one suggested in~\cite{mathias}.}
A simple combinatorial algorithm was pointed out to us by Robin Kothari~\cite{kothari_private}. Here, the idea is to loop trough all possible assignments of gates in the pattern with gates in the circuit, and then (efficiently) check if the assignment leads to a valid match. For completeness, we outline this algorithm in \cref{app:combinatorial_algo}.
This algorithm has a similar worst-case complexity as ours, but performs much worse in practice because the process of trying all assignments cannot be terminated early, whereas our algorithm typically runs much faster than its worst-case complexity suggests (see \cref{sec:numerics_intro}).

\paragraph{Circuit optimization from pattern matching.}
One of the main applications of pattern matching in quantum circuits is to reduce the gate count of the circuit.
The idea for optimizing quantum circuits using a pattern matching algorithm was introduced in~\cite{Maslov_reversible_logic} based on the rewriting rules found in~\cite{iwama02}.  
For this, one considers the special case where the pattern implements the identity operation. 
In the context of circuit optimization, such patterns are often called \emph{templates}.
More formally, in~\cite{maslov03} a template $T$ is defined as a sequence of unitary gates $U_i$ such that  $U_{|T|}\dots U_1= \id$.

To see how a pattern matching algorithm can be applied to optimize a circuit, assume that the pattern matching algorithm finds the gate sequence $U_a \dots U_b$ in a circuit $C$ for some $1 \leq a \leq b \leq |T|$ that matches the template, i.e., that after suitably interchanging commuting gates in the circuit and the template, the gate sequence $U_a\dots U_b$ appears in this order in the commuted circuit (with no other gates in between).\footnote{As explained above, a pattern matching algorithm does not require the complete pattern to be present in the circuit, but can find a maximal subsequence of gates in the template that matches a part of the circuit.}
 Since the full template implements the identity operator and since each unitary $U_i$ has an inverse $(U_i)^{-1}=U_i^{\dagger}$, we find an alternative representation of this gate sequence as $U_a\dots U_b= U_{a-1}^{\dagger}\dots U_1^{\dagger} U_{|T|}^{\dagger}\dots U_{b+1}^{\dagger}$. If this gate sequence has lower implementation cost than the original one, we may replace the found match with the gate sequence $U_{a-1}^{\dagger} \dots U_1^{\dagger}U_{|T|}^{\dagger}\dots U_{b+1}^{\dagger}$ in the circuit. 
 This also shows why it is important to find maximal matches: if the gate sequence $U_a \dots U_b$ is as long as possible, then the equivalent expression $U_{a-1}^{\dagger}\dots U_1^{\dagger} U_{|T|}^{\dagger}\dots U_{b+1}^{\dagger}$ contains as few gates as possible, so that we save as many gates as possible by replacing the former with the latter.
Indeed, in~\cite{mathias} it was shown that using an exact algorithm like ours for pattern matching can reduce the gate count of reversible circuits further by up to 28\% compared to the heuristic algorithms~\cite{Maslov_reversible_logic}.
 Another circuit optimization technique, called peephole optimization~\cite{prasad_data_2006, kliuchnikov_optimization_2013}, is also based on pattern matching and is described in \cref{sec:peephole}.
%  \RI{Maybe add a short section about peephole optimization here.}

%%%%%%%%%%%%%%%

\subsection{Pattern matching algorithm for quantum circuits} \label{sec:algo_intro}
In this section, we give an overview of our pattern matching algorithm and show how it works on a concrete example. The detailed description of the algorithm, including pseudocode, a proof that all maximal matches are found, and an analysis of the worst-case time complexity can be found in~\cref{sec:temp_match}.
We start with a brief description of the quantum circuit model and a canonical form of quantum circuits that we use in our work.

\paragraph{Quantum circuit model.} 
In the circuit model of quantum computation, information carried in qubit wires is modified by quantum gates which mathematically are described by unitary operations.  Examples of gates include the \cnot{} gate, Toffoli gate, and single-qubit gates such as the Pauli-$X$, Pauli-$Z$, and Hadamard $H$ gate.
While we pick a fixed gate set for our example below, our algorithm works for an arbitrary gate set; in fact, it suffices just to know which gates commute and which do not, without actually knowing the unitary implemented by each gate.

It is convenient to represent quantum circuits diagrammatically.
Each qubit is represented by a wire and gates are shown using a variety of symbols. Conventionally time flows from left to right.
The $i$-th gate in a circuit $C$ is denoted by $C_i$.
For example, the circuit in \cref{fig_canonical_form_circuit1} shows a \cnot{} gate controlled on qubit 1 and acting on qubit 2, followed by another \cnot{} gate, a Pauli-$Z$ gate, and finally a Toffoli gate controlled on qubits 2 and 3 and acting on qubit 1. 

% The circuit depicted in Figure~\ref{fig_notation} shows a \cnot{} gate controlling on the most significant qubit (with label 1) and acting on the least significant qubit (with label 4), a Toffoli gate controlling on the qubits 1,2 and acting on the qubit 3 and finally a $R_x$ rotation with rotation angle $\pi$ acting on qubit 1.
% \begin{figure}[htb]
% \centering
% \[
% \Qcircuit @C=1.0em @R=0.6em {
% \lstick{\textnormal{qubit 1}}&& \qw   & \ctrl{3} & \ctrl{1} & \gate{R_x(\pi)} & \qw  \\
% \lstick{\textnormal{qubit 2}}&& \qw   & \qw  & \ctrl{1} & \qw  & \qw  \\
% \lstick{\textnormal{qubit 3}}&& \qw   & \qw & \targ & \qw & \qw  \\
% \lstick{\textnormal{qubit 4}}&& \qw   & \targ  & \qw  & \qw  & \qw  
% }
% \]
% \caption{A simple example of a quantum circuit.}
% \label{fig_notation}
% \end{figure}

\begin{figure}[t]
\centering
\begin{subfigure}[b]{0.3\textwidth}
%\subfloat[template]{
\centering
\scalebox{0.8}{
$$
\Qcircuit @C=0.4em @R=0.4em {
&& & & 1& 2& 3&4 & \\
\lstick{\textnormal{qubit 1}} && & \qw  & \ctrl{1} & \ctrl{2} &\gate{Z}& \targ & \qw  \\
\lstick{\textnormal{qubit 2}} && & \qw  & \targ & \qw  & \qw  & \ctrl{-1}  & \qw  \\
\lstick{\textnormal{qubit 3}} && & \qw  & \qw  & \targ & \qw  & \ctrl{-2} & \qw  
}
$$
}
\subcaption{Circuit $C$}
\label{fig_canonical_form_circuit1}
\end{subfigure}
%\subfloat[circuit]{
\begin{subfigure}[b]{0.3\textwidth}
\centering
\scalebox{0.8}{
$$\Qcircuit @C=0.4em @R=0.4em {
& & 3& 2& 1&4 & \\
& \qw   &\gate{Z}  & \ctrl{2} & \ctrl{1} & \targ & \qw  \\
& \qw & \qw  & \qw  & \targ & \ctrl{-1}  & \qw  \\
& \qw & \qw  & \targ & \qw  & \ctrl{-2} & \qw  
}
$$
}
\subcaption{Commuted circuit $C$}
\label{fig_canonical_form_circuit2}
\end{subfigure}
\begin{subfigure}[b]{0.3\textwidth}
\centering
\scalebox{0.6}{
\begin{tikzpicture}
\node[circle,draw] at (0,0) (1) {1};
 \node[circle,draw] at (0,-\y) (2) {2};
 \node[circle,draw] at (0,-2*\y) (3) {3};
 \node[circle,draw] at (\x,-2*\y) (4) {4};
 \draw[->] (1) -- (4);
  \draw[->] (2) -- (4);
    \draw[->] (3) -- (4);
\end{tikzpicture}}
\subcaption{Canonical form of circuit $C$}
\label{fig_canonical_form_caconical_rep}
\end{subfigure}
\caption{An example of a quantum circuit (a), and equivalent quantum circuit obtained by interchanging commuting gates (b), and its canonical form (c), which is independent of the ordering of pairwise commuting gates chosen in the circuit picture.}
\label{fig_canonical_form}
\end{figure}

\paragraph{Canonical form of quantum circuits}
The circuit representation of a quantum computation is usually not unique because various gates may commute. For example, the two circuits represented in Figure~\ref{fig_canonical_form_circuit1} and Figure~\ref{fig_canonical_form_circuit2} implement the same computation.
For some applications, it is desirable to work with a representation of quantum circuits that does not change under pairwise commutation of gates. 
Such a representation, called the \emph{canonical form}, was introduced in~\cite{Rahman:2014:AQT:2711453.2629537} and we will use it extensively in this work.
The canonical form of a quantum circuit is a directed acyclic graph with the following two properties: firstly, vertices in the graph correspond to individual gates in the circuit. We can index all the gates in the circuit (in some fixed order) and label the graph vertices with these indices.
Secondly, the graph has an edge $i \to j$ from vertex $i$ to vertex $j$ if by repeatedly interchanging commuting gates in the circuit, one can bring gate $i$ immediately to the left of gate $j$, but gates $i$ and $j$ themselves do not commute. In other words, in the commuted circuit gate $j$ immediately follows gate $i$, but one cannot change the order of gates $i$ and $j$.
An example for the canonical form is given in \cref{fig_canonical_form_caconical_rep}.
Using the canonical form is not strictly necessary for our algorithm to be efficient, but it will simplify its description and lower the runtime for some subroutines. 
For any circuit its canonical form can be computed efficiently with an algorithm whose worst-case complexity scales quadratically in the number of gates (see Section~\ref{sec:canonical_form}).

%We also allow to use these definitions with a set of vertices as input, e.g., for a subset $S$ of vertices in $G$, \texttt{Succ}$(S,G):=\{$\texttt{Succ}$(v):v \in S \}$.
%Vertices that do not have any predecessors or successors are called initial or terminal vertices, respectively. The set of all root or terminal vertices in a graph $G$ is denoted by \textit{InitialV}$(G)$ and \textit{TerminalV}$(G)$, respectively. 
%We define the distance $d_G(v_1,v_2)$ between two vertices $v_1$ $v_2$ in a graph $G$ to be the longest path that exists from $v_1$ to $v_2$.

%%%%%%%%%%%%%%%

\input{algo_example}

%%%%%%%%%%%%%%%%%%%%%%%%%%%%%%%%%%%%%%%%%%%%%%%%%%%%%%%%%%%%%%%%%%
%%%%%%%%%%%%%%%%%%%%%%%%%%%%%%%%%%%%%%%%%%%%%%%%%%%%%%%%%%%%%%%%%%
%%%%%%%%%%%%%%%%%%%%%%%%%%%%%%%%%%%%%%%%%%%%%%%%%%%%%%%%%%%%%%%%%%
%\section{Efficient template matching}  \label{sec_mainRes}
%In this section we present an efficient algorithm that provably finds all maximal matches of a given template $T$ in a given circuit $C$. 
%We start with a discussion about difficulties that arise from the non-commutative nature of quantum circuits before presenting a pseudocode of the algorithm for efficient template matching where the details are shifted to Appendix~\ref{app_pseudoCode} to improve the readability of the manuscript.
%We refer the reader to Section~\ref{sec_example} which presents a pedagogical toy example of the algorithm that may be helpful to understand the overall idea of how the algorithm works. Understanding this example is eventually simpler than reading through the pseudocode of the algorithm.

%%%%%%%%%%%%%%%%%%%%%%%%%%%%%%%%%

\subsection{Heuristics} \label{sec:heuristics_intro}
The two most computationally expensive parts of our pattern matching algorithm are looping over all possible assignments of pattern qubits to circuit qubits, and building the tree of options that need to be considered by \texttt{BackwardMatch}.
We can speed up these parts of our algorithm using heuristics, at the cost of no longer being guaranteed to find all maximal matches. 
To speed up the assignment of pattern qubits to circuit qubits, note that the choice of starting gate already fixes some qubits, namely the ones on which the starting gate acts. 
If we assume that in addition to the starting gate, a small number of gates preceding or following the starting gate are also part of the maximal match, then we can use these gates to fix further qubits.
To reduce the number of options considered by \texttt{BackwardMatch}, we can repeatedly prune the tree of options and only keep the ones that have led to the highest number of matches until that point, reducing the number of considered options significantly.
Details on both heuristics are given in \cref{sec:heuristics_full}.

Our heuristics are parameterized by ``quality parameters'', meaning that one can adjust the trade-off between classical runtime and output quality depending on how much classical computation one is willing to perform to optimize the quantum circuit.
For small circuits on near-term devices, it may be worth running the algorithm without heuristics to be guaranteed optimal results, whereas for larger circuits in the future, these heuristics will increase the practical utility of our results. In contrast, known heuristic algorithms~\cite{Maslov_reversible_logic, Rahman:2014:AQT:2711453.2629537} provide less control over this trade-off, since they are not based on an exact algorithm like ours.

% There are two parts of the algorithm \texttt{PatternMatch} that are computationally particularly espensive: 1) Looping trough all qubit choices in the circuit on which we search for matching gates (see line~\ref{PatternMatch_loop_qubit_choice} in Algorithm~\ref{algo:PatternMatch}). 2) Building up the tree of possible matching scenarios in the algorithm \texttt{BackwardMatch} (Algorithm~\ref{algo:BackwardMatch}). We provide heuristics for both cases.

\subsection{Numerical experiments} \label{sec:numerics_intro}
In this section, we give an overview of numerical experiments demonstrating that our algorithm indeed runs significantly faster in practice than suggested by the worst-case complexity in \cref{eq_runntime} (see Figure~\ref{fig:numerics_scaling_C}). This makes it the first practical algorithm with a provable guarantee of finding all maximal matches. We demonstrate the use of our algorithm for optimizing quantum circuits in Table~\ref{tab:optimization_random}. 
Further numerical results regarding the scaling of our algorithm in the number of gates and qubits, and its performance for optimizing benchmark circuits can be found in Section~\ref{sec:circuit_opt}. 

\begin{figure}[ht]
\hspace{3mm}
\begin{subfigure}[b]{0.3\textwidth}
\centering
\scalebox{0.85}{
$$
    \Qcircuit @C=1.0em @R=0.0em @!R {
	 	\lstick{ {q}_{0} :  } & \gate{X} & \qw & \qw & \qw & \ctrl{1} & \qw & \qw\\
	 	\lstick{ {q}_{1} :  } & \qw & \targ & \ctrl{1} & \gate{X} & \ctrl{1} & \qw & \qw\\
	 	\lstick{ {q}_{2} :  } & \gate{X} & \ctrl{-1} & \targ & \qw & \targ & \qw & \qw\\
	 }
$$
}
\subcaption{Pattern}
\label{fig:pattern}
\end{subfigure}
\hfill
\begin{subfigure}[b]{0.6\textwidth}
\centering
\includegraphics[width=1\textwidth]{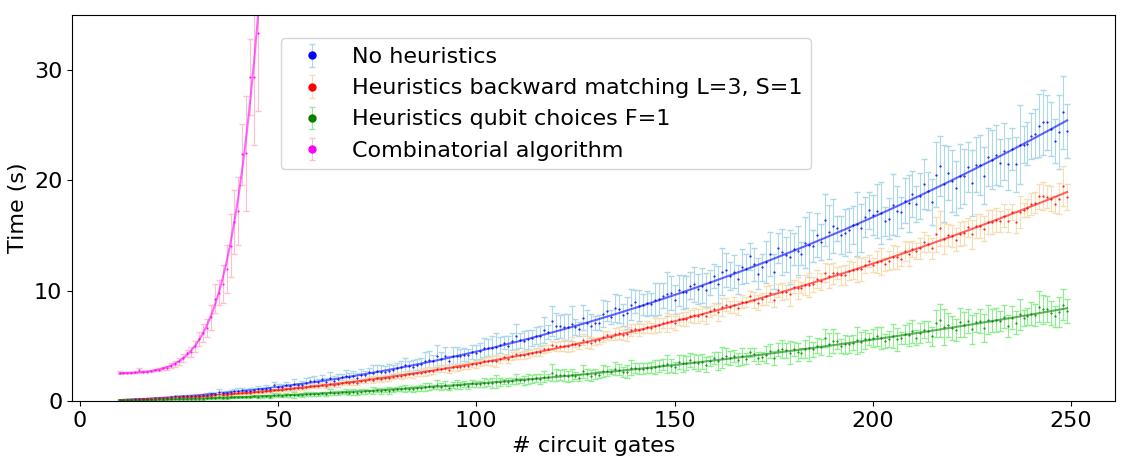}
\subcaption{Runtimes dependent on circuit size $|C|$}
\label{fig:comparison_2_mean}
\end{subfigure}
\caption{ The figure shows the runtimes of \texttt{PatternMatch} for finding maximal matches of the pattern~\ref{fig:pattern} ($n_T=3$, $|T|=6$) in randomly generated circuits with $|C| \in [10,250]$ gates on $n_C=6$ qubits consisting of X, \cnot{} and Toffoli gates. The gates as well as the qubits the gates act on are chosen uniformly at random. For each circuit size, 15 random quantum circuits are generated. We run the \texttt{PatternMatch} algorithm on these 15 circuits (taking as input the canonical form, though creating the canonical form itself scales roughly linearly and only adds approximately 10\% to the runtime -- see \cref{sec:numerics_canonical_form} for details).
We plot the mean and the standard deviation of the runtimes of these 15 circuits. 
The least-squares fit of a quadratic polynomial to the runtimes for the exact algorithm (without heuristics) is $3.84 \cdot 10
^{-4}(|C|^2 + 17.34|C| - 93.25)$. If we were to include a cubic term in the fit, the ratio of the cubic to the quadratic coefficient would 1/6000, so the scaling of the algorithm is indeed approximately quadratic in the circuit size. 
We also show the runtimes of the combinatorial algorithm described in \cref{app:combinatorial_algo} to demonstrate that while the worst-case complexities of both algorithms are similar, ours scales much better in practice. 
Further, we run our algorithm with the two heuristics described in \cref{sec:heuristics_intro}, where $L, S$, and $F$ are the ``quality parameters'' mentioned in \cref{sec:heuristics_intro}.
We find that the heuristics improve the runtime by roughly a constant factor, but do not affect the scaling in practice. 
Using the \texttt{BackwardMatch} heuristics causes the algorithm to miss 0.75 gates per match on average (for matches that are longer than half the pattern).
Using the qubit heuristics causes the algorithm to miss (on average) 39\% of matches that are longer than half the pattern.}
\label{fig:numerics_scaling_C}
\end{figure}

%\rmo{Polynomial fit 3.84E-04(x² + 17.34x - 93.25), Polynomial fit 2.82E-04(x² + 20.92x - 56.87), Polynomial fit 1.21E-04(x² + 29.3x - 30.51). ratio 1/6000 if we take x³, mean loss backward match 0.75 length/match, mean loss heuristics qubits 39\% of matches longer than half the templates}

Running the algorithm \texttt{PatternMatch} described in the previous section has a worst-case time complexity $\mathcal{O}\left(  |C|^{|T|+3} |T|^{|T| + 4}\, n_C^{n_T-1} \right)$, where $|C|$ and $n_C$ are the number of gates and qubits in the circuit, and $|T|$ and $n_T$ are the number of gates and qubits in the pattern (see~\cref{sec_complexity} for the detailed analysis). 
In practical applications, one usually tries to find a fairly small pattern of perhaps 5-10 gates on two or three qubits in a potentially very large circuit.
For such a small constant-sized pattern, we therefore have a worst-case runtime of $\mathcal{O}( \text{poly}(|C|,n_C))$, where the degree of the polynomial in $|C|$ may be around 10, i.e., very large for practical applications.
Fortunately, it turns out that our algorithm runs significantly faster in practice than suggested by the worst-case analysis (Figure~\ref{fig:numerics_scaling_C}), with an almost quadratic scaling in the circuit size $|C|$.
On a high level, this is due to the fact that for a typical quantum circuit, the forward part of the circuit (introduced in \cref{sec:algo_intro}) is relatively large.
Therefore, the majority of gates can be matched by the efficient subroutine \texttt{ForwardMatch}, and only a relatively small number of gates needs to be considered by the more inefficient subroutine \texttt{BackwardMatch}. Further, the tree of considered matching scenarios in \texttt{BackwardMatch} is usually much smaller than the one considered in the worst-case, since our algorithm \texttt{BackwardMatch} carefully selects the scenarios that could lead to maximal matches.  
We also find that the heuristics introduced in \cref{sec:heuristics_intro} lead to a roughly constant-factor improvement in runtime (see Figure~\ref{fig:numerics_scaling_C} for details).

We can use our pattern matching algorithm for circuit optimization by template matching as described in \cref{sec:prev_work}.
\cref{tab:optimization_random} gives an overview of the results: we find that for random circuits, we can achieve an approximately 30\% reduction in gate count, both for the exact algorithm and using the qubit heuristics.
We give a more detailed analysis of the use of our algorithm for circuit optimization in \cref{sec:numerics_full}.

% Additionally, we numerically analyse how much the heuristics introduced in \cref{sec:heuristics_intro} speed up our algorithm and how much they affect the output quality, i.e., how many maximal matches are missed.
% We find that for larger circuits, the heuristics to speed up the backward matching algorithm result in much shorter running times while still finding most maximal matches. The heuristics for the qubit assignment also significantly reduce the runtime, but do lead to more substantial losses in the output quality, i.e., more maximal matches are missed (Figure~\ref{fig:numerics_scaling_C}).
% Details on the numerical analysis are presented in \cref{sec:numerics_full}.

\begin{table}[t]
\centering
\scalebox{0.72}{
\begin{tabular}{c|cc|ccc|ccc}
\toprule
{} & \multicolumn{2}{c|}{Before template Matching} & \multicolumn{3}{c|}{After template matching,} & \multicolumn{3}{c}{After template matching,} \\
{} & \multicolumn{2}{c|}{} & \multicolumn{3}{c|}{without heuristics for qubit choice} & \multicolumn{3}{c}{heuristics for qubit choice, $L=1$} \\\midrule
 circuit name & \# gate in circuit $|C|$ & \# qubits $n_C$ & \# gate in circuit $|C|$ & gain (\%) & time (s) & \# gate in circuit $|C|$ & gain (\%) & time (s)\\ \toprule
$Random_{100}$ & 100 & 10 & 70 & 30.0 & 75.75 & 70 & 30 & 42.89 \\
$Random_{300}$ & 300 & 10 & 187 & 37.6 & 388.43 & 188 & 37.3 & 203.08\\
$Random_{500}$ & 500 & 10 & 351 & 29.80 & 1265.11 & 351 & 29.80 & 582.87\\
$Random_{700}$ & 700 & 10 & 494 & 29.4 & 2624.69 & 496 & 29.1 & 1146.27\\
\bottomrule
\end{tabular}
}
\caption{The table shows the results of applying template matching based on \texttt{PatternMatch} with 36 templates given in~\cite{Maslov_2007} (the number $n_T$ of qubits in the templates is ranging from 2 to 4 and the number $|T|$ of gates from 2 to 9) on random circuits with $n_C=10$ qubits and number of gates $|C|$ ranging from 100 to 700. All templates and circuits consist of X, \cnot{} and Toffoli gates. The middle and the right columns show how much the gate count can be reduced using template matching based on \texttt{PatternMatch} without and with heuristics for choosing the qubits with exploration length $L=1$ (defined in Section~\ref{sec:heuristics_full}), respectively. The ratio of $X$, \cnot{} and Toffoli gates that can be saved is roughly 70/25/5 for most of the circuits, i.e., the majority of savings stems from single-qubit $X$ gates.} 
\label{tab:optimization_random}
\end{table}

\subsection{Discussion}
Pattern matching is an important building block for many quantum circuit optimization techniques.
In this paper, we presented the first algorithm that provably finds all maximal matches of a pattern in a quantum or reversible classical circuit, while still scaling efficiently in practice. 
This provides a step towards overcoming the current trade-off between fast heuristic algorithms that might miss maximal matches, and exact algorithms that find all matches, but are highly impractical to run.
Overcoming the limitations of heuristic algorithm is relevant for practical applications: as shown in~\cite{mathias}, finding all maximal matches can help to reduce the gate count of reversible circuits further by up to 28\% compared to the heuristic algorithms~\cite{Maslov_reversible_logic}. 
More concretely, in this work we showed that our matching algorithm can also be used to further simplify circuits that were already optimized with several state-of-the-art techniques (see Table~\ref{tab:optimization_bench}).
 We also show  that our algorithm can be used for peephole optimization~\cite{prasad_data_2006, kliuchnikov_optimization_2013} in Section~\ref{sec:peephole}, where instead of searching for a fixed pattern $T$ in a circuit $C$, we can search for the longest connected parts in $C$ such that every gate only acts on a small subset of the qubits. We demonstrate that we can find sequences of gates (shown in Figure
~\ref{fig:circuit_chain}) that would have been missed by previous algorithms~\cite{prasad_data_2006, kliuchnikov_optimization_2013}.

We analyzed the worst-case complexity of the matching algorithm and showed that it scales polynomially in the circuit size for a fixed pattern size, although the degree of the polynomial does scale linearly with the pattern size.
We showed that in practical applications, the scaling derived in the theoretical worst-case analysis is far too pessimistic: our numerics suggest that in practice, our algorithm scales quadratically in the circuit size for a fixed number of qubits in the circuit and a fixed pattern size.
In addition to our exact algorithm, we also provide heuristics to further reduce the runtime of our algorithm. 
In contrast to previous heuristic algorithms ~\cite{Maslov_reversible_logic, Rahman:2014:AQT:2711453.2629537}, our heuristics are simplifications of our exact algorithm, where the degree of simplification is paramterized by a ``quality parameter''.
This parameter allows us to precisely control the trade-off between finding all maximal matches in the exact algorithm, and faster runtimes provided by the heuristics. 
We note that the heuristics we currently use are very simple and leave it as future work to find more sophisticated heuristics that might offer a better trade-off between runtime and output quality, especially for the choice of qubit assignment.

%In this paper we presented an algorithm for pattern matching in quantum circuits and hence, in particular, also for sub-circuit searches in reversible logic circuits.
%Our algorithm is the first that probably finds all maximal matches, while still featuring short runtimes in practice.

%We formally proved the correctness of our algorithm, analysed its worst-case complexity, and presented numerical results showing that for typical circuits, our algorithm scales considerably better than suggested by the worst-case complexity.
%Additionally, we introduced a family of heuristics for further speeding up our algorithm. Numerical simulations indicate that even using these heuristics, the algorithm still finds most maximal matches.

There are a number of other ways one could generalize our algorithm or further increase its practical utility.
Firstly, we only consider pairwise commutation relations in this work. In general, it could happen that in a circuit $C=(C_1,C_2,C_3)$, no gates commute pairwise, but the unitary corresponding to $(C_1,C_2)$ could commute with the unitary corresponding to $C_3$. Hence, one could bring the circuit $C$ into the form $(C_3,C_1,C_2)$, which could help matching in principle.
However, multiplying gates to check commutation relations is computationally expensive and we leave it for future work to find an efficient and practical implementation of this.

Secondly, to increase the utility of pattern matching algorithms such as ours for circuit optimization, it would be beneficial to find more useful templates, i.e., patterns that evaluate to the identity operation.
%In general, not too many templates are known. 
While universal decomposition schemes such as the one given in~\cite{iten_quantum_2016} can in principle produce whole parameter groups of templates, it remains to be investigated how useful these templates are in practice, and whether there are ways of systematically generating useful templates.

Thirdly, in addition to template matching, one could also apply our algorithm to circuit optimization schemes that use patterns that do not require a fixed a qubit assignment, since this would eliminate the need for looping over possible qubit assignments in our algorithm. 
For example, one could search for longest sequences of gates that only contain gates in the Clifford group, and then simplify the found circuits using re-synthesizing methods as the ones introduced in~\cite{kliuchnikov_optimization_2013}. 

Finally, extending our work in a graph-theoretic direction, it might also be possible to generalize our approach to the subgraph isomorphism problem under more general vertex rewriting rules, instead of the specific ones based on gate commutation that we use in our work.

\paragraph{Organisation.}
After introducing notation in Section~\ref{sec:notation}, we describe an algorithm that creates the canonical form of quantum circuits in Section~\ref{sec:canonical_form}. In Section~\ref{sec:possible_difficulties} we discuss the difficulties for efficient pattern matching due to the fact that quantum gates may commute. 
We then present and analyze the matching algorithm in Section~\ref{sec:temp_match}. The pseudocode is given in Section~\ref{sec:pseudo_code}. We discuss the correctness and complexity of the algorithm in Section~\ref{sec_correctness} and Section~\ref{sec_complexity}.
Details on the heuristics introduced in Section~\ref{sec:heuristics_intro} are discussed in Section~\ref{sec:heuristics_full}, and all our numerical results are given in Section~\ref{sec:numerics_full}. 

\paragraph{Acknowledgements.} We thank Roger Colbeck, Ali Javadi-Abhari, Vadym Kliuchnikov,  Dmitri Maslov,  Mathias Soeken for helpful discussions. We would like to thank Luca Mondada for working on a pattern matching algorithm as a semester thesis at ETH Z\"urich~\cite{mondada18}.
We thank Robin Kothari for pointing out the simple combinatorial algorithm (see~\cref{app:combinatorial_algo}) used as a benchmark in this work.
RI and DS acknowledges support from the Swiss National Science Foundation through SNSF project No.~200020-165843 and through the National Centre of Competence in Research \textit{Quantum Science and Technology} (QSIT). 
IBM, the IBM logo, and ibm.com are trademarks of International Business Machines Corp., registered in many jurisdictions worldwide. Other product and service names might be trademarks of IBM or other companies. The current list of IBM trademarks is available at \url{https://www.ibm.com/legal/copytrade}.

%% file: algo_example.tex
\paragraph{Overview of algorithm \texttt{PatternMatch}.} \label{sec:overview_TempMatch}

Given a pattern $T$ and a circuit $C$, we want to find all maximal matches of $T$ in $C$, i.e., all instances where a maximally-sized sub-circuit of $T$ equals a sub-circuit of $C$, taking into account qubit reordering and swapping of commuting gates.
To account for qubit reordering, we simply loop over all assignments of qubits in the pattern to qubits in the circuit.
For a fixed ordering of qubits, we start by picking a \emph{starting gate} from the pattern that matches a gate in the circuit. 
Again, we loop over all possible starting gates and describe the algorithm for a fixed choice of starting gate.

For intuition, it is helpful to consider the circuit reordered in the way that allows for the largest possible match.
Then, we can think of the circuit as consisting of three regions: an unmatched region to the left, a matched region in the middle, and an unmatched region to the right.
On a high level, the algorithm then needs to decide whether to keep a gate in the matched middle region, or to use the commutation relations to push it out to the unmatched side regions.
The difficulty of the problem stems from the fact that this decision depends on all other gates, so the algorithm cannot simply consider one gate at a time.

We can use the starting gate to partition the other gates in the pattern into two parts $\TLeft$ and $\TRight$.
Here, $\TLeft$ contains all the gates in the pattern that can be commuted to the left of the starting gate, and $\TRight$ contains the remaining gates (in the order in which they appear in the circuit, so the first gate in $\TRight$ is the starting gate). We write $T \simeq (\TLeft, \TRight)$ to stress that the two circuits $T$ and $(\TLeft, \TRight)$ implement the same computation, however, the gate order might be different.
Translated to the graph picture, $\TRight$ contains all successors of the starting gate, i.e., all vertices that can be reached from the starting gate by a directed path.
Similarly, we can partition the circuit $C \simeq (\CLeft, \CRight)$.

Our algorithm \texttt{PatternMatch} (\cref{algo:TempMatch}) proceeds in two steps: first, it executes a subroutine called \texttt{ForwardMatch} (\cref{algo:ForwardMatch}) that finds a match of the forward-part $\TRight$ of the pattern in the forward-part $\CRight$ of the circuit, followed by a subroutine \texttt{BackwardMatch} (\cref{algo:BackwardMatch}) that incorporates the gates in $\TLeft$.

The subroutine \texttt{ForwardMatch} essentially matches gates greedily, i.e., it greedily decides whether to include gates in the middle matched part part of the circuit, or push them to the right unmatched part.
Because of the special structure of the forward-part, one can show that a greedy strategy is indeed optimal (see \cref{lem_forwardMatch}).

The subroutine \texttt{BackwardMatch} needs to decide which of the gates from $\TLeft$ to include in the match.
This is a more difficult task because adding gates from $\TLeft$ to the match may require ``unmatching'' some of the matches found by \texttt{ForwardMatch}.
Hence, \texttt{BackwardMatch} needs to make a tradeoff between how many gates to include from the backward direction, and how many matches in the forward direction to destroy for this purpose.
To do so, \texttt{BackwardMatch} builds a tree of possible options and then finds the optimal one among them.

\paragraph{Example for algorithm \texttt{PatternMatch}.} \label{sec_example}
To demonstrate our pattern matching algorithm, we present an example, shown in ~\cref{fig_initial_setting}. 
We stress that the algorithm \texttt{PatternMatch} works for any gate set, not just the one used in the example.
The full algorithm is given in~\cref{sec:pseudo_code}.

% Tony: I think this is pretty standard and not worth describing in detail here, I also changed it before to just saying: loop over all starting configurations. We can change it back if you'd like of course :)
% Let us consider a fixed starting point for the matching, and for simplicity, let us assume that we want to start matching with the first gate of the template.\footnote{In the algorithm, we loop over all the gates in the template and consider them as starting gates. The given choice does not restrict the generality of the example, since you may consider the given template as a part of a larger template, where the first gate corresponds to a gate with index $i$ of the larger template.} In the circuit, the algorithm \texttt{TempMatch} loops over all possible starting points for a match of the first gate and over all choices and orders of five qubits out of the eight qubits of the given circuit.\footnote{This loop leads to the term $\frac{n_C!}{(n_C-n_T)!}$ in the complexity of the algorighm \texttt{TempMatch} (see Theorem~\ref{thm_Complexity}). }  Let us consider $C_8$ as the starting gate for a match (as shown in Figure~\ref{fig_circuit_start}) and the case where qubits $3-7$ have been chosen with the order given by the mapping of the qubit labels $(1,2,3,4,5)$ of the template to the qubit labels $(7,6,5,4,3)$ in the circuit. For simplicity, let us denote the template with relabeled qubit (which is denoted by $\tilde{T}$ in Algorithm~\ref{algo:TempMatch}) again by $T$ in the following.

As explained above, we need to loop over all possible assignments of qubits in the pattern to qubits in the circuit and all starting gates.
For concreteness, consider the case where the starting gate is chosen as shown in~\cref{fig_initial_setting}, and qubits $(1, 2, 3, 4, 5)$ in the pattern are assigned to qubits $(4, 5, 6, 7, 8)$ in the circuit.
Using \cref{algo:CreateCanonicalForm}, we create the canonical forms $G^{T}$ and $G^{C}$ of the pattern $T$ and the circuit $C$, respectively.
We denote by $G^T_i$ the $i$-th vertex in the graph $G^T$, which is identified with the gate $T_i$, and analogously for $G^C_i$.
Then, we can execute the two main subroutines, \texttt{ForwardMatch} and \texttt{BackwardMatch}.

% The algorithm \texttt{TempMatch} first runs Algorithm~\ref{algo:CreateCanonicalForm} twice to construct the canonical forms  $G^{T}$ and $G^{C}$ of the template $T$ and the circuit $C$, respectively. The resulting graphical representations of $G^{T}$ and $G^{C}$ are depicted in Figure~\ref{fig_canonical_ex}.

\begin{figure}[ht!]
\centering
\begin{subfigure}[b]{0.4\textwidth}
\tiny
\[
\Qcircuit @C=0.5em @R=0.35em{
&& 1 & 2 & \push{(3,\!4)} & \push{(5,\!6)} & 7 & 8 & 9 & 10 & 11 & 12 \\ 
\\
\lstick{\textnormal{qubit 1}} & \qw& \targ& \qw & \gate{Z} &\ctrl{1} & \qw  & \qw  & \targ & \qw  & \qw  & \ctrl{3} & \qw  \\
\lstick{\textnormal{qubit 2}}& \qw&\qw& \qw \qw&  \qw  & \targ  & \ctrl{1} & \gate{X} & \ctrl{-1} & \gate{X} & \ctrl{1} & \qw  & \qw  \\
\lstick{\textnormal{qubit 3}} & \qw & \qw  & \qw&\targ  & \qw &\targ   & \qw & \qw  & \qw  & \targ & \qw  & \qw  \\
\lstick{\textnormal{qubit 4}} & \qw  & \ctrl{-3} &   \qw & \qw & \ctrl{1} & \qw & \qw  & \qw  & \qw &\qw & \targ & \qw  \\
\lstick{\textnormal{qubit 5}} & \qw& \qw  & \gate{X} & \ctrl{-2} & \targ  & \qw& \qw& \qw& \qw& \qw& \qw& \qw
 \gategroup{3}{3}{6}{3}{.5em}{.}
 \gategroup{3}{5}{7}{11}{.8em}{--}
}
\] 
\subcaption{Pattern $T$}
\label{fig_Ex1}
\end{subfigure}
\qquad
\begin{subfigure}[b]{0.5\textwidth}
\tiny
\[
 \Qcircuit  @C=0.5em @R=0.4em
{
&& 1 & 2 & 3 & 4 & 5 & 6 & 7 & 8 & 9 & 10 & \push{(11,\!12)} & 13 & \push{(14,\!15)} & 16 & 17 & 18 & 19 & \push{(20,\!21)}  \\
\\
\lstick{\textnormal{qubit 1}} & \qw& \qw  & \qw & \qw & \qw & \qw   & \qw & \qw  & \qw  & \qw  & \ctrl{5} & \qw  & \qw & \qw  & \qw  & \qw  & \qw  & \qw   & \qw  & \qw \\
\lstick{\textnormal{qubit 2}} & \qw  & \qw & \qw & \qw & \qw & \qw  & \ctrl{3} & \qw  & \qw  & \qw  & \qw  & \qw & \qw & \targ & \ctrl{1} & \targ  & \qw  & \qw  & \qw  & \qw  \\
\lstick{\textnormal{qubit 3}}& \qw& \qw & \qw & \qw  & \qw  & \qw  & \qw  & \qw  & \qw  & \qw  & \qw  & \qw  & \qw  & \qw& \targ & \qw   & \qw  & \qw  & \qw  & \qw  \\
\lstick{\textnormal{qubit 4}}& \qw  & \qw & \qw & \qw & \qw & \qw  & \qw  & \targ & \ctrl{1}  & \qw  & \qw& \gate{Z} & \targ & \ctrl{-2} & \qw  & \ctrl{-2} & \ctrl{2} & \ctrl{3} & \gate{X} & \qw  \\
\lstick{\textnormal{qubit 5}}& \qw & \qw & \qw & \qw & \qw & \qw & \targ & \qw  & \targ & \ctrl{1}  & \qw  & \gate{X}  & \ctrl{-1} & \gate{X} & \qw  & \qw  & \qw  & \qw    & \ctrl{1} & \qw  \\
\lstick{\textnormal{qubit 6}}& \qw & \qw  &\targ &\qw&\targ&\qw & \qw  & \qw  & \qw  & \targ & \targ  & \qw  & \qw  & \qw  & \qw& \qw  & \targ & \qw  & \targ & \qw  \\
\lstick{\textnormal{qubit 7}}& \qw &\ctrl{1} &\qw&\ctrl{1} &\ctrl{-1} &\ctrl{1} & \qw  & \ctrl{-3} & \qw& \qw & \qw  & \qw & \qw & \qw& \qw  & \qw   & \qw  & \targ & \qw  & \qw  \\
\lstick{\textnormal{qubit 8}}& \qw  &\targ& \ctrl{-2} &\targ&\ctrl{-1} &\targ& \qw  &\qw  &\qw &\qw  & \qw  &\qw  & \qw & \qw  & \qw   & \qw  & \qw  &\qw & \qw& \qw   \\
\\
&&& 4& 6& &&& 1& 5 & 7 && \push{(3, \! 8)} & 9 & 10 &&&&& 11
        \gategroup{6}{9}{9}{9}{.5em}{.}
		\gategroup{8}{4}{10}{5}{.6em}{--}        
        \gategroup{6}{10}{8}{11}{.6em}{--}
        \gategroup{6}{13}{7}{14}{.6em}{--}
        \gategroup{7}{15}{7}{15}{.6em}{--}
        \gategroup{7}{20}{8}{20}{.6em}{--}
    } 
\]
\subcaption{Circuit $C$}
\label{fig_circuit_start} 
\end{subfigure}
\caption{Pattern $T$ and circuit $C$ used for our example. Indices assigned to the gates are shown above each circuit (with some gates drawn in parallel for readability). In our example, we consider the assignment of pattern qubits $(1, 2, 3, 4, 5)$ to circuit qubits $(4, 5, 6, 7, 8)$ and the starting gate highlighted by a dotted box. The maximal match that our algorithm finds is highlighted by dashed boxes. The indices below the circuit $C$ show which gates from the pattern are assigned to gates in the circuit.
%that should be maximally matched with a connected part \tm{connected part wasn't defined so far. do we need this?} of the circuit $C$. We start with matching the two marked gates. The numbers at the top denote the indices of the individual gates.}
%For simplicity, we represent some gates in parallel. However, we may think of them as being stored in an ordered list and each gate having its own index. We recall that target as well as control nodes of different \cnot{} gates (or Toffoli gates) commute with each other. In addition, $R_x$ and $R_z$ gates commute with target and control nodes of \cnot{} gates, respectively. 
}
\label{fig_initial_setting}
\end{figure}

\begin{figure}[ht!]
\centering
\begin{subfigure}[b]{0.4\textwidth}
\centering
\scalebox{0.5}{
\begin{tikzpicture}
\node[circle,draw,ultra thick] at (\x,0) (3) {3};
\node[circle,draw,fill=lightgray,ultra thick] at (0,-\y) (1) {1};
\node[circle,draw] at (0,-2*\y) (2) {2};
\node[circle,draw,ultra thick] at (\x,-\y) (5) {5};    
\node[circle,draw] at (\x,-2*\y) (4) {4};
\node[circle,draw,ultra thick] at (2*\x,-\y) (7) {7};    
\node[circle,draw] at (2*\x,-2*\y) (6) {6};
\node[circle,draw,ultra thick] at (3*\x,-\y) (8) {8}; 
\node[circle,draw,ultra thick] at (4*\x,-\y) (9) {9}; 
\node[circle,draw,ultra thick] at (5*\x,-\y) (10) {10}; 
\node[circle,draw,ultra thick] at (5*\x,-2*\y) (12) {12}; 
\node[circle,draw,ultra thick] at (6*\x,-\y) (11) {11}; 
\draw[->] (1) -- (3);
\draw[->] (1) -- (5);
\draw[->] (2) -- (4);
\draw[->] (3) -- (9.north west);
\draw[->] (5) -- (7);
\draw[->] (7) -- (8);
\draw[->] (8) -- (9);
\draw[->] (9) -- (10);
\draw[->] (10) -- (11);
\draw[->] (4) -- (6);
\draw[->] (6) -- (12);
\draw[->] (9) -- (12);
\end{tikzpicture}}
\subcaption{Canonical form of $T$}
\label{fig_canonical_T}
\end{subfigure}
\begin{subfigure}[b]{0.55\textwidth}
\centering
\scalebox{0.5}{
\begin{tikzpicture}
\node[circle,draw] at (0,-1) (1) {1};
\node[circle,draw] at (0,-2*\y) (6) {6};
\node[circle,draw,fill=lightgray,ultra thick] at (0,-3*\y) (7) {7};
\node[circle,draw] at (0,-4*\y) (10) {10};
\node[circle,draw] at (\x,-1) (2) {2};    
\node[circle,draw,ultra thick] at (\x,-3*\y) (8) {8};
\node[circle,draw,ultra thick] at (\x,-4*\y) (11) {11};
\node[circle,draw] at (2*\x,-1) (3) {3};
\node[circle,draw,ultra thick] at (2*\x,-3*\y) (9) {9};
\node[circle,draw] at (3*\x,-1) (4) {4};
\node[circle,draw,ultra thick] at (3*\x,-3*\y) (12) {12};
\node[circle,draw] at (4*\x,-1) (5) {5};
\node[circle,draw,ultra thick] at (4*\x,-3*\y) (13) {13};
\node[circle,draw,ultra thick] at (5*\x,-2*\y) (14) {14};
\node[circle,draw,ultra thick] at (5*\x,-3*\y) (15) {15};
\node[circle,draw,ultra thick] at (5*\x,-4*\y) (18) {18};
\node[circle,draw,ultra thick] at (5*\x,-1*\y) (19) {19};
\node[circle,draw,ultra thick] at (6*\x,-2*\y) (16) {16};
\node[circle,draw,ultra thick] at (6*\x,-3*\y) (21) {21};
\node[circle,draw,ultra thick] at (7*\x,-2*\y) (17) {17};
\node[circle,draw,ultra thick] at (8*\x,-2*\y) (20) {20};
\draw[->] (1) -- (2);
\draw[->] (2) -- (3);
\draw[->] (3) -- (4);
\draw[->] (4) -- (5);
\draw[->] (5) -- (19);
\draw[->] (6) -- (9.north west);
\draw[->] (7) -- (8);
\draw[->] (7) -- (11);
\draw[->] (8) -- (9);
\draw[->] (9) -- (12);
\draw[->] (12) -- (13);
\draw[->] (11) -- (13.south west);
\draw[->] (13) -- (14);
\draw[->] (13) -- (15);
\draw[->] (13) -- (18);
\draw[->] (13) -- (19);
\draw[->] (14) -- (16);
\draw[->] (16) -- (17);
\draw[->] (17) -- (20);
\draw[->] (15) -- (21);
\draw[->] (18.south east) -- (20);
\draw[->] (19) -- (20.north west);
\end{tikzpicture}}
\subcaption{Canonical form of $C$}
\label{fig_canonical_C} 
\end{subfigure}
\caption{Canonical forms of the pattern $T$ and the circuit $C$ from~\cref{fig_initial_setting}. The starting gates are marked in gray. The forward parts $\TRight$ and $\CRight$ are highlighted with bold outlines.}
\label{fig_canonical_ex}
\end{figure}

\textbf{Step 1: \texttt{ForwardMatch}.} The subroutine~\texttt{ForwardMatch} (\cref{algo:ForwardMatch}) greedily matches gates in the forward parts of the pattern and circuit (see \cref{lem_forwardMatch} for the proof that the greedy strategy is optimal).
The algorithm starts by initializing a list  $\mathit{MatchedVertexList}=(G_7^C)$ of matched vertices that have direct successors left to consider for matching.  The attribute $\mathit{SuccessorsToVisit}$ of $G_7^C$ is then set equal to $(G_8^C,G_{11}^C)$, since these are the direct successors of the starting vertex $G_7^C$. 

First, the vertex with lowest label in $G_7^C.\mathit{SuccessorsToVisit}$, i.e., the vertex $G^C_8$, is considered for matching (and removed from the list $\mathit{SuccessorsToVisit}$). 
The only gates in $T$ that $G_8^C$ could be matched with are the direct successors of the starting gate $G_1^T$, i.e., the gates $G_3^T$ and $G_5^T$. These gates in the pattern are found by the algorithm \texttt{FindForwardCandidates} (\cref{algo:FindForwardCandidates}).  Note that the algorithm \texttt{FindForwardCandidates} excludes direct successors that can not be moved next to the last matched gate in the pattern (since we require that the matched gates can be moved next to each other).
We find the match $T_5 = C_8$ and we set the attribute $\mathit{matchedWith}$ of vertex $G_8^C$ equal to $G^T_5$ (see the code after line~\ref{ForwardMatch_matching_case} in~\cref{algo:ForwardMatch}). 
Now we also need to consider all direct successors of $G^C_8$ as candidates for a match in the next round, so we set the attribute $\mathit{SuccessorsToVisit}$ of $G^C_8$ equal to $(G^C_{9})$ and add the vertex $G_8^C$ to the list of matched vertices $\mathit{MatchedVertexList}$. 
%We order the list $\mathit{MatchedVertexList}$ in increasing order according to the smallest label in the attribute list $\mathit{SuccessorsToVisit}$ of the vertices. In the considered case, this label is equal to 10 for $G^C_8$ and 12 for $G^C_7$, hence we end up with $\mathit{MatchedVertexList}=(G^C_8,G^C_7)$.

%In the second round of the while-loop in Algorithm~\ref{algo:ForwardMatch},  we consider the direct successors of the first vertex in $\mathit{MatchedVertexList}$ for further matches, 
In the second round of the while-loop in~\cref{algo:ForwardMatch}, we pick the the vertex with the smallest index in all the lists $\mathit{SuccessorsToVisit}$, in this case $G^C_{9}$ (and we remove $G^C_{9}$ from $G^C_{8}.\mathit{SuccessorsToVisit}$). We can match $G^C_8$ with $G^T_5$  and add the direct successors of $G^T_5$ as possible $\mathit{SuccessorsToVisit}$ for the next round.
Repeating this procedure, we subsequently match $C_{11} = T_3, \;C_{12} = T_8 $, and $C_{13} = T_9$.
% i.e., the vertex $G^T_7$. Again we find that $T_7=C_{10}$ and we can match the two gates, set $G^C_{9}.\mathit{SuccessorsToVisit}=G^C_{12}$ and add $G^C_{9}$ to $\mathit{MatchedVertexList}$.
%Similarly, in the next few cycles of the while loop the vertex $G^T_3$  is matched with $G^C_{11}$ and the vertex  $G^T_8$ is matched with $G^C_{12}$. 

Then, vertex $G^C_{14}$ (as a direct successor of vertex $G^C_{13}$) is considered for matching with direct successors of the vertex $G^T_9$. 
However, neither of the direct successors of $G^T_9$ (i.e., $G^T_{10}$ and $G^T_{12}$) can be matched with $G^C_{14}$ because the corresponding gates in the circuit differ. 
Therefore, we must \emph{right-block} vertex $G^C_{14}$.
Recall that we can think of the circuit as being composed of a left unmatched part, a middle matched part, and a right unmatched part.
In this picture, right-blocking $G^C_{14}$ means permanently moving the gate $C_{14}$ to the right unmatched part.
In doing so, we necessarily need to push all gates to the right of $C_{14}$ that do not commute with $C_{14}$ to the right unmatched part, too.
In the graph picture, this corresponds to also right-blocking all successors of $G^C_{14}$, i.e., $G^C_{16}, G^C_{17}, G^C_{20}$. These gates will not be considered for matching in the future and remain unmatched.

The remaining rounds of~\texttt{ForwardMatch} proceed analogously: $G^C_{15}$ and $G^T_{10}$ are matched, $G^C_{18}$ is right-blocked (and all its successors are already right-blocked), $G^C_{19}$ and $G^T_{12}$ are matched, and $G^C_{21}$ is matched with $G^T_{11}$. At this point, all vertices in $\CRight$ are either matched or right-blocked, as a result~\texttt{ForwardMatch} terminates.
The state after finishing~\texttt{ForwardMatch} is shown in \cref{fig_forward_ex}.

\begin{figure}[htb]
\centering
\begin{subfigure}[b]{0.4\textwidth}
\centering
\scalebox{0.5}{
\begin{tikzpicture}
\node[circle,draw,ultra thick,fill=lime] at (\x,0) (3) {3};
\node[circle,draw,fill=lightgray,ultra thick] at (0,-\y) (1) {1};
\node[circle,draw] at (0,-2*\y) (2) {2};
\node[circle,draw,ultra thick,fill=lime] at (\x,-\y) (5) {5};    
\node[circle,draw] at (\x,-2*\y) (4) {4};
\node[circle,draw,ultra thick,fill=lime] at (2*\x,-\y) (7) {7};    
\node[circle,draw] at (2*\x,-2*\y) (6) {6};
\node[circle,draw,ultra thick,fill=lime] at (3*\x,-\y) (8) {8}; 
\node[circle,draw,ultra thick,fill=lime] at (4*\x,-\y) (9) {9}; 
\node[circle,draw,ultra thick,fill=lime] at (5*\x,-\y) (10) {10}; 
\node[circle,draw,ultra thick,fill=lime] at (5*\x,-2*\y) (12) {12}; 
\node[circle,draw,ultra thick,fill=lime] at (6*\x,-\y) (11) {11}; 
\draw[->] (1) -- (3);
\draw[->] (1) -- (5);
\draw[->] (2) -- (4);
\draw[->] (3) -- (9.north west);
\draw[->] (5) -- (7);
\draw[->] (7) -- (8);
\draw[->] (8) -- (9);
\draw[->] (9) -- (10);
\draw[->] (10) -- (11);
\draw[->] (4) -- (6);
\draw[->] (6) -- (12);
\draw[->] (9) -- (12);
\end{tikzpicture}}
\end{subfigure}
\begin{subfigure}[b]{0.55\textwidth}
\centering
\scalebox{0.5}{
\begin{tikzpicture}
\node[circle,draw] at (0,-1) (1) {1};
\node[circle,draw] at (0,-2*\y) (6) {6};
\node[circle,draw,fill=lightgray,ultra thick] at (0,-3*\y) (7) {7};
\node[circle,draw] at (0,-4*\y) (10) {10};
\node[circle,draw] at (\x,-1) (2) {2};    
\node[circle,draw,ultra thick,fill=lime] at (\x,-3*\y) (8) {8};
\node[circle,draw,ultra thick,fill=lime] at (\x,-4*\y) (11) {11};
\node[circle,draw] at (2*\x,-1) (3) {3};
\node[circle,draw,ultra thick,fill=lime] at (2*\x,-3*\y) (9) {9};
\node[circle,draw] at (3*\x,-1) (4) {4};
\node[circle,draw,ultra thick,fill=lime] at (3*\x,-3*\y) (12) {12};
\node[circle,draw] at (4*\x,-1) (5) {5};
\node[circle,draw,ultra thick,fill=lime] at (4*\x,-3*\y) (13) {13};
\node[circle,draw,ultra thick,fill=black,text=white] at (5*\x,-2*\y) (14) {14};
\node[circle,draw,ultra thick,fill=lime] at (5*\x,-3*\y) (15) {15};
\node[circle,draw,ultra thick,fill=black,text=white] at (5*\x,-4*\y) (18) {18};
\node[circle,draw,ultra thick,fill=lime] at (5*\x,-1*\y) (19) {19};
\node[circle,draw,ultra thick,fill=black,text=white] at (6*\x,-2*\y) (16) {16};
\node[circle,draw,ultra thick,fill=lime] at (6*\x,-3*\y) (21) {21};
\node[circle,draw,ultra thick,fill=black,text=white] at (7*\x,-2*\y) (17) {17};
\node[circle,draw,ultra thick,fill=black,text=white] at (8*\x,-2*\y) (20) {20};
\draw[->] (1) -- (2);
\draw[->] (2) -- (3);
\draw[->] (3) -- (4);
\draw[->] (4) -- (5);
\draw[->] (5) -- (19);
\draw[->] (6) -- (9.north west);
\draw[->] (7) -- (8);
\draw[->] (7) -- (11);
\draw[->] (8) -- (9);
\draw[->] (9) -- (12);
\draw[->] (12) -- (13);
\draw[->] (11) -- (13.south west);
\draw[->] (13) -- (14);
\draw[->] (13) -- (15);
\draw[->] (13) -- (18);
\draw[->] (13) -- (19);
\draw[->] (14) -- (16);
\draw[->] (16) -- (17);
\draw[->] (17) -- (20);
\draw[->] (15) -- (21);
\draw[->] (18.south east) -- (20);
\draw[->] (19) -- (20.north west);
\end{tikzpicture}}
\end{subfigure}
\caption{State after~\texttt{ForwardMatch}. The starting vertices are marked in grey, the matched vertices in green, and the right-blocked vertices in black.}
\label{fig_forward_ex}
\end{figure}

\textbf{Step 2: \texttt{BackwardMatch}.} %It remains to find all maximal matches by also considering the vertices that are  are not successors of the starting vertex $G^C_7$. In general, it could help to move gates that disturb this matching process as far as possible to the right (which corresponds to blocking the corresponding vertex and its successors). However, this might block some gates that have already been matched in the forward matching process. Hence, there is a tradeoff between matching more gates on the left or on the right of the starting gate. This tradeoff makes the matching process costly in general, since one has to go through all the possibilities of moving disturbing gates to the left or to the right. However, based on the fact that we have already found a maximal match in the forward direction of the starting gate (by running \texttt{ForwardMatch}), we can show that this process is still efficient in the circuit size (but not in the template size in the worst-case) as discussed in the proof of Lemma~\ref{lem:ComplexityBackwardMatch} given in Section~\ref{sec_complexity}. 
%\footnote{Some simple heuristics may help here to speed up the matching process significantly, not missing many matches.} For the algorithm to stay efficient, it is crucial that we have already found a maximal match of $\TRight$ on the right hand side of the starting gate in the circuit. 
The subroutine \texttt{BackwardMatch} (\cref{algo:BackwardMatch}) decides which gates from $\TLeft$ to include in the match. As explained above, the main difficulty here is that including gates from $\TLeft$ may require unmatching gates in the forward direction, so the optimal trade-off between matches in the forward and backward direction has to be found. 

We store all vertices in $G^C$ that have not been matched and are not blocked in a sorted list $\mathit{GateIndices}=(10,6,5,4,3,2,1)$ and start by considering the largest one, i.e., vertex $G^C_{10}$, for matching. The indices 2,4,6 in the pattern are then found by \texttt{FindBackwardCandidates}, i.e., we consider all unmatched gates in the pattern as candidates for a next match (in backwards direction).\footnote{If we would for example match $T_2$, we would have to unmatch $T_{12}$, because $T_6$ could not be matched later on in the backward matching process. Our algorithm handles this case correctly, however, in the considered example, such a case does not occur.}
Since $T_2,T_4,T_6 \neq C_{10}$, we cannot match $G^C_{10}$, so the gate $C_{10}$ will be permanently moved to the left unmatched part of the circuit.\footnote{In fact, since $C_{11}$ commutes with all other gates in the circuit, we could equivalently move it to the right unmatched part, i.e., right-block it.} We call this \emph{left-blocking}. Analogously to right-blocking, we also need to left-block all predecessors of a left-blocked vertex (but in the case of $G^C_{10}$, there are none). Similarly, we left-block $G^C_6$.

Then, vertex $G^C_5$ is considered for matching. We find that it matches the candidate vertex $G^T_6$.
Now, we have two options, shown in \cref{fig_backward_ex}: we could either match $G^C_5$, i.e., include the gate $C_5$ in the middle matched part of the circuit. 
Alternatively, we could push the gate $C_5$ all the way to the right unmatched part of the circuit, i.e., right-block  $G^C_5$.\footnote{In principle, we could also left-block $G^C_5$, and hence also block all of its predecessors. However, we prove in~\cref{lem_backwardsMatch} that this case can be ignored as long as matching $G^T_6$ does not destroy other matches (see line~\ref{backwardsMatch_check_if_matches_destroyed} in~\cref{algo:BackwardMatch})}
The second option has the disadvantage that we need to block $G^C_5$ and all its successors, even ones that have already been matched during \texttt{ForwardMatch}.
However, it might enable us to match more of the predecessors of $G^C_5$.
At this stage, we cannot yet decide which option will result in more matches overall, and we keep track of both in a tree of options (called \textit{MatchingScenarios} in~\cref{algo:BackwardMatch}). 

\begin{figure}[htb]
\centering
\begin{subfigure}[b]{0.48\textwidth}
\centering
\scalebox{0.5}{
\begin{tikzpicture}
\node[circle,draw] at (0,-1) (1) {1};
\node[circle,draw,fill=blue,text=white] at (0,-2*\y) (6) {6};
\node[circle,draw,fill=lightgray,ultra thick] at (0,-3*\y) (7) {7};
\node[circle,draw,fill=blue,text=white] at (0,-4*\y) (10) {10};
\node[circle,draw] at (\x,-1) (2) {2};    
\node[circle,draw,ultra thick,fill=lime] at (\x,-3*\y) (8) {8};
\node[circle,draw,ultra thick,fill=lime] at (\x,-4*\y) (11) {11};
\node[circle,draw] at (2*\x,-1) (3) {3};
\node[circle,draw,ultra thick,fill=lime] at (2*\x,-3*\y) (9) {9};
\node[circle,draw] at (3*\x,-1) (4) {4};
\node[circle,draw,ultra thick,fill=lime] at (3*\x,-3*\y) (12) {12};
\node[circle,draw,fill=orange] at (4*\x,-1) (5) {5};
\node[circle,draw,ultra thick,fill=lime] at (4*\x,-3*\y) (13) {13};
\node[circle,draw,ultra thick,fill=black,text=white] at (5*\x,-2*\y) (14) {14};
\node[circle,draw,ultra thick,fill=lime] at (5*\x,-3*\y) (15) {15};
\node[circle,draw,ultra thick,fill=black,text=white] at (5*\x,-4*\y) (18) {18};
\node[circle,draw,ultra thick,fill=lime] at (5*\x,-1*\y) (19) {19};
\node[circle,draw,ultra thick,fill=black,text=white] at (6*\x,-2*\y) (16) {16};
\node[circle,draw,ultra thick,fill=lime] at (6*\x,-3*\y) (21) {21};
\node[circle,draw,ultra thick,fill=black,text=white] at (7*\x,-2*\y) (17) {17};
\node[circle,draw,ultra thick,fill=black,text=white] at (8*\x,-2*\y) (20) {20};
\draw[->] (1) -- (2);
\draw[->] (2) -- (3);
\draw[->] (3) -- (4);
\draw[->] (4) -- (5);
\draw[->] (5) -- (19);
\draw[->] (6) -- (9.north west);
\draw[->] (7) -- (8);
\draw[->] (7) -- (11);
\draw[->] (8) -- (9);
\draw[->] (9) -- (12);
\draw[->] (12) -- (13);
\draw[->] (11) -- (13.south west);
\draw[->] (13) -- (14);
\draw[->] (13) -- (15);
\draw[->] (13) -- (18);
\draw[->] (13) -- (19);
\draw[->] (14) -- (16);
\draw[->] (16) -- (17);
\draw[->] (17) -- (20);
\draw[->] (15) -- (21);
\draw[->] (18.south east) -- (20);
\draw[->] (19) -- (20.north west);
\end{tikzpicture}}
\subcaption{Option 1: match vertex $G^C_5$}
\end{subfigure}
\begin{subfigure}[b]{0.48\textwidth}
\centering
\scalebox{0.5}{
\begin{tikzpicture}
\node[circle,draw] at (0,-1) (1) {1};
\node[circle,draw,fill=blue,text=white] at (0,-2*\y) (6) {6};
\node[circle,draw,fill=lightgray,ultra thick] at (0,-3*\y) (7) {7};
\node[circle,draw,fill=blue,text=white] at (0,-4*\y) (10) {10};
\node[circle,draw] at (\x,-1) (2) {2};    
\node[circle,draw,ultra thick,fill=lime] at (\x,-3*\y) (8) {8};
\node[circle,draw,ultra thick,fill=lime] at (\x,-4*\y) (11) {11};
\node[circle,draw] at (2*\x,-1) (3) {3};
\node[circle,draw,ultra thick,fill=lime] at (2*\x,-3*\y) (9) {9};
\node[circle,draw] at (3*\x,-1) (4) {4};
\node[circle,draw,ultra thick,fill=lime] at (3*\x,-3*\y) (12) {12};
\node[circle,draw,fill=black,text=white] at (4*\x,-1) (5) {5};
\node[circle,draw,ultra thick,fill=lime] at (4*\x,-3*\y) (13) {13};
\node[circle,draw,ultra thick,fill=black,text=white] at (5*\x,-2*\y) (14) {14};
\node[circle,draw,ultra thick,fill=lime] at (5*\x,-3*\y) (15) {15};
\node[circle,draw,ultra thick,fill=black,text=white] at (5*\x,-4*\y) (18) {18};
\node[circle,draw,ultra thick,fill=black,text=white] at (5*\x,-1*\y) (19) {19};
\node[circle,draw,ultra thick,fill=black,text=white] at (6*\x,-2*\y) (16) {16};
\node[circle,draw,ultra thick,fill=lime] at (6*\x,-3*\y) (21) {21};
\node[circle,draw,ultra thick,fill=black,text=white] at (7*\x,-2*\y) (17) {17};
\node[circle,draw,ultra thick,fill=black,text=white] at (8*\x,-2*\y) (20) {20};
\draw[->] (1) -- (2);
\draw[->] (2) -- (3);
\draw[->] (3) -- (4);
\draw[->] (4) -- (5);
\draw[->] (5) -- (19);
\draw[->] (6) -- (9.north west);
\draw[->] (7) -- (8);
\draw[->] (7) -- (11);
\draw[->] (8) -- (9);
\draw[->] (9) -- (12);
\draw[->] (12) -- (13);
\draw[->] (11) -- (13.south west);
\draw[->] (13) -- (14);
\draw[->] (13) -- (15);
\draw[->] (13) -- (18);
\draw[->] (13) -- (19);
\draw[->] (14) -- (16);
\draw[->] (16) -- (17);
\draw[->] (17) -- (20);
\draw[->] (15) -- (21);
\draw[->] (18.south east) -- (20);
\draw[->] (19) -- (20.north west);
\end{tikzpicture}}
\subcaption{Option 2: right-block vertex $G^C_5$}
\label{fig_E11}
\end{subfigure}
\caption{Grey denotes the starting vertex, green marks vertices matched in \texttt{ForwardMatch}, orange marks vertices matched in \texttt{BackwardMatch} (these will never be blocked again), black marks right-blocked vertices, and blue marks left-blocked vertices. We have two options for vertex $5$ that could lead to a maximal match:
we can match vertex $G^C_5$ with vertex $G^T_5$ (option 1); or we right-block $G^C_5$ and all its successors, including the previously matched vertex $G^C_{19}$.}
\label{fig_backward_ex}
\end{figure}

Considering one option at a time, if we matched the vertex $G^C_5$, one finds that no further gates can be matched (without right-blocking some of the predecessors of $G^C_5$, which would also block $G^C_5$; this we do not have to consider, since it would lead to the same scenario as the non-matching case already added to the tree \textit{MatchingScenarios}). Hence, we could match 9 gates in total in this scenario.

On the other hand, if we do not match vertex $G^C_5$, we proceed as follows: no match can be found for $G^C_4$. We could either left- or right-block it. Since all successors of $G^C_4$ are already blocked, but its predecessors are not, right-blocking is the better option. We then see that $G^C_3$ can be matched with $G^T_6$, and $G^C_2$ with $G^T_4$. No match can be found for $G^C_1$, and since it has no predecessors, but matched successors, left-blocking it is the best option. This way, we can match 10 gates in total in this matching scenario, which turns out to be the maximal match (over all qubit assignments and starting gates). The maximal match in the canonical form is shown in ~\cref{fig_full_match}, and in ~\cref{fig_circuit_start} for the circuit picture.

\begin{figure}[htb]
\centering
\begin{subfigure}[b]{0.4\textwidth}
\centering
\scalebox{0.5}{
\begin{tikzpicture}
\node[circle,draw,ultra thick,fill=lime] at (\x,0) (3) {3};
\node[circle,draw,fill=lightgray,ultra thick] at (0,-\y) (1) {1};
\node[circle,draw] at (0,-2*\y) (2) {2};
\node[circle,draw,ultra thick,fill=lime] at (\x,-\y) (5) {5};    
\node[circle,draw,fill=orange] at (\x,-2*\y) (4) {4};
\node[circle,draw,ultra thick,fill=lime] at (2*\x,-\y) (7) {7};    
\node[circle,draw,fill=orange] at (2*\x,-2*\y) (6) {6};
\node[circle,draw,ultra thick,fill=lime] at (3*\x,-\y) (8) {8}; 
\node[circle,draw,ultra thick,fill=lime] at (4*\x,-\y) (9) {9}; 
\node[circle,draw,ultra thick,fill=lime] at (5*\x,-\y) (10) {10}; 
\node[circle,draw,ultra thick] at (5*\x,-2*\y) (12) {12}; 
\node[circle,draw,ultra thick,fill=lime] at (6*\x,-\y) (11) {11}; 
\draw[->] (1) -- (3);
\draw[->] (1) -- (5);
\draw[->] (2) -- (4);
\draw[->] (3) -- (9.north west);
\draw[->] (5) -- (7);
\draw[->] (7) -- (8);
\draw[->] (8) -- (9);
\draw[->] (9) -- (10);
\draw[->] (10) -- (11);
\draw[->] (4) -- (6);
\draw[->] (6) -- (12);
\draw[->] (9) -- (12);
\end{tikzpicture}}
\label{fig_backward_T}
\end{subfigure}
\begin{subfigure}[b]{0.55\textwidth}
\centering
\scalebox{0.5}{
\begin{tikzpicture}
\node[circle,draw,fill=blue,text=white] at (0,-1) (1) {1};
\node[circle,draw,fill=blue,text=white] at (0,-2*\y) (6) {6};
\node[circle,draw,fill=lightgray,ultra thick] at (0,-3*\y) (7) {7};
\node[circle,draw,fill=blue,text=white] at (0,-4*\y) (10) {10};
\node[circle,draw,fill=orange] at (\x,-1) (2) {2};    
\node[circle,draw,ultra thick,fill=lime] at (\x,-3*\y) (8) {8};
\node[circle,draw,ultra thick,fill=lime] at (\x,-4*\y) (11) {11};
\node[circle,draw,fill=orange] at (2*\x,-1) (3) {3};
\node[circle,draw,ultra thick,fill=lime] at (2*\x,-3*\y) (9) {9};
\node[circle,draw,fill=black,text=white] at (3*\x,-1) (4) {4};
\node[circle,draw,ultra thick,fill=lime] at (3*\x,-3*\y) (12) {12};
\node[circle,draw,fill=black,text=white] at (4*\x,-1) (5) {5};
\node[circle,draw,ultra thick,fill=lime] at (4*\x,-3*\y) (13) {13};
\node[circle,draw,ultra thick,fill=black,text=white] at (5*\x,-2*\y) (14) {14};
\node[circle,draw,ultra thick,fill=lime] at (5*\x,-3*\y) (15) {15};
\node[circle,draw,ultra thick,fill=black,text=white] at (5*\x,-4*\y) (18) {18};
\node[circle,draw,ultra thick,fill=black,text=white] at (5*\x,-1*\y) (19) {19};
\node[circle,draw,ultra thick,fill=black,text=white] at (6*\x,-2*\y) (16) {16};
\node[circle,draw,ultra thick,fill=lime] at (6*\x,-3*\y) (21) {21};
\node[circle,draw,ultra thick,fill=black,text=white] at (7*\x,-2*\y) (17) {17};
\node[circle,draw,ultra thick,fill=black,text=white] at (8*\x,-2*\y) (20) {20};
\draw[->] (1) -- (2);
\draw[->] (2) -- (3);
\draw[->] (3) -- (4);
\draw[->] (4) -- (5);
\draw[->] (5) -- (19);
\draw[->] (6) -- (9.north west);
\draw[->] (7) -- (8);
\draw[->] (7) -- (11);
\draw[->] (8) -- (9);
\draw[->] (9) -- (12);
\draw[->] (12) -- (13);
\draw[->] (11) -- (13.south west);
\draw[->] (13) -- (14);
\draw[->] (13) -- (15);
\draw[->] (13) -- (18);
\draw[->] (13) -- (19);
\draw[->] (14) -- (16);
\draw[->] (16) -- (17);
\draw[->] (17) -- (20);
\draw[->] (15) -- (21);
\draw[->] (18.south east) -- (20);
\draw[->] (19) -- (20.north west);
\end{tikzpicture}}
\label{fig_backward_C} 
\end{subfigure}
\caption{Maximal match found by \texttt{PatternMatch} with colours as in \cref{fig_backward_ex}.}
\label{fig_full_match}
\end{figure}